\documentclass[11pt,letterpaper]{article}
\usepackage{fullpage}

\usepackage{amssymb,amsmath,amsthm}
\usepackage{latexsym,mathpartir,alltt}

\usepackage{graphicx}
\usepackage{zsym}

\usepackage{relsize}
\usepackage{longtable}

\usepackage[pdftex,colorlinks]{hyperref}
\hypersetup{
    bookmarks=true,         
    unicode=false,          
    pdftoolbar=true,        
    pdfmenubar=true,        
    pdffitwindow=true,      
    pdftitle={},    
    pdfauthor={},     
    pdfsubject={},   
    pdfnewwindow=true,      
    pdfkeywords={}, 
    colorlinks=true,       
    linkcolor=blue,          
    citecolor=blue,        
    filecolor=magenta,      
    urlcolor=blue           
}

\usepackage[numbers,sort&compress]{natbib}

\usepackage{hypernat}

\theoremstyle{plain}
\newtheorem{thm}{Theorem}[section]

\newtheorem{lem}[thm]{Lemma}

\theoremstyle{definition}
\newtheorem{eg}[thm]{Example}
\newtheorem{defn}[thm]{Definition}

\mathchardef\mhyphen="2D

\newcommand{\ttime}{\tau}

\newcommand{\until}{\mathop{\mathrm{U}}}
\newcommand{\since}{\mathop{\mathrm{S}}}
\newcommand{\here}{{\downarrow}}

\newcommand{\diap}{\Zdiamond}

\newcommand{\boxp}{\Zsquare}

\newcommand{\diam}{\Zdiamondminus}

\newcommand{\boxm}{\Zsquareminus}
\newcommand{\globally}{\mathop{\text{\bf G}}}
\newcommand{\eventually}{\mathop{\text{\bf F}}}

\newcommand{\imp}{\mathrel{\supset}}
\newcommand{\conj}{\mathrel{\wedge}}
\newcommand{\disj}{\mathrel{\vee}}

\newcommand{\symt}{{\tt tt}}
\newcommand{\symf}{{\tt ff}}
\newcommand{\symu}{{\tt uu}}

\newcommand{\trans}[2]{(#2)^{#1}}

\newcommand{\pred}[1]{{\tt #1}}

\newcommand{\fv}{\pred{fv}}
\newcommand{\dom}{\pred{dom}}

\newcommand{\case}{\noindent \textbf{Case.}}
\newcommand{\subcase}{\noindent \textbf{Subcase.}}

\newcommand{\cut}[1]{}

\newcommand{\dual}[1]{\overline{#1}}

\newcommand{\attr}[1]{\text{\textit{#1}}}
\newcommand{\purpose}[1]{\text{\textit{#1}}}

\newcommand{\reduce}{{\sf reduce}}
\newcommand{\audit}[3]{\reduce(#1, #2) = #3}
\newcommand{\structure}{{\cal L}}
\newcommand{\domain}{\mathrm{D}}

\newcommand{\sat}{{\tt sat}}
\newcommand{\atoms}{{\tt atoms}}

\newcommand{\lift}[1]{\widehat{#1}}

\newcommand{\iterate}[1]{\mathrel{\xrightarrow{#1}}}
\newcommand{\rewrite}{\rightarrow}

\begin{document}

\title{A Logical Method for Policy Enforcement over Evolving Audit
  Logs\thanks{This work was partially supported by the U.S. Army
    Research Office contract "Perpetually Available and Secure
    Information Systems" (DAAD19-02-1-0389) to Carnegie Mellon CyLab,
    the NSF Science and Technology Center TRUST, the NSF CyberTrust
    grant “Privacy, Compliance and Information Risk in Complex
    Organizational Processes,” the AFOSR MURI “Collaborative Policies
    and Assured Information Sharing”, and HHS Grant no. HHS
    90TR0003/01. The views and conclusions contained in this document
    are those of the authors and should not be interpreted as
    representing the official policies, either expressed or implied, of
    any sponsoring institution, the U.S. government or any other
    entity.}  }

\author{Deepak Garg \and Limin Jia \and Anupam Datta}

\date{Technical Report CMU-CyLab-11-002 \\ ~\\ Revision of May 06, 2011}

\maketitle

\begin{abstract}
We present an iterative algorithm for enforcing policies represented
in a first-order logic, which can, in particular, express all
transmission-related clauses in the HIPAA Privacy Rule. The logic has
three features that raise challenges for enforcement --- uninterpreted
predicates (used to model subjective concepts in privacy policies),
real-time temporal properties, and quantification over infinite
domains (such as the set of messages containing personal
information). The algorithm operates over audit logs that are
inherently incomplete and evolve over time.  In each iteration, the
algorithm provably checks as much of the policy as possible over the
current log and outputs a residual policy that can only be checked
when the log is extended with additional information. We prove
correctness and termination properties of the algorithm. While these
results are developed in a general form, accounting for many different
sources of incompleteness in audit logs, we also prove that for the
special case of logs that maintain a complete record of all relevant
actions, the algorithm effectively enforces all safety and co-safety
properties. The algorithm can significantly help automate enforcement
of policies derived from the HIPAA Privacy Rule.
\end{abstract}

\maketitle

\section{Introduction}
\label{sec:intro}


Organizations, such as hospitals, banks, and universities, that
collect, use, and share personal information have to ensure that they
do so in a manner that respects the privacy of the information
subjects. In fact, designing effective processes to audit transmission
and access logs to ensure compliance with privacy regulations, such as
the Health Insurance Portability and Accountability Act
(HIPAA)~\cite{HIPAA}, has become one of the greatest challenges facing
organizations today (see, for example, a recent survey from Deloitte
and the Ponemon Institute~\cite{DP2007}).  State-of-the-art commercial
tools such as the FairWarning~\cite{fairwarning}  allow auditors to mine
access and transmission logs and flag potential violations of policy,
but do not help decide which flagged items are actual violations, even
though privacy legislation often lays down objective criteria to make
such decisions. We address this challenge by developing a novel,
logic-based method for computer-assisted enforcement of policies.
This method can be used to enforce a rich class of privacy and
security policies that include, in particular, real privacy
regulations like HIPAA.




\paragraph{Policy Specification} 
The first challenge for policy enforcement is formal specification of
real policies. This challenge was addressed in our prior work on
PrivacyLFP~\cite{DeYoung+10:privacylfp:wpes}, an expressive
first-order temporal logic, in which we represented formally all
transmission-related clauses of the HIPAA and GLBA Privacy
Laws. PrivacyLFP is more expressive than prior logics considered for
expressing policies, including propositional temporal
logics~\cite{Barth+06:paci:fap,Giblin+05:REALM} and first-order metric
temporal logic~\cite{Basin+2010:cav}.


Building on the prior work on specification of privacy laws in
PrivacyLFP, this paper presents an algorithm for enforcing policies
represented in the logic, through iterative analysis of audit logs,
which we assume are collected independently and provided to us. The
policy enforcement algorithm and the formulation and proof of its
properties are the main contribution of this paper.

Three concepts in privacy legislation (and PrivacyLFP) make mechanical
enforcement particularly difficult; we discuss these concepts briefly.
First, PrivacyLFP includes \emph{uninterpreted or subjective
  predicates} to model subjective parts of privacy laws. For example,
HIPAA allows transmission of protected health information about an
individual from a hospital to a law enforcement agency if the hospital
believes that the death of the individual was suspicious. Such beliefs
are represented using uninterpreted predicates because the truth value
of these predicates cannot, in general, be determined mechanically.



Second, PrivacyLFP allows first-order quantification over infinite
domains (e.g., the set of messages or the set of time points). For
example, many HIPAA clauses are of the form \linebreak[6]$\forall p_1,
p_2, m.(\pred{send}(p_1, p_2, m) \imp \phi)$ where $p_1$ and $p_2$ are
principals and $m$ is a message. Note that this formula quantifies
over the infinite set of messages, so if an enforcement algorithm were
to blindly instantiate the quantifiers with all possible values in the
domain, then it will not terminate.  However, only a finite number of
messages are relevant in determining the truth value of this
formula. This is because the number of messages transmitted from a
hospital is finite and hence the predicate $\pred{send}(p_1, p_2, m)$
is true for only a finite number of substitutions for the variable $m$
(and similarly for $p_1$ and $p_2$).  To ensure that the number of
relevant substitutions for every quantified variable is finite, we use
the idea of \emph{mode checking} from logic
programming~\cite{apt94:modes}, and restrict the syntax of quantifiers
in PrivacyLFP slightly. The finite substitution property for
quantified variables over infinite domains is defined in
Section~\ref{sec:enforcement}, and ensures termination of our policy
enforcement algorithm. The restriction on quantification does not
significantly limit representation of HIPAA clauses, a claim we
justify in Section~\ref{sec:hipaa-case}.

Third, the representation of one transmission-related clause --
Section 6802(c) -- of the GLBA Privacy Law forces PrivacyLFP to
include fixpoint operators. In this paper, we do not consider
fixpoints because the representation of most privacy legislation
including all of HIPAA does not require fixpoints. We note that
including the least fixpoint operator in our algorithm may not be
difficult, but supporting the greatest fixpoint may require a
substantial effort.

\paragraph{Audit logs}
Another significant challenge in mechanical enforcement of privacy
policies is that the logs maintained by organizations may be
incomplete, i.e., they may not contain enough information to decide
whether or not the policy has been violated. For instance, in the
absence of human input, a machine may not be able to decide whether
any instance of a predicate that refers to subjective beliefs is true
or not. Similarly, we may not be able to predict whether a predicate
holds in the future or not. As an important contribution, we observe
that such possibly incomplete logs can be abstractly represented as
three-valued, \emph{partial structures} that map each atomic formula
to either true, false, or
unknown~\cite{Bruns:2000:GMC,Godefroid:2005:MCV}. We define the
semantics of our logic over such structures. Further, by designing our
enforcement algorithm to work with partial structures in general, we
provide a uniform account of policy enforcement with different forms
of log incompleteness.

We explicitly discuss in Section~\ref{sec:instance:pc} a special case
of partial structures that are complete up to a point of time.  This
instance corresponds to the standard model of traces used in prior
work on enforcement of temporal privacy
properties~\cite{Basin+2010:cav}.  We show that on such structures,
our algorithm yields a method to find violations of safety
properties~\cite{alpern87:safety} and satisfactions of co-safety
properties~\cite{bauer10:rm} at the earliest possible time, as may be
expected.

A second important observation is that, in practice, structures evolve
over time by gathering more information. We formalize this growth as a
natural order, $\structure_1 \geq \structure_2$ (structure
$\structure_1$ \emph{extends} structure $\structure_2$), meaning that
$\structure_1$ has more information than $\structure_2$.  We present a
general definition of extension of partial structures, which
encompasses, in particular, notions of temporal (actions are added to
the end of a trace) and spatial (distributed logs are merged)
extensions.

\paragraph{Policy Enforcement}
As our central contribution, we propose an iterative process for
privacy policy enforcement.  At each iteration, our algorithm takes as
inputs a structure $\structure$ abstracting the then-current audit log
and a policy specification $\varphi$, verifies parts of the policy
that depend solely on the given structure, and outputs a residual
policy $\varphi'$ that contains all the conditions that need to be
verified when more information becomes available. We write
$\audit{\structure}{\varphi}{\varphi'}$ to denote one iteration of our
reduction algorithm.  The residual policy $\varphi'$ is checked on
extensions of $\structure$.

Our reduction algorithm has several desirable properties that we prove
formally. First, the algorithm always \emph{terminates}. As noted
earlier, the finite substitution property for variables quantified
over infinite domains is crucial for termination. Second, it is
\emph{correct}: given a structure $\structure$ and a policy $\varphi$,
any extension of $\structure$ satisfies the policy $\varphi$ if and
only if it satisfies the residual formula $\varphi'$. Third, it is
\emph{minimal}: the residual formula only contains atoms whose truth
value cannot be determined from the structure.

Our algorithm has been designed for after-the-fact (a-posteriori)
audit, not runtime verification. However, as shown in
Section~\ref{sec:instance:pc}, for the specific case of policies that
do not contain any subjective predicates or future obligations, the
algorithm may be executed at each privacy-relevant event to act as a
runtime monitor, if all relevant past system logs can be provided to
it.

\paragraph{Application to HIPAA}
Our technical results have important implications for enforcing
practical privacy policies, in particular, the HIPAA Privacy Rule.  As
discussed in Section~\ref{sec:hipaa-case}, not only can our algorithm
be used to automatically instantiate all quantifiers in all 84
transmission-related clauses of HIPAA, but it can also automatically
discharge the large percentage of non-subjective atoms in instantiated
clauses. For example, we estimate that in 17 of the 84 clauses, all
atoms can be discharged automatically, and in 24 other clauses, at
least 80\% of the atoms can be discharged automatically.


\paragraph{Summary of Contributions}
In summary, the contributions of this paper are:
\begin{itemize}
\item An iterative algorithm for enforcing policies represented in
  PrivacyLFP, a rich logic with quantification over infinite domains,
  and formulation and proofs of the algorithm's properties
  (Section~\ref{sec:enforcement})
\item Use of mode analysis from logic programming to ensure that
  infinite quantifiers result only in a finite number of relevant
  substitutions (Section~\ref{sec:enforcement})
\item A formal model of incomplete audit logs as three-valued
  structures (Section~\ref{sec:structures})
\end{itemize}

\paragraph{Organization}
In Section~\ref{sec:logic}, we review PrivacyLFP to the extent needed
for this paper. Section~\ref{sec:structures} presents partial
structures and defines the semantics of PrivacyLFP over
them. Section~\ref{sec:enforcement} presents our policy enforcement
algorithm and its properties. Section~\ref{sec:instances} discusses the
behavior of our algorithm on structures that are complete and those
that are complete up to a point of time. In the latter case, we also
present associated results about enforcement of safety and co-safety
properties. Section~\ref{sec:hipaa-case} describes how the work in
this paper applies to the HIPAA Privacy Rule.
Section~\ref{sec:related} provides a detailed comparison with related
work and Section~\ref{sec:conclusion} presents conclusions and
directions for future work.  


\section{Policy Logic}
\label{sec:logic}



We use PrivacyLFP~\cite{DeYoung+10:privacylfp:wpes} to represent
policies, but restrict the syntax of first-order quantifiers slightly
to facilitate enforcement and drop fixpoint operators. PrivacyLFP
consists of an outer policy logic with connectives of temporal logic
and an inner, equally expressive sublogic without connectives of
temporal logic to which the outer syntax is translated. Our
enforcement algorithm works only with the inner sublogic. In this
section we review both the outer syntax and the sublogic, as well as
the translation.



\subsection{Syntax of the Policy Logic}
\label{sec:syntax}

\begin{figure}
\[\begin{array}{lllll}
\mbox{Objective predicates} & p_O \\
\mbox{Subjective predicates} & p_S\\
\mbox{Objective atoms} & P_O & ::= & p_O(t_1,\ldots,t_n) \\
\mbox{Subjective atoms} & P_S & ::= & p_S(t_1,\ldots,t_n)\\

\mbox{Formulas} & \alpha, \beta & ::= & P_O ~|~ P_S ~|~ \top ~|~
\bot~| \\ & & & \alpha_1 \conj \alpha_2 ~|~ \alpha_1 \disj
\alpha_2 ~|~ \neg \alpha ~| \\ & & & \forall \vec{x}. (c
\imp \alpha) ~|~ \exists \vec{x}. (c \conj \alpha) ~| \\ & &
& \here x.\alpha ~|~ \alpha \since \beta ~|~ \alpha
\until\beta ~| \\ & & & \boxm \alpha~|~ \boxp \alpha\\
\mbox{Restrictions} & c & ::= & P_O ~|~ \top ~|~ \bot ~|~
c_1 \conj c_2 ~|\\
& & &  c_1 \disj c_2 ~|~ \exists x. c
\end{array}\]
\caption{Timed First-order Temporal Logic with Restricted Quantifiers}
\label{fig:syntax-temporal}
\end{figure}

The syntax of our policy logic is shown in
Figure~\ref{fig:syntax-temporal}. We distinguish two classes of
predicate symbols: 1)~\emph{objective predicates}, denoted $p_O$, that
can be decided automatically using information from logs or using
constraint solvers and 2)~\emph{subjective predicates}, denoted $p_S$,
that require human input to resolve. Both classes of predicates are
illustrated in examples later. An atom is a predicate applied to a
list of terms (terms are denoted $t$). Based on the class of its
predicate, an atom is also classified as either objective or
subjective, written $P_O$ and $P_S$, respectively.


Propositional connectives $\top$ (true), $\bot$ (false), $\conj$
(conjunction), $\disj$ (disjunction), and $\neg$ (negation) have their
usual meanings. Anticipating the requirements of the enforcement
algorithm of Section~\ref{sec:enforcement}, first-order quantifiers
$\forall \vec{x}. (c \imp \alpha)$ and $\exists \vec{x}. (c \conj
\alpha)$ in the logic are forced to include a formula $c$ called a
\emph{restriction}. By definition, $\forall x. (c \imp \alpha)$ is
true iff all instances of $\vec{x}$ that satisfy $c$, also satisfy
$\alpha$. ($\exists \vec{x}. (c \conj \alpha)$ has a similar
definition.) To make enforcement tractable, we require that the set of
instances of $\vec{x}$ satisfying $c$ be computable. This is ensured
by limiting $c$ to a reduced class of formulas that, in particular,
excludes subjective predicates (see the syntax of $c$ in
Figure~\ref{fig:syntax-temporal}), and through a static analysis that
we describe in Section~\ref{sec:enforcement}.


Further, our logic includes standard connectives of linear temporal
logic (LTL)~\cite{manna95:temporal} that provide quantification over
the sequence of states in a system, relative to a current state:
$\alpha \since \beta$ ($\beta$ holds at some state in the past and
$\alpha$ holds since then), $\alpha \until \beta$ ($\beta$ holds at
some state in the future and $\alpha$ holds until then), $\boxm
\alpha$ ($\alpha$ holds at all states in the past) and $\boxp\alpha$
($\alpha$ holds at all states in the future). Other temporal operators
can be defined, e.g., $\diam \alpha = \top \since \alpha$ ($\alpha$
holds at some state in the past) and $\diap \alpha = \top \until
\alpha$ ($\alpha$ holds at some state in the future).


Finally, to represent clock time, which often occurs in privacy
policies, we assume that each state of a system has a time point
associated with it. Time points, denoted $\ttime$, are elements of
$\mathrm{T} = \{x \in \mathrm{R} ~|~ x \geq 0 \} \cup
\{\infty\}$. They measure clock time elapsed from a fixed reference
point and order states linearly.  Relations between time points are
captured in logical formulas using the \emph{freeze} quantifier $\here
x. \alpha$ of timed propositional temporal logic
(TPTL)~\cite{AlurHenzinger:artl}, which means ``$\alpha$ holds with
the current time bound to $x$.'' (Examples below illustrate the
quantifier.) Since we have no occasion to reason explicitly about
states, we identify a state with the time point associated with it,
and use the letter $\ttime$ and any of the terms ``state'', ``time
point'', ``time'', and ``point'' to refer to both states and time
points. We make the assumption that on any trace there are only
finitely many time points between two given finite time points.



We illustrate the syntax of our logic through two examples that are
based on the formalization of HIPAA in PrivacyLFP. These examples are
also used later in the paper.

\begin{eg}\label{eg:past}
As a first example, we represent in our logic the following policy
about disclosure (transmission) of health information from one entity
(e.g., a hospital or doctor) to another.
\begin{quote}
An entity may send an individual's protected health information (phi)
to another entity only if the receiving entity is the patient's doctor
and the purpose of the transmission is treatment, or the individual
has previously consented to the transmission.
\end{quote}


Our formalization assumes that each transmitted message $m$
is \emph{tagged} by the sender (in a machine-readable format) with the
names of individuals whose information it carries as well the
attributes of information it carries (attributes include ``address'',
``social security number'', ``medications'', ``medical history'',
etc.). The predicate $\pred{tagged}(m,q,t)$ means that message $m$ is
tagged as carrying individual $q$'s attribute $t$. Tagging may or may
not reflect accurately the content of the message. Similarly, we
assume that each message $m$ is labeled in a machine readable format
with a purpose $u$ (e.g., ``treatment'', ``healthcare'', etc.). This
is represented by the predicate $\pred{purp}(m,u)$. Because we assume
that name and attribute tags as well as purpose labels are machine
readable, both $\pred{tagged}$ and $\pred{purp}$ are \emph{objective
predicates} -- their truth or falsity can be checked using a program.

Attributes are assumed to have a hierarchy, e.g., the attribute
``medications'' is contained in ``medical history''. This is
formalized as the predicate $\pred{attr\_in}(\mbox{medications},
\mbox{medical-history})$. We assume that the hierarchy can be
mechanically checked, so $\pred{attr\_in}$ is an objective
predicate. The predicate $\pred{purp\_in}(u,u')$ means that purpose
$u$ is a special case of purpose $u'$, e.g.,
$\pred{purp\_in}(\mbox{surgery}, \mbox{treatment})$. In contrast to
attributes, we assume that the purpose hierarchy cannot be computed,
so $\pred{purp\_in}$ is a subjective predicate. In an enforcement
system, it must be checked through human input.

Finally, each action or fact that can be recorded in a system log
(such as sending a message or that Alice is in role doctor) is
represented as an objective predicate. For this example we need three
objective predicates: $\pred{send}(p_1,p_2,m)$ meaning that entity
$p_1$ sends message $m$ to entity $p_2$, $\pred{consents}(q, a)$ which
means that individual $q$ consents to the action $a$, and
$\pred{inrole}(p,r)$ which means that principal $p$ is in role
$r$. Here, the only action consented to is
$\pred{sendaction}(p_1,p_2,(q,t))$, which corresponds to $p_1$ sending
to $p_2$ a message containing information about $q$'s attribute $t$.

The above policy can be formalized in our logic as follows.
\begin{tabbing}
\=$\alpha_{pol1}$ = \\
\>~$\forall p_1,$\=$ p_2, m,u,q,t.~($\=$\pred{send}(p_1, p_2,
m) \conj \pred{purp}(m,u) \conj$\\
\>\>\>$\pred{tagged}(m,q,t) \conj \pred{attr\_in}(t, \attr{phi}))$\\
\>\> $\imp$ \= $(\pred{inrole}(p_2, \pred{doc}(q)) \conj \pred{purp\_in}(u,\purpose{treatment}))$\\
\>\>\> $\disj\diam\pred{consents}(q, \pred{sendaction}(p_1, p_2, (q,t)))$
\end{tabbing}

In words, if entity $p_1$ sends to entity $p_2$ a message $m$, $m$ is
tagged as carrying attribute $t$ of individual $q$, where $t$ is a
form of $phi$ (protected health information), and $m$ is labeled with
purpose $u$, then either $p_2$ (the recipient) is a doctor of $q$
(atom $\pred{inrole}(p_2, \pred{doc}(q))$) and $u$ is a type of
treatment, or $q$ has consented to this transmission in the past (last
line of $\alpha_{pol1}$). The temporal operator $\diam$ is used to
indicate that the consent may have been given by $q$ in some earlier
state. Also, the universal quantifier in the formula above carries a
restriction $(\pred{send}(p_1, p_2, m) \conj \pred{purp}(m,u) \conj
\pred{tagged}(m,q,t) \conj \pred{attr\_in}(t, \attr{phi}))$, as
required by our syntax. The technical reason for including
restrictions is explained in Section~\ref{sec:enforcement}.
\end{eg}



\begin{eg}\label{eg:future}
Our next example is a policy governing entity response to an
individual's request for her own information.
\begin{quote}
If an individual requests her information from an entity, then some
administrator in the records department of the entity must respond to
the individual at the earliest feasible time, but not later than 30
days after the request.
\end{quote}

To represent this policy we need one more objective predicate,
$\pred{req}(p,t)$, which means that individual $p$ requests
information about attribute $t$ from her record. Further, we need two
new subjective predicates: $\pred{contains}(m,q,t)$ (message $m$
contains attribute $t$ of individual $q$) and $\pred{ftr}(p,t)$ (it is
feasible to respond to individual $p$ with attribute $t$ at the
current time). The latter clearly requires human input to resolve,
because ``feasibility'' cannot be defined mechanically, while the
former requires human input because we assume that message payloads
may contain natural language text.

The logical specification of this policy is shown below:
\begin{tabbing}
\=$\alpha_{pol2}$ = \\
\>~$\here \ttime. \forall p,$\=$t.~\pred{req}(p, t)$\\
\>\>$\imp$\=$ ~\neg\pred{ftr}(p,  t)$\\
\>\> \>$\until$ $\here\ttime'.$\=$ ~\pred{in}(\ttime',\ttime,\ttime+30)$ \\
\>\>\>\> $\conj\exists q,m.~ ($\= $\pred{inrole}(q,\mbox{records}) \conj \pred{send}(q, p,m) \conj $\\
\>\>\>\>\>$\pred{contains}(m, p,t))$
\end{tabbing}

The top-most quantifier $\here\ttime$ binds $\ttime$ to the time at
which a request occurs and, similarly, $\here \ttime'$ binds $\ttime'$
to the time at which a response is
sent. $\pred{in}(\ttime',\ttime,\ttime+30)$, formally explained in
Section~\ref{sec:translation}, implies that $\ttime' \leq \ttime+30$,
thus enforcing the constraint that the response be sent within 30 days
of the request, as required by the policy. The until operator $\until$
is used to include the obligation that it be infeasible to respond
until the response is actually sent.
\end{eg}

\subsection{Translation to a Smaller Syntax}
\label{sec:translation}

Policies expressed in PrivacyLFP's outer syntax can be translated into
a smaller sublogic without temporal connectives and negation. This
smaller syntax of formulas $\varphi,\psi$ of the sublogic is shown
below. Other syntactic categories such as restrictions $c$ are not
changed.
\[\begin{array}{@{}lllll}
\mbox{Formulas} & \varphi & ::= & P_O ~|~ P_S ~|~ \top ~|~ \bot ~|~
\varphi_1 \conj \varphi_2 ~|~ \varphi_1 \disj \varphi_2 ~| \\ & & &
\forall \vec{x}. (c \imp \varphi) ~|~ \exists \vec{x}. (c \conj
\varphi)
\end{array}\]

We surmount the absence of negation in the sublogic by \emph{defining}
for each formula $\varphi$ a dual $\dual{\varphi}$ that behaves
exactly as $\neg \varphi$ would. For defining duals of atoms, we
assume that each predicate $p$ has a dual $\dual{p}$ such that
$p(t_1,\ldots,t_n)$ is true iff $\dual{p}(t_1,\ldots,t_n)$ is false
(the relation between $p$ and $\dual{p}$ is formalized in
Section~\ref{sec:structures}). We define $\dual{\varphi}$ by induction
on $\varphi$, as in the representative clauses below (for the
remaining clauses see Appendix~\ref{app:logic}).
\[\begin{array}{ccc}
\dual{p_O(t_1,\ldots,t_n)} & = & \dual{p_O}(t_1,\ldots,t_n) \\
\dual{\varphi \conj \psi} & = & \dual{\varphi} \disj \dual{\psi} \\
\dual{\forall \vec{x}. (c \imp \varphi)} & = & \exists \vec{x}. (c \conj \dual{\varphi})\\
\dual{\exists \vec{x}. (c \conj \varphi)} & = & \forall \vec{x}. (c \imp \dual{\varphi})
\end{array}\]

Temporal connectives are translated to the sublogic by making time
points (states) and the ordering relation between them explicit in
first-order formulas in a standard way
(see~\cite{DeYoung+10:privacylfp:wpes}). Briefly, we assume that for
every predicate symbol in the logic there is a predicate of the same
name in the sublogic, but with one extra argument of type time:
$p(t_1,\ldots,t_n,\ttime)$ in the sublogic means that
$p(t_1,\ldots,t_n)$ holds at time $\ttime$ in the logic. Further,
assume that the new objective predicate $\pred{in}(\ttime, \ttime_1,
\ttime_2)$ means that $\ttime$ is an observed time point (in the trace
of interpretation) satisfying $\ttime_1 \leq \ttime \leq \ttime_2$.
Finally, let $\Xi [\vec{t}/\vec{x}]$ denote the result of substituting
the terms $\vec{t}$ for variables $\vec{x}$ in the syntactic entity
$\Xi$.  Then, representative clauses of the translation
$\trans{\ttime}{\bullet}$ of restrictions and formulas of the logic to
those of the sublogic, indexed by a ``current time'' $\ttime$, are
shown below (the full translation is shown in
Appendix~\ref{app:logic}):
\[\begin{array}{ccl}

\trans{\ttime}{p_O(t_1,\ldots,t_n)} & = &
p_O(t_1,\ldots,t_n,\ttime) \\

\trans{\ttime}{p_S(t_1,\ldots,t_n)} & = &
p_S(t_1,\ldots,t_n,\ttime) \\

\trans{\ttime}{\neg \alpha} & = & \dual{\trans{\ttime}{\alpha}}\\

\trans{\ttime}{\forall \vec{x}. (c \imp \alpha)} & =
& \forall \vec{x}. (\trans{\ttime}{c} \imp
\trans{\ttime}{\alpha})\\

\trans{\ttime}{\here x. \alpha} & = &
\trans{\ttime}{\alpha[\ttime/x]}\\

\trans{\ttime}{\alpha \since \beta} & = & \exists
\ttime'. (\pred{in}(\ttime', 0, \ttime) \conj
\trans{\ttime'}{\beta} 
\\ & & ~ \conj (\forall \ttime''. (({\tt
  in}(\ttime'',\ttime',\ttime) \conj \ttime' \not= \ttime'') 
\\& & \qquad\quad \imp
\trans{\ttime''}{\alpha})))\\

\trans{\ttime}{\alpha \until \beta} & = & \exists
\ttime'. (\pred{in}(\ttime', \ttime, \infty) \conj
\trans{\ttime'}{\beta} 
\\ & & ~ \conj (\forall \ttime''. (({\tt
  in}(\ttime'',\ttime,\ttime') \conj \ttime'' \not= \ttime') 
\\& & \qquad\quad\imp
  \trans{\ttime''}{\alpha})))\\


\end{array}\]

We briefly explain some of the clauses of the translation.  In
$\trans{\ttime}{\here x.\alpha}$, $x$ binds to the current time, which
is $\ttime$; therefore, $\ttime$ substitutes $x$ in $\alpha$ in the
translation. $\alpha\since\beta$ means that $\beta$ is true at some
time point in the past, which is captured by the existentially
quantified variable $\ttime'$ in the translation, and the restriction
that $\pred{in}(\ttime', 0, \ttime)$. Further, $\alpha$ should be true
at all time points between $\ttime'$ and now ($\ttime$); this is
encoded as $\forall \ttime''. ((\pred{in}(\ttime'',\ttime',\ttime)
\conj \ttime'' \not= \ttime')\imp \trans{\ttime''}{\alpha})$.


\begin{eg}\label{eg:policies:translated}
In Section~\ref{sec:syntax} we presented two sample policies,
$\alpha_{pol1}$ and $\alpha_{pol2}$. In general, we may wish to
enforce each of these policies in each state. To express the phrase
``in each state'', we define an abbreviation: $\globally \alpha =
\forall \ttime. (\pred{in}(\ttime,0,\infty) \imp
\trans{\ttime}{\alpha})$, which means that $\alpha$ holds at each time
point $\ttime$. Then, using the translation above and simplifying
slightly, we get:
\begin{tabbing}
\=$\globally \alpha_{pol1}$ = \\
\>~$\forall \ttime, $\=$ p_1, p_2, m, u, q, t.$ \\
\>\>$(\pred{in}(\ttime,0,\infty) \conj \pred{send}(p_1, p_2, m,\ttime) \conj \pred{purp}(m,u,\ttime) \conj$\\
\>\>$~\pred{tagged}(m,q,t,\ttime) \conj \pred{attr\_in}(t, \attr{phi},\ttime))$\\
\>\>$~\imp($\=$($\=$ \pred{inrole}(p_2, \pred{doc}(q),\ttime) \conj$\\
\>\>\>\> $\pred{purp\_in}(u,\purpose{treatment},\ttime)) \disj$\\
\>\>\>$(\exists \ttime'.~ ($\=$ \pred{in}(\ttime',0,\ttime) \conj$\\
\>\>\>\>$\pred{consents}(q, \pred{sendaction}(p_1, p_2, (q,t)),\ttime'))))$
\end{tabbing}

\begin{tabbing}
\=$\globally \alpha_{pol2}$ = \\
\>~$\forall \ttime$\=$,p,t. ~(\pred{in}(\ttime,0,\infty) \conj \pred{req}(p, t, \ttime))$\\
\>\> $\imp \exists 
\ttime'$\=$,q,m.$\\
\>\>\> $(($\=$\pred{in}(\ttime',\ttime,\ttime+30) \conj \pred{inrole}(q,
\mbox{records},\ttime') \conj$\\
\>\>\>\> $\pred{send}(q, p,  m,\ttime')) \conj \pred{contains}(m,p,t,\ttime') \conj$\\
\>\>\> $\forall \ttime'' $\=$.~(\pred{in}(\ttime'', \ttime, \ttime') \conj
\ttime'' \not= \ttime') $\\
\>\>\>\>$\imp \dual{\pred{ftr}}(p, t, \ttime''))$
\end{tabbing}



Note that all atoms, except those like $\pred{in}(\ldots)$ and
$\ttime'' \not= \ttime'$ that are introduced by the translation
itself, have a new last argument, which is a time point. For certain
predicates like $\pred{tagged}$, $\pred{attr\_in}$ and
$\pred{purp\_in}$, whose truth is independent of time, this last
argument is redundant. For instance, if $\pred{attr\_in}(t,t',\ttime)$
for some $\ttime$, then $\pred{attr\_in}(t,t',\ttime')$ for
all~$\ttime'$.
\end{eg}

\section{Partial Structures and Semantics}
\label{sec:structures}

Next, we define \emph{partial structures}, an abstraction of audit
logs over which our enforcement algorithm
(Section~\ref{sec:enforcement}) works. We call our structures partial
because they do not necessarily stipulate the truth or falsity of
every atom, thus accurately reflecting the fact that audit logs may be
incomplete in practice. We also illustrate, by virtue of example,
various kinds of audit log incompleteness that our partial structures
generalize. Finally, we define the \emph{semantics} (meanings) of
formulas of the sublogic on partial structures. This definition is
used in Section~\ref{sec:enforcement} to state the correctness of our
enforcement mechanism. Partial structures have been used, both
explicitly and implicitly, in prior work on policy enforcement; we
compare to such work in Section~\ref{sec:related}.

Fix a domain of individuals $\domain$. A \emph{partial structure}
(abbrev.\ structure) $\structure$ over $\domain$ consists of a total
function $\rho_\structure$ from ground (variable-free) atoms of the
logic to the \emph{three-value} set $\{\symt,\symf,\symu\}$. We say
that the atom $P$ is true, false, or unknown in the structure
$\structure$ if $\rho_\structure(P)$ is $\symt$, $\symf$, or $\symu$,
respectively. In practice, the structure $\structure$ may be defined
using system logs (hence the notation $\structure$), whence
for every subjective atom $P_S$, $\rho_\structure(P_S)$ would be
$\symu$.

The semantics of our sublogic lift the definition of \emph{truth} to
formulas $\varphi$ by induction on $\varphi$: we write $\structure
\models \varphi$ to mean that ``$\varphi$ is true in the structure
$\structure$''. Restrictions $c$ are a subsyntax of formulas
$\varphi$, so we do not define the relation separately for them.
\begin{itemize}
\item[-] $\structure \models P$ iff $\rho_\structure(P) = \symt$
\item[-] $\structure \models \top$
\item[-] $\structure \models \varphi \conj \psi$ iff $\structure
  \models \varphi$ and $\structure \models \psi$
\item[-] $\structure \models \varphi \disj \psi$ iff $\structure
  \models \varphi$ or $\structure \models \psi$
\item[-] $\structure \models \forall \vec{x}. (c \imp \varphi)$ iff
  for all $\vec{t} \in \domain$ either $\structure \models
  \dual{c}[\vec{t}/\vec{x}]$ or $\structure \models
  \varphi[\vec{t}/\vec{x}]$
\item[-] $\structure \models \exists \vec{x}. (c \conj \varphi)$ iff
  there exists $\vec{t} \in \domain$ such that $\structure
  \models c[\vec{t}/\vec{x}]$ and $\structure \models
  \varphi[\vec{t}/\vec{x}]$
\end{itemize}

For dual atoms, we define $\rho_\structure(\dual{P}) =
\dual{\rho_\structure(P)}$, where $\dual{\symt} = \symf$,
$\dual{\symf} = \symt$, and $\dual{\symu} = \symu$. We say that a
formula $\varphi$ is \emph{false} on the structure $\structure$ if
$\structure \models \dual{\varphi}$. The following two properties
hold:
\begin{enumerate}
\item Consistency: A formula $\varphi$ cannot be simultaneously true
  and false in the structure $\structure$, i.e., either $\structure
  \not \models \varphi$ or $\structure \not \models \dual{\varphi}$
\item Incompleteness: A formula $\varphi$ may be neither true nor
  false in a structure $\structure$, i.e., $\structure \not \models
  \varphi$ and $\structure \not \models \dual{\varphi}$ may both hold.
\end{enumerate}
The first property follows by induction on $\varphi$. The second
property follows from a simple example. Consider a structure
$\structure$ and an atom $P$ such that $\rho_{\structure}(P) =
\symu$. Then, $\structure \not \models P$ and $\structure \not \models
\dual{P}$.

\paragraph{Incompleteness in Practice}
We list below several ways in which system logs may be incomplete, and
describe how each can be modeled in partial structures by varying the
definition of~$\rho_\structure$.
\begin{itemize}
\item Subjective incompleteness: An audit log may not contain
  information about subjective predicates. This may be modeled by
  requiring that $\rho_\structure(P_S) = \symu$ for every subjective
  atom $P_S$. We revisit subjective incompleteness in the context of
  our enforcement algorithm in Section~\ref{sec:instance:oc}.
\item Future incompleteness: An audit log may not contain information
  about the future, which is necessary to enforce policies like that
  in Example~\ref{eg:future}. This may be modeled by assuming that for
  each time $\ttime$ greater than the last point observed in
  $\structure$, and for all $p$, $t_1,\ldots,t_n$,
  $\rho_\structure(p(t_1,\ldots,t_n,\ttime)) = \symu$. (Recall that in
  our translation of the outer logic, the last argument $\ttime$ is
  the time at which the predicate's truth is tested.) We revisit
  future incompleteness in the context of our enforcement algorithm in
  Section~\ref{sec:instance:pc}.
\item Spatial incompleteness: An audit log may not record all
  predicates. For instance, with reference to Example~\ref{eg:past},
  it is conceivable that the predicates $\pred{send}$ and
  $\pred{inrole}$ are stored on separate sites. If we audit at the
  first site, information about $\pred{inrole}$ may be unavailable.
  Such incompleteness is easily modeled like subjective
  incompleteness. For instance, we may assume that
  $\rho_\structure(\pred{inrole}(p,r,\ttime)) = \symu$ for all
  $p,r,\ttime$.
\item Past incompleteness: An audit log may not record the
  \emph{existence} of certain relevant states, even those in the
  past. This has implications for enforcing temporal operators, e.g.,
  we may be unable to check that $\boxm \alpha$ simply because we
  cannot determine what states existed in the past. This form of
  incompleteness can be formally modeled by assuming that if a time
  point $\ttime$ does not occur in an audit log $\structure$, then
  $\rho_\structure(\pred{in}(\ttime,\ttime',\ttime'')) = \symu$. In
  the special case where it is \emph{certain} that the time point
  $\ttime$ \emph{does not exist}, we would have
  $\rho_\structure(\pred{in}(\ttime,\ttime',\ttime'')) = \symf$.
\end{itemize}

Our enforcement algorithm (Section~\ref{sec:enforcement}) works with
partial structures in general and, hence, takes into account all these
forms of incompleteness. We comment on some specific instances in
Section~\ref{sec:instances}.

\paragraph{Structure Extension}
In practice, system logs evolve over time by gathering more
information. This leads to a natural order, $\structure_1 \geq
\structure_2$ on structures ($\structure_1$ \emph{extends}
$\structure_2$), meaning that $\structure_1$ has more information than
$\structure_2$. Formally, $\structure_1 \geq \structure_2$ for all
ground atoms $P$, $\rho_{\structure_2}(P) \in \{\symt, \symf\}$
implies $\rho_{\structure_1}(P) = \rho_{\structure_2}(P)$. Thus, as
structures extend, the valuation of an atom may change from $\symu$ to
either $\symt$ or $\symf$, but cannot change once it is either $\symt$
or $\symf$. The following property follows by induction on $\varphi$:
\begin{itemize}
\item Monotonicity: $\structure_1 \geq \structure_2$ and $\structure_2
  \models \varphi$ imply $\structure_1 \models \varphi$.
\end{itemize}
Replacing $\varphi$ with $\dual{\varphi}$, we also obtain that
$\structure_1 \geq \structure_2$ and $\structure_2 \models
\dual{\varphi}$ imply $\structure_1 \models \dual{\varphi}$. Hence, if
$\structure_1 \geq \structure_2$ then $\structure_1$ preserves both
the $\structure_2$-truth and $\structure_2$-falsity of every formula
$\varphi$.

In the next section, we use this order between structures to both
explain and prove formal properties of our enforcement algorithm.


\section{Policy Enforcement}
\label{sec:enforcement}

Our main technical contribution is an iterative process for enforcing
policies written in the sublogic. Through the translation of
Section~\ref{sec:translation}, the same process applies to policies
written in the entire policy logic. At each iteration, our algorithm
takes as input a policy $\varphi$ and the available audit log
abstracted as a partial structure $\structure$, and outputs a residual
policy $\psi$ that contains exactly the parts of $\varphi$ that could
not be verified due to lack of information in $\structure$. Such an
iteration is written $\audit{\structure}{\varphi}{\psi}$. In practice,
$\psi$ may contain subjective predicates and future obligations. Once
more information becomes available, extending $\structure$ to
$\structure'$ ($\structure' \geq \structure$), another iteration of
the algorithm can be used with inputs $\psi$ and $\structure'$ to
obtain a new formula $\psi'$. This process can be continued till a
formula trivially equivalent to $\top$ or $\bot$ is obtained, or the
truth or falsity of the remaining formula is decided by human
intervention.  By design, our algorithm satisfies three important
properties:
\begin{itemize}
\item Termination: Each iteration terminates.
\item Correctness: If $\audit{\structure}{\varphi}{\psi}$, then for
  all extensions $\structure'$ of $\structure$, $\structure' \models
  \varphi$ iff $\structure' \models \psi$.
\item Minimality: If $\audit{\structure}{\varphi}{\psi}$, then an
  atom occurs in $\psi$ only if it occurs in $\varphi$ and its
  valuation on $\structure$ is $\symu$.
\end{itemize}

The technically difficult part of the algorithm is its treatment of
quantifiers $\forall x. \varphi$ and $\exists x.\varphi$ in the
input. Indeed, for \emph{propositional logic} (logic without
quantifiers), an algorithm satisfying the three properties above can
be constructed trivially: define $\reduce(\structure,\varphi)$ to be
the formula obtained by replacing each atom $P$ in $\varphi$ with
$\top$ if $\rho_\structure(P) = \symt$, with $\bot$ if
$\rho_\structure(P) = \symf$, and with $P$ itself if
$\rho_\structure(P) = \symu$. This algorithm terminates because
formulas are finite, its correctness can be proved by a simple
induction on $\varphi$, and minimality is obvious from the definition
of $\reduce$.

However, as the reader may already anticipate, this simple idea does
not extend to quantifiers. Consider, for instance, the behavior of the
algorithm on inputs $\forall x.\varphi$ and $\structure$. Because the
output must be minimal, in order to reduce $\forall x.\varphi$, the
algorithm must instantiate $x$ with each possible element of the
domain $\domain$ and check the truth or falsity of $\varphi$ for that
instance on $\structure$. This immediately leads to non-termination
because in models of realistic privacy policies the domain $\domain$
must be infinite, e.g., permissible time points and transmitted
messages (which may contain free-text in natural language) are both
infinite sets.

Given the need for an infinite domain, something intrinsic in
$\varphi$ must \emph{limit} the number of \emph{relevant instances} of
$x$ that need to be checked to a finite number. This is precisely what
our restricted form of universal quantification, $\forall \vec{x}. (c
\imp \varphi)$, accomplishes. Through syntactic restrictions of
Figure~\ref{fig:syntax-temporal} and other static checks described
later, we ensure that there are only a finite number of instances of
$\vec{x}$ for which $c$ is true on the given structure
$\structure$. Further, all such instances can be mechanically computed
from $\structure$. Although fulfilling these requirements is
non-trivial, given that they hold, the rest of the algorithm is
natural and syntax-directed.


Briefly, our enforcement regime contains the following components:
\begin{itemize}
\item An efficiently checkable relation $\vdash \varphi$ on policies,
  called a \emph{mode analysis} (borrowing the term from logic
  programming~\cite{apt94:modes}), which ensures that the relevant
  instances of each quantified variable in $\varphi$ are finite and
  computable.
\item A function $\lift{\sat}(\structure, c)$ that computes all
  satisfying instances of the restriction $c$.
\item The function $\reduce(\structure, \varphi)$ that codifies a
  single iteration of enforcement. The definition of
  $\reduce(\structure, \varphi)$ relies on $\lift{\sat}(\structure,
  c)$ and assumes that $\vdash \varphi$.
\end{itemize}

In the following, we explain each of these three components, starting
with the main algorithm $\reduce$ (Section~\ref{sec:reduce}). After
proving its correctness and minimality
(Section~\ref{sec:correctness:minimality}), we proceed to define
$\lift{\sat}$ and the relation $\vdash \varphi$
(Section~\ref{sec:mode}).




\subsection{Iterative Enforcement Algorithm}
\label{sec:reduce}

The core of our enforcement regime is a computable function
$\reduce(\structure, \varphi) = \psi$, that discharges obligations
from the prevalent policy $\varphi$ using information from the extant
structure $\structure$ to obtain a residual policy $\psi$. Given an
initial policy $\varphi_0$ and a sequence of structures $\structure_1
\leq \structure_2 \leq \ldots \leq \structure_n$, the reduction
algorithm can be applied repeatedly to obtain
$\varphi_1,\ldots,\varphi_n$ such that $\reduce(\structure_i,
\varphi_{i-1}) = \varphi_i$. We write this process in symbols as
$\varphi_0 \iterate{\structure_1} \varphi_1 \ldots
\iterate{\structure_n} \varphi_n$. Correctness
(Theorem~\ref{thm:correctness}) guarantees that $\varphi_n$ is
\emph{equivalent} to $\varphi_0$ in all extensions of $\structure_n$,
while minimality (Theorem~\ref{thm:minimality}) certifies that
$\varphi_n$ contains only those atoms of $\varphi_0$ that could not be
discharged using the information in $\structure_n$ (by definition,
$\structure_n$ subsumes the information in $\structure_1, \ldots,
\structure_{n-1}$). We note that our correctness and minimality
results are independent of the frequency or scheme used for
application of $\reduce$. 


The definition of $\reduce(\structure, \varphi)$ has two dependencies,
whose formal definitions are postponed to
Section~\ref{sec:mode}. First, the function assumes that its input
$\varphi$ is \emph{well-moded}, formally written $\vdash
\varphi$. Well-modedness is a static check, linear in the size of
$\varphi$, which \emph{ensures} that the satisfying instances of each
restriction $c$ in each quantifier in $\varphi$ are \emph{finite and
  computable}. Second, $\reduce(\structure, \varphi)$ assumes a
function $\lift{\sat}(\structure, c)$ that computes all satisfying
instances of restriction $c$ in structure $\structure$. The output of
$\lift{\sat}(\structure,c)$ is a finite set of substitutions
$\{\sigma_1, \ldots, \sigma_n\}$, where each substitution $\sigma_i$
is a finite map from free variables of $c$ to ground
terms. $\lift{\sat}(\structure, c)$ satisfies the following condition:
$\structure \models c \sigma$ iff $\sigma \in \lift{\sat}(\structure,
c)$.

\begin{figure}
\[\begin{array}{@{}lll}

\reduce(\structure, P) & = & \left\lbrace
\begin{array}{ll}
\top & \mbox{ if $\rho_{\structure}(P) = \symt$}\\
\bot & \mbox{ if $\rho_{\structure}(P) = \symf$}\\
P & \mbox{ if $\rho_{\structure}(P) = \symu$}
\end{array} \right. \\

\\

\reduce(\structure, \top) & = & \top \\

\reduce(\structure, \bot) & = & \bot \\

\reduce(\structure, \varphi_1 \conj \varphi_2) & = &
\reduce(\structure, \varphi_1) \conj \reduce(\structure, \varphi_2) \\

\reduce(\structure, \varphi_1 \disj \varphi_2) & = &
\reduce(\structure, \varphi_1) \disj \reduce(\structure, \varphi_2) \\

\\

\reduce(\structure, \forall \vec{x}.(c \imp \varphi)) & = & 
\begin{array}[t]{@{}l}
\mbox{let}\\
~\{\sigma_1,\ldots,\sigma_n\} \leftarrow \lift{\sat}(\structure,c)\\
~\{ \vec{t_i} \leftarrow \sigma_i(\vec{x}) \}_{i=1}^n\\
~S \leftarrow \{\vec{t_1},\ldots,\vec{t_n}\}\\
~\{ \psi_i \leftarrow \reduce(\structure,
\varphi[\vec{t_i}/\vec{x}]) \}_{i=1}^n\\
~\psi' \leftarrow \forall \vec{x}.((c \conj \vec{x} \not \in S) \imp \varphi)\\
\mbox{return } \\
~~~ \psi_1 \conj \ldots \conj \psi_n \conj \psi'
\end{array} \\
\\

\reduce(\structure, \exists \vec{x}.(c \conj \varphi)) & = & 
\begin{array}[t]{@{}l}
\mbox{let}\\
~\{\sigma_1,\ldots,\sigma_n\} \leftarrow \lift{\sat}(\structure,c)\\
~\{ \vec{t_i} \leftarrow \sigma_i(\vec{x}) \}_{i=1}^n\\
~S \leftarrow \{\vec{t_1},\ldots,\vec{t_n}\}\\
~\{ \psi_i \leftarrow \reduce(\structure,
\varphi[\vec{t_i}/\vec{x}]) \}_{i=1}^n\\
~\psi' \leftarrow \exists \vec{x}.((c \conj \vec{x} \not \in S) \conj \varphi)\\
\mbox{return } \\
~~~ \psi_1 \disj \ldots \disj \psi_n \disj \psi'
\end{array}

\end{array}\]
\caption{Definition of $\reduce(\structure, \varphi)$}
\label{fig:reduce}
\end{figure}

The function $\reduce(\structure,\varphi)$ is defined by induction on
$\varphi$ in Figure~\ref{fig:reduce}.  For atoms $P$,
$\reduce(\structure, P)$ equals $\top$, $\bot$, or $P$, according to
whether $\rho_{\structure}(P)$ equals $\symt$, $\symf$, or $\symu$. In
particular, in the absence of human input $\rho_{\structure}(P_S) =
\symu$ for a subjective atom $P_S$ and hence, in the absence of human
input, $\reduce(\structure, P_S) = P_S$. The clauses for the
connectives $\top$, $\bot$, $\conj$, and $\disj$ are
straightforward. To evaluate $\reduce(\structure, \forall \vec{x}.(c
\imp \varphi))$, we first determine the set of instances of $\vec{x}$
that satisfy $c$ by calling $\lift{\sat}(\structure, c)$. For each
such instance $\vec{t_1},\ldots,\vec{t_n}$, we reduce
$\varphi[\vec{t_i}/\vec{x}]$ to $\psi_i$ through a recursive call to
$\reduce$. Because all instances of $\varphi$ must hold in order for
$\forall \vec{x}.(c\imp \varphi)$ to be true, the output is $\psi_1
\conj \ldots \conj \psi_n \conj \psi'$, where the last conjunct
$\psi'$ records the fact that instances of $\vec{x}$ \emph{other than}
$\vec{t_1},\ldots,\vec{t_n}$ have not been considered. The latter is
necessary because there may be instances of $\vec{x}$ satisfying $c$
in extensions of $\structure$, but not $\structure$ itself. Precisely,
we define $S = \{\vec{t_1}, \ldots, \vec{t_n}\}$ and $\psi' = \forall
\vec{x}.((c \conj \vec{x} \not \in S) \imp \varphi)$.  The new
conjunct $\vec{x} \not \in S$ prevents the instances $\vec{t_1},
\ldots, \vec{t_n}$ from being checked again in subsequent
iterations. Formally, $\vec{x} \not \in S$ is an objective predicate
that encodes the negation of usual finite-set membership.  The
treatment of $\exists \vec{x}.(c \conj \varphi)$ is dual; in that
case, the output contains disjunctions because the truth of any one
instance of $\varphi$ suffices for the formula to hold.

\begin{eg}\label{eg:reduce}
We illustrate iterative enforcement on the policy $\varphi_0 =
{\globally \alpha_{pol2}}$ that we obtained via translation in
Example~\ref{eg:policies:translated}. The policy requires that the
recipient of a request for information respond within 30 days with the
information. We advise the reader to revisit the example for the
definition of $\varphi_0$. For the purpose of explanation, let us
define $\varphi(\ttime,p,t)$ by pattern matching to be the formula
satisfying $\varphi_0 = \forall
\ttime,p,t. ~(\pred{in}(\ttime,0,\infty) \conj
\pred{req}(p,t,\ttime)) \imp \varphi(\ttime,p,t)$. Informally,
$\varphi(\ttime,p,t)$ is the obligation that must be satisfied if
principal $p$ requests information about attribute $t$ from her record
at time $\ttime$.

Suppose that we first run $\reduce(\structure, \varphi_0)$ in a
structure $\structure$ which has the states $1,3,7$, only one request
--- $\mbox{Alice}$'s request for her medical record (attribute $mr$)
at time $3$, and no other information. Intuitively, this information
implies that $\lift{\sat}(\structure, \pred{in}(\ttime,0,\infty) \conj
\pred{req}(p,t,\ttime)) = \{(\ttime,p,t) \mapsto (3,\mbox{Alice},mr)
\}$. (We check formally in Example~\ref{eg:lift:sat} that this is
actually the case.) Hence, by the definition of $\reduce$, we know
that $\reduce(\structure, \varphi_0) = \psi_1 \conj \varphi_0'$, where
$\psi_1 = \reduce(\structure,
\varphi[(3,\mbox{Alice},mr)/(\ttime,p,t)])$ and $\varphi_0' = \forall
\ttime,p,t. ~(\pred{in}(\ttime,0,\infty) \conj \pred{req}(p,t,\ttime)
\conj (\ttime,p,t) \not \in \{(3,\mbox{Alice},mr)\}) \imp
\varphi(\ttime,p,t)$. The reader may check that because the trace has
no other information, $\psi_1 =
\varphi[(3,\mbox{Alice},mr)/(\ttime,p,t)]$, so the output of the
reduction is $\psi_1 \conj \varphi_0'$. Expansion of the formula
$\psi_1$ shows that it is precisely the obligation that the recipient
respond to $\mbox{Alice}$ with her medical record in 30 days. Call
this entire output $\varphi_1$.

Consider a second round of audit on the reduced policy $\varphi_1$ and
an extended trace $\structure'$ which has the additional state $11$ in
which $\mbox{Bob}$, in role ``$\mbox{records}$'', responds with a
message $M$ to $\mbox{Alice}$.  Since $\varphi_1 = \psi_1 \conj
\varphi_0'$, we have $\reduce(\structure',\varphi_1) =
\reduce(\structure',\psi_1) \conj \reduce(\structure',
\varphi_0')$. The reader may check that
$\reduce(\structure',\varphi_0') = \varphi_0'$ because the top-level
restriction in $\varphi_0'$ has no satisfying instance in
$\structure'$. Thus, we consider here the reduction of $\psi_1$. Note
that $\psi_1$ has the form $\exists \ttime', q, m. ~(
(\pred{in}(\ttime',3,33) \conj \pred{inrole}(q,
\mbox{records},\ttime') \conj\pred{send}(q, \mbox{Alice}, m,\ttime'))
\conj \varphi'(\ttime',q,m))$.  To calculate its reduction, we first
observe that from the information in $\structure'$, it should follow
that $\lift{\sat}(\structure', \pred{in}(\ttime',3,33) \conj
\pred{inrole}(q, \mbox{records},\ttime') \conj \pred{send}(q,
\mbox{Alice}, m,\ttime')) = \{(\ttime',q,m) \mapsto (11,\mbox{Bob},M)
\}$. (Again, we check formally in Example~\ref{eg:lift:sat} that this
is the case.)  Consequently, $\reduce(\structure',\psi_1) = \psi_1'
\disj \varphi_1'$, where $\psi_1' = \reduce(\structure',
\varphi'(11,\mbox{Bob},M))$ and $\varphi_1' = \exists \ttime', q,
m. ~( (\pred{in}(\ttime',3,33) \conj \pred{inrole}(q,
\mbox{records},\ttime') \conj\pred{send}(q, \mbox{Alice}, m,\ttime')
\conj (\ttime',q,m) \not \in \{(11,\mbox{Bob},M) \}) \conj
\varphi'(\ttime',q,m))$. We calculate $\psi_1'$ below. The second
disjunct $\varphi_1'$ simply means that the policy is satisfied if at
some point other than $11$ (but before $33$), someone in role
``records'' sends $\mbox{Alice}$'s $mr$ to her.

What is $\psi_1' = \reduce(\structure', \varphi'(11,\mbox{Bob},M))$?
Expanding $\varphi'$, we have $\varphi'(11,\mbox{Bob},M) = \linebreak[6]
\pred{contains}(M,\mbox{Alice},mr,11) \conj \psi_2'$, where $\psi_2' =
\forall \ttime'' .~(\pred{in}(\ttime'', 3, 11) \conj \ttime'' \not=
11) \imp \dual{\pred{ftr}}(\mbox{Alice}, mr, \ttime'')$. Because
$\pred{contains}$ is a subjective predicate,
$\rho_{\structure'}(\pred{contains}(M,\mbox{Alice},mr,11)) = \symu$
so, by definition, $\reduce(\structure',
\pred{contains}(M,\mbox{Alice},mr,11)) =
\pred{contains}(M,\mbox{Alice},mr,11)$. Hence, if
$\reduce(\structure', \psi_2') = \psi_2''$, then $\psi_1' =
\pred{contains}(M,\mbox{Alice},mr,11) \conj \psi_2''$.

To compute $\psi_2''$, we note that $\lift{\sat}(\structure',
\pred{in}(\ttime'', 3, 11) \conj \ttime'' \not= 11) = \{\ttime''
\mapsto 3, \ttime'' \mapsto 7\}$. It follows that
$\reduce(\structure', \psi_2') = \psi_2'' =
\dual{\pred{ftr}}(\mbox{Alice}, mr, 3) \conj
\dual{\pred{ftr}}(\mbox{Alice}, mr, 7) \conj \psi_2'''$, where
$\psi_2''' = \forall \ttime'' .~(\pred{in}(\ttime'', 3, 11) \conj
\ttime'' \not= 11 \conj \ttime'' \not \in \{3,7\}) \imp
\dual{\pred{ftr}}(\mbox{Alice}, mr, \ttime'')$. Informally, $\psi_2''$
means that it should have been infeasible to respond to $\mbox{Alice}$
at times $3$ and $7$ (which are the only two observed time points on
$\structure'$ before the response at time $11$), and also at any other
time points between $3$ and $11$ that may show up in extensions of
$\structure'$.

Putting back the various formulae, we have
$\reduce(\structure',\varphi_1) = (\psi_1' \disj \varphi_1') \conj
\varphi_0'$, where $\psi_1' = \pred{contains}(M,\mbox{Alice},mr,11)
\conj \psi_2''$ means that the message $M$ sent to $\mbox{Alice}$ at
time $11$ contain her $mr$ and that it be infeasible to respond
earlier ($\psi_2''$), $\varphi_1'$ allows for the possibility to
satisfy $\mbox{Alice}$'s request through another response before time
$33$, and $\varphi_0'$ enforces the top-level policy on any other
requests.  This is exactly what we might expect from an informal
analysis.  Further, note that the reduction exposes the ground
subjective atoms $\pred{contains}(M,\mbox{Alice},mr,11)$,
$\dual{\pred{ftr}}(\mbox{Alice}, mr, 3)$ and
$\dual{\pred{ftr}}(\mbox{Alice}, mr, 7)$ for a human
auditor to inspect and discharge.
\end{eg}



\subsection{Correctness and Minimality of Enforcement}
\label{sec:correctness:minimality}

The function $\reduce$ is correct in the sense that its input and
output formulas contain the same obligations. Formally, if
$\reduce(\structure, \varphi) = \psi$, then in \emph{all extensions of
  $\structure$}, $\varphi$ is true iff $\psi$ is true and $\varphi$ is
false iff $\psi$ is false.

\begin{thm}[Correctness of $\reduce$]\label{thm:correctness}
If $\reduce(\structure, \varphi) = \psi$ and $\structure' \geq
\structure$, then (1)~$\structure' \models \varphi$ iff $\structure'
\models \psi$ and (2)~$\structure' \models \dual{\varphi}$ iff $\structure'
\models \dual{\psi}$.
\end{thm}
\begin{proof}
See Appendix~\ref{app:enforcement}, Theorem~\ref{thm:correctness:app}.
\end{proof}

The proof of this theorem relies on correctness of $\lift{\sat}$,
which we prove in the next subsection
(Theorem~\ref{thm:sat:correct}). Correctness of iterative enforcement
is an immediate corollary of Theorem~\ref{thm:correctness}. We can
prove by induction on $n$ that if $\varphi_0 \iterate{\structure_1}
\varphi_1 \ldots \iterate{\structure_n} \varphi_n$, then for all
extensions $\structure' \geq \structure_n$, $\structure' \models
\varphi_n$ iff $\structure' \models \varphi_0$ and $\structure'
\models \dual{\varphi_n}$ iff $\structure' \models \dual{\varphi_0}$.

Next, we wish to prove that if $\reduce(\structure, \varphi) = \psi$
then $\psi$ is minimal with respect to $\varphi$ and $\structure$,
i.e., an atom occurs in $\psi$ only if it occurs in $\varphi$ and its
interpretation in $\structure$ is unknown. Unfortunately, owing to
quantification, there is no standard definition of the set of atoms of
a formula of first-order logic. In the following, we provide one
natural definition of the atoms of a formula and characterize
minimality with respect to it; other similar characterizations are
possible. If $\vdash \varphi$, we define the set of atoms of a formula
$\varphi$ with respect to a structure $\structure$ as follows.
\[\begin{array}{lll}
\atoms(\structure, P_S) & = & \{ P_S \} \\

\atoms(\structure, P_O) & = & \{ P_O \} \\

\atoms(\structure, \top) & = & \{ \} \\

\atoms(\structure, \bot) & = & \{ \} \\

\atoms(\structure, \varphi_1 \conj \varphi_2) & = & \atoms(\structure,
\varphi_1) \cup \atoms(\structure, \varphi_2)\\

\atoms(\structure, \varphi_1 \disj \varphi_2) & = & \atoms(\structure,
\varphi_1) \cup \atoms(\structure, \varphi_2)\\

\atoms(\structure, \forall \vec{x}.(c \imp \varphi)) & = &
\bigcup_{\sigma \in \lift{\sat}(\structure, c)} \atoms(\structure,
\varphi\sigma) \\

\atoms(\structure, \exists \vec{x}.(c \conj \varphi)) & = &
\bigcup_{\sigma \in \lift{\sat}(\structure, c)} \atoms(\structure,
\varphi\sigma)

\end{array}\]

The following theorem characterizes minimality of $\reduce$ with
respect to the above definition of atoms in a formula.

\begin{thm}[Minimality]\label{thm:minimality}
Suppose $\vdash \varphi$ and $\reduce(\structure, \varphi) =
\psi$. Then $\atoms(\structure, \psi) \subseteq \atoms(\structure,
\varphi) \cap \{P ~|~ \rho_\structure(P) = \symu\}$.
\end{thm}
\begin{proof}
See Appendix~\ref{app:enforcement}, Theorem~\ref{thm:minimality:app}.
\end{proof}

\begin{eg}\label{eg:minimality}
Revisiting Example~\ref{eg:reduce}, we check that the output produced
by the second reduction satisfies Theorem~\ref{thm:minimality}. Recall
that the second reduction is $\reduce(\structure', \varphi_1) =
(\psi_1' \disj \varphi_1') \conj \varphi_0'$. $\varphi_1'$ and
$\varphi_0'$ each have top-level quantifiers whose guards have no
satisfying instances in $\structure'$, so, by definition of $\atoms$,
$\varphi_1'$ and $\varphi_0'$ have no atoms
w.r.t.\ $\structure'$. Thus we turn to $\psi_1'$. It is easy to check
that $\atoms(\structure', \psi_1')$ is the three element set
$\{\pred{contains}(M, \mbox{Alice}, mr, 11),
\dual{\pred{ftr}}(\mbox{Alice}, mr, 3), \linebreak[6]
\dual{\pred{ftr}}(\mbox{Alice}, mr, 7)\}$. Further, from the analysis
of Example~\ref{eg:reduce}, each of these three atoms also exist in
$\atoms(\structure', \varphi_1)$. Finally, each of the three atoms is
subjective, so each has a valuation $\symu$ in $\structure$.
\end{eg}



\subsection{Quantifier Instantiation and Mode Analysis}
\label{sec:mode}

Having described our main enforcement function $\reduce$, we turn to
the mode analysis relation $\vdash \varphi$ and the function
$\lift{\sat}$ on which the definition of $\reduce$ relies. The rest of
this paper can be understood without understanding this section, so
the disinclined reader may choose to skip it.

\paragraph{Input and Output}
The objective of our mode analysis, as mentioned earlier, is to ensure
that the set of satisfying instances of quantified variables $\vec{x}$
in a restriction $c$ be both finite and computable. Our method of mode
analysis is inspired by, and based on a similar technique in logic
programming (see, e.g.~\cite{apt94:modes}). The key observation in
mode analysis is that, for many predicates, the set of all satisfying
instances on any given structure can be computed finitely if arguments
in certain positions are ground. The reason why instances can be
computed may vary from predicate to predicate; we illustrate some such
computations from prior examples.
\begin{enumerate}
\item Given a ground $m$, the set of $q,t$ such that
  $\pred{tagged}(m,q,t, \ttime)$ holds is finite and can be computed
  from $m$ itself, as we assumed in Example~\ref{eg:past}. (Note that
  the last argument $\ttime$ is an artifact of our translation and is
  irrelevant here.)
\item For an action predicate like $\pred{send}(p_1,p_2,m,\ttime)$, we
  can compute all instances of $p_1$, $p_2$, $m$, $\ttime$ for which
  $\pred{send}(p_1,p_2,m,\ttime)$ holds simply by querying the given
  system log.
\item Given ground $\ttime_2, \ttime_3$, we can compute all $\ttime_1$
  such that $\pred{in}(\ttime_1,\ttime_2,\ttime_3)$ by looking at the
  states in the given system log and selecting the subset that lie in
  the interval $[\ttime_2,\ttime_3]$.
\item Given ground $r$ and $\ttime$, we can compute all principals $p$ such
  that $\pred{inrole}(p,r,\ttime)$ by looking at the roles' database.
\end{enumerate}

Note that in each of the cases 1--4, we require that certain argument
positions be ground (e.g., $m$ in 1 and $\ttime_2,\ttime_3$ in 3), and
compute others (e.g., $q,t$ in 1 and $\ttime_1$ in 3). We call these
the \emph{input} and \emph{output} argument positions, respectively.
Formally, we represent input and output positions by two \emph{partial
  functions} $I$ and $O$ (input and output) from predicates to
$2^{\mathrm{N}}$, which we assume are given to us. The functions are
partial because satisfying instances of certain predicates, including
all subjective predicates, are not computable. Following the earlier
example, we could choose:
\begin{enumerate}
\item $I(\pred{tagged}) = \{1\}$, $O(\pred{tagged}) = \{2,3\}$
\item $I(\pred{send}) = \{\}$, $O(\pred{send}) = \{1,2,3,4\}$
\item $I(\pred{in}) = \{2,3\}$, $O(\pred{in}) = \{1\}$
\item $I(\pred{inrole}) = \{2,3\}$, $O(\pred{inrole}) = \{1\}$
\end{enumerate}

For a subjective predicate $p_S$, $I(p_S)$ and $O(p_S)$ are
undefined. The sets $I(p)$ and $O(p)$ are called a moding of predicate
$p$. If $i \in I(p)$ ($i \in O(p)$), we say that the $i$th argument of
$p$ is in input (output) mode. Certain arguments may be in neither
input nor output mode, e.g., argument $4$ of the predicate
$\pred{tagged}$. Also, the same predicate may be moded in multiple
ways. For example, both the assignments ($I(\pred{send}) = \{\}$,
$O(\pred{send}) = \{1,2,3,4\}$) and ($I(\pred{send}) = \{1\}$,
$O(\pred{send}) = \{2,3,4\}$) are correct. However, it suffices to
assume that each predicate has a unique moding, because we can use
different names for predicates with the same interpretation but
different modings.


\paragraph{Substitution Computation}
A substitution $\sigma$ is a finite map from variables to ground
terms. Say that a substitution $\sigma'$ extends a substitution
$\sigma$, written $\sigma' \geq \sigma$, if $\dom(\sigma') \supseteq
\dom(\sigma)$ and for all $x \in \dom(\sigma)$, $\sigma(x) =
\sigma'(x)$. We abstract the computation of terms in output positions
from terms in input positions as a \emph{partial computable function}
$\sat$. The input of the function is a pair containing a structure and
an atom; its output is a finite \emph{set} of substitutions. The
function $\sat$ satisfies the following condition:
\begin{quote}
Given a structure $\structure$ and
an atom $p(t_1,\ldots,t_n)$ such that for all $i \in I(p)$, $t_i$ is
ground, $\sat(\structure, p(t_1,\ldots,t_n))$ is the set of all
substitutions for variables in $\bigcup_{i \in O(p)} t_i$ that have
extensions $\sigma$ such that $\structure \models
p(t_1,\ldots,t_n)\sigma$. 
\end{quote}
For example, if in structure $\structure$, principal $\mbox{Charlie}$
has doctors $\mbox{Alice}$ and $\mbox{Bob}$ at time $\ttime$, then
$\sat(\structure, \pred{inrole}(p, \pred{doc}(\mbox{Charlie}),
\ttime))$ would be the two element set $\{p \mapsto \mbox{Alice}, p
\mapsto \mbox{Bob} \}$. If the input arguments in atom $P$ are not
ground, then $\sat(\structure, P)$ may be undefined. For example, if
either $\tau_2$ or $\tau_3$ is not ground, then $\sat(\structure,
\pred{in}(\tau_1, \tau_2, \tau_3))$ is undefined. Because subjective
predicates are not computable, $\sat(\structure, P_S)$ is also
undefined for every subjective atom $P_S$. In practice, the function
$\sat(\structure, P)$ could be implemented through queries to the
database that stores the audit log.


We lift the function $\sat$ to the function $\lift{\sat}$ that
computes satisfying instances of restrictions. The specification of
the lifted function $\lift{\sat}(\structure, c)$ is similar to that of
$\sat$: Given a partially ground restriction $c$,
$\lift{\sat}(\structure, c)$ is a finite set of substitutions
characterizing all satisfying instances of $c$. 
\[\begin{array}{lll}
\lift{\sat}(\structure, p_O(t_1,\ldots,t_n)) & = & \sat(\structure,
p_O(t_1,\ldots,t_n)) \\

\lift{\sat}(\structure, \top) & = & \{\bullet\} \\

\lift{\sat}(\structure, \bot) & = & \{\} \\

\lift{\sat}(\structure, c_1 \conj c_2) & = & \bigcup_{\sigma \in
\lift{\sat}(\structure, c_1) } \sigma + \lift{\sat}(\structure, c_2
\sigma)\\

\lift{\sat}(\structure, c_1 \disj c_2) & = & \lift{\sat}(\structure,
c_1) \cup \lift{\sat}(\structure, c_2)\\

\lift{\sat}(\structure, \exists x. c) & = & \lift{\sat}(\structure, c)
\backslash \{x\} ~~~~~ \mbox{($x$ fresh)}
\end{array}\]

For atoms, the definition of $\lift{\sat}$ coincides with that of
$\sat$. Since $\top$ must always be true, $\lift{\sat}(\structure,
\top)$ contains only the empty substitution (denoted $\bullet$). Since
$\bot$ can never be satisfied, $\lift{\sat}(\structure, \bot)$ is
empty. For $c_1 \conj c_2$, the set of satisfying instances is
obtained by taking those of $c_1$ (denoted $\sigma$ above), and
conjoining those with satisfying instances of $c_2 \sigma$ (the
operation $+$ is composition of substitutions with disjoint
domains). The set of satisfying instances of $c_1 \disj c_2$ is the
union of the satisfying instances of $c_1$ and $c_2$. Satisfying
instances of $\exists x.c$ are obtained by taking those of $c$, and
removing the substitutions for $x$.

$\lift{\sat}$ is a partial function because the underlying function
$\sat$ is partial. For instance, taking an example from
Section~\ref{sec:logic}, $\lift{\sat}(\structure,
\pred{send}(p_1,p_2,m,\ttime) \conj \pred{tagged}(m',q,t,\ttime'))$ is
undefined if $m'$ is a variable because any substitution $\sigma$ in
the output of the recursive call $\lift{\sat}(\structure,
\pred{send}(p_1,p_2,m,\ttime))$ will not contain $m'$ in its domain
and, therefore, in the call to $\lift{\sat}(\structure,
\pred{tagged}(m',q,t,\ttime')\sigma)$, the first argument to
$\pred{tagged}$ will be non-ground. Since $I(\pred{tagged}) = \{1\}$,
this recursive call may fail to return an answer. On the other hand,
$\lift{\sat}(\structure, \pred{send}(p_1,p_2,m,\ttime) \conj
\pred{tagged}(m,q,t,\ttime'))$ is defined because the first argument
of $\pred{tagged}$ in the second recursive call is $m$, which is
grounded by the substitution $\sigma$ of the first recursive
call. Despite being partial, $\lift{\sat}(\structure, c)$ represents
all satisfying instances of $c$, whenever it is defined, as formalized
by the following theorem.

\begin{thm}[Correctness of $\lift{\sat}$]\label{thm:sat:correct}
If $\lift{\sat}(\structure, c)$ is defined then for any substitution
$\sigma'$ with $\dom(\sigma') \supseteq \fv(c)$, $\structure \models
c\sigma'$ iff there is a substitution $\sigma
\in\lift{\sat}(\structure, c)$ such that $\sigma' \geq \sigma$.
\end{thm}
\begin{proof}
See Appendix~\ref{app:enforcement}, Theorem~\ref{thm:sat:correct:app}.
\end{proof}

\begin{eg}\label{eg:lift:sat}
In Example~\ref{eg:reduce}, we informally evaluated $\lift{\sat}$ at
several places. Here, we justify the first two evaluations. In the
first instance, we said that $\lift{\sat}(\structure,
\pred{in}(\ttime,0,\infty) \conj \pred{req}(p,t,\ttime)) =
\{(\ttime,p,t) \mapsto (3,\mbox{Alice},mr)\}$. This follows from the
observation that from the information in the structure $\structure$,
we must have $\sat(\structure, \pred{in}(\ttime,0,\infty)) = \{\ttime
\mapsto 1,\ttime \mapsto 3, \ttime \mapsto 7\}$, $\sat(\structure,
\pred{req}(p,t,3)) = \{(p,t) \mapsto (\mbox{Alice}, mr)\}$ and
$\sat(\structure, \pred{req}(p,t,\ttime)) = \{\}$ for $\ttime \not =
3$. The result of applying $\lift{\sat}$ follows from its definition.

Similarly, we calculated that $\lift{\sat}(\structure',
\pred{in}(\ttime',3,33) \conj \pred{inrole}(q, \mbox{records},\ttime')
\conj \pred{send}(q, \mbox{Alice}, m,\ttime'))\linebreak[6] =
\{(\ttime',q,m) \mapsto (11,\mbox{Bob},M)\}$. This follows because,
from the description of $\structure'$, $\sat(\structure',
\pred{in}(\ttime', 3, 33))\linebreak[6] = \{\ttime' \mapsto 3, \ttime'
\mapsto 7, \ttime' \mapsto 11\}$, $\sat(\structure',
\pred{inrole}(q,\mbox{records},T)) = \{q \mapsto \mbox{Bob}\}$ for $T
= 11$ and $\{\}$ otherwise, and $\sat(\structure',
\pred{send}(q,p,m,\ttime')) = \{(q,p,m,\ttime') \mapsto (\mbox{Bob},
\mbox{Alice}, M, 11)\}$.
\end{eg}

\begin{figure}
\framebox{$\chi_I \vdash c : \chi_O$}
\begin{mathpar}
\inferrule{\forall k \in I(p_O). ~\fv(t_k) \subseteq \chi_I
  \\ \chi_O = \chi_I \cup (\bigcup_{j \in O(p_O)} \fv(t_j))}{
  \chi_I \vdash p_O(t_1,\ldots,t_n) : \chi_O}\and
\inferrule{ }{ \chi_I \vdash \top: \chi_I} \and
\inferrule{ }{ \chi_I \vdash \bot: \chi_I} \and
\inferrule{ \chi_I \vdash c_1 : \chi \\ \chi \vdash c_2:
  \chi_O}{ \chi_I \vdash c_1 \conj c_2: \chi_O}\and
\inferrule{ \chi_I \vdash c_1: \chi_1 \\  \chi_I
  \vdash c_2: \chi_2}{ \chi_I \vdash c_1 \disj
  c_2: \chi_1 \cap \chi_2}\and
\inferrule{\chi_I \vdash c: \chi_O}{
  \chi_I \vdash \exists x.c : \chi_O \backslash
  \{x\}}\and
\end{mathpar}

\framebox{$\chi \vdash \varphi$}
\begin{mathpar}
\inferrule{\forall k.~\fv(t_k) \subseteq \chi}{
  \chi \vdash p(t_1,\ldots,t_k)}\and
\inferrule{ }{ \chi \vdash \top} \and
\inferrule{ }{ \chi \vdash \bot} \and
\inferrule{ \chi \vdash \varphi_1 \\ \chi \vdash \varphi_2}{ \chi
  \vdash \varphi_1 \conj \varphi_2}\and
\inferrule{ \chi \vdash \varphi_1 \\ \chi \vdash \varphi_2}{ \chi
  \vdash \varphi_1 \disj \varphi_2}\and
\inferrule{ \chi \vdash c: \chi_O \\ \vec{x} \subseteq \chi_O \\ \fv(c)
  \subseteq \chi \cup \vec{x} \\ \chi_O \vdash \varphi}{
  \chi \vdash \forall \vec{x}. (c \imp \varphi)}\and
\inferrule{\chi \vdash c: \chi_O \\ \vec{x} \subseteq \chi_O \\ \fv(c)
  \subseteq \chi \cup \vec{x} \\ \chi_O \vdash \varphi}{
  \chi \vdash \exists \vec{x}.(c \conj \varphi)}\and
\end{mathpar}

(In the rules for quantifiers, bound variables $x$ or $\vec{x}$ must
be renamed so that they are fresh.)
\caption{Moding Rules}
\label{fig:moding}
\end{figure}

\paragraph{Mode Analysis}
Next, we define a static check of restrictions to rule out those on
which $\lift{\sat}$ is not defined, e.g.,
$\pred{send}(p_1,p_2,m,\ttime) \conj \pred{tagged}(m',q,t,\ttime')$
described earlier. This static check is what we call the mode
analysis. A restriction that passes the check is called
\emph{well-moded}.  Formally, we define well-modedness as a relation
$\chi_I \vdash c: \chi_O$, where $\chi_I$ and $\chi_O$ are sets of
variables. If the relation holds, then for any $\sigma$ with
$\dom(\sigma) \supseteq \chi_I$ and any $\structure$,
$\lift{\sat}(\structure, c\sigma)$ is defined and, further, any
substitution in it contains all of $\chi_O \backslash \chi_I$ in its
domain. ($\chi_I$ and $\chi_O$ are analogues of inputs and outputs for
restrictions.)

The relation $\chi_I \vdash c: \chi_O$ is defined by the rules of
Figure~\ref{fig:moding}, which also constitute a linear-time decision
procedure for deciding the relation (with inputs $c$ and $\chi_I$ and
output $\chi_O$). We explain some of the rules. An atom
$p(t_1,\ldots,t_k)$ is well-moded if the free variables (abbreviated
$\fv$) of input positions are ground (premise $\forall k \in
I(p_O). ~\fv(t_k) \subseteq \chi_I$ of the first rule) and the output
$\chi_O$ equals $\chi_I$ (which is already ground) unioned with
$\bigcup_{j \in O(p_O)} \fv(t_j)$ (all of which must be in the domain
of $\sat(\structure, p(t_1,\ldots,t_n))$). The rule for conjunctions
$c_1 \conj c_2$ chains the outputs $\chi$ of $c_1$ into the inputs of
$c_2$. The following theorem establishes that $\lift{\sat}$ is total
on well-moded restrictions and also establishes the relation between
$\chi_I$, $\chi_O$ and the substitutions in the output of
$\lift{\sat}$.

\begin{thm}[Totality of $\lift{\sat}$]\label{thm:sat:total}
If $\chi_I \vdash c: \chi_O$, then for all structures $\structure$ and
all substitutions $\sigma$ with $\dom(\sigma) \supseteq \chi_I$,
$\lift{\sat}(\structure, c\sigma)$ is defined and, further, for each
substitution $\sigma' \in \lift{\sat}(\structure, c\sigma)$, $\chi_I \cup
\dom(\sigma') \supseteq \chi_O$.
\end{thm}
\begin{proof}
See Appendix~\ref{app:enforcement},
Theorem~\ref{thm:sat:total:app}.
\end{proof}

We extend the mode-check on restrictions to formulas $\varphi$ of the
sublogic. The objective of this mode-check is two-fold. First, the
check ensures that all restrictions occurring in $\varphi$ are
well-moded in the sense described above. Second, for quantifiers
$\forall \vec{x}. (c \imp \varphi')$ and $\exists \vec{x}. (c \conj
\varphi')$, the check ensures that the quantified variables $\vec{x}$
are contained in the \emph{outputs} ($\chi_O$) of the restriction
$c$. (Hence, by Theorems~\ref{thm:sat:correct}
and~\ref{thm:sat:total}, any substitution in
$\lift{\sat}(\structure,c)$ grounds $\vec{x}$, which is central to the
termination of $\reduce$.) The mode-check is formalized as the
relation $\chi \vdash \varphi$, meaning that for any substitution
$\sigma$ with $\dom(\sigma) \supseteq \chi$, the formula $\varphi
\sigma$ is well-moded. Its straightforward rules are shown in
Figure~\ref{fig:moding}. The rules constitute a linear-time decision
procedure for checking the relation (with inputs $\chi$ and
$\varphi$). In the rules for $\forall \vec{x}. (c \imp \varphi')$ and
$\exists \vec{x}. (c \conj \varphi')$, the first premises check that
$c$ is well-moded. The second premises ensure that the variables
$\vec{x}$ are contained in the output $\chi_O$ of the mode check on
$c$. The third premises ensure that $c$ is closed. It can easily be
checked that if $\chi \vdash \varphi$, then $\fv(\varphi) \subseteq
\chi$.

We call a formula $\varphi$ well-moded if $\{\} \vdash \varphi$, which
we abbreviate to $\vdash \varphi$. The following theorem shows that on
well-moded formulas, the function $\reduce$ is total. Further on a
well-moded input, the output is also well-moded (so the output can
used as input in a subsequent iteration).

\begin{thm}[Totality of $\reduce$]\label{thm:reduce:total}
If $\vdash \varphi$ then there is a $\psi$ such that
$\reduce(\structure,\varphi) = \psi$ and $ \vdash \psi$.
\end{thm}
\begin{proof}
See Appendix~\ref{app:enforcement},
Theorem~\ref{thm:reduce:total:app}.
\end{proof}

\begin{eg}
It can easily be checked that the formulas $\globally \alpha_{pol1}$
and $\globally \alpha_{pol2}$ defined in
Example~\ref{eg:policies:translated} are all well-moded (e.g.,~$\vdash
{\globally \alpha_{pol1}}$) using the definitions of $I$ and $O$
presented at the beginning of this subsection.
\end{eg}


\section{Specific Instances of Enforcement}
\label{sec:instances}

We analyze the behavior of our enforcement algorithm on two restricted
classes of structures.  First, we consider \emph{objectively-complete}
structures -- those that map every objective atom to either $\symt$ or
$\symf$ (Section~\ref{sec:instance:oc}). We show that for such
structures $\structure$, the output of $\reduce(\structure,\varphi)$
can be simplified to conjunctions and disjunctions of ground
subjective atoms through trivial rewriting (e.g., replacing $\top
\conj \psi$ with $\psi$), thus making it more amenable to human
inspection. We also obtain a decision procedure to decide the truth
and falsity of input formulas without subjective predicates.

Second, we consider \emph{past-complete} structures, those that have
complete information up to a specific point of time
(Section~\ref{sec:instance:pc}). This corresponds to the standard
assumption in every existing work on enforcement of temporal
properties that the audit log contains all past information. In
particular, we show that on past-complete traces, our algorithm yields
a method to find violations of safety
properties~\cite{alpern87:safety} and satisfactions of co-safety
properties~\cite{bauer10:rm} at the earliest.



\subsection{Execution on Objectively-Complete Structures}
\label{sec:instance:oc}

We analyze the output of $\reduce(\structure, \varphi)$ when
$\structure$ is objectively-complete. Although
objective-\linebreak[6]completeness requires that truth and falsity of
objective atoms be determined even in the future, it may model some
realistic settings. For instance, after audit-relevant information has
been gathered from all possible sources, it may be assumed that any
fact not explicitly seen is, by default, false. The resulting
structure would be objectively-complete. Objectively-complete
structures correspond to the case of subjective incompleteness from
Section~\ref{sec:structures}.

\begin{defn}
A structure $\structure$ is called objectively-complete if for all
objective atoms $P_O$, $\rho_\structure(P_O) \in \{\symt, \symf\}$.
\end{defn}

If a structure $\structure$ is objectively-complete, then during the
execution of $\reduce(\structure, \varphi)$, \emph{all} relevant
substitutions can be found for quantifiers and \emph{all} objective
atoms can be replaced with either $\top$ or $\bot$. Indeed, we show in
this subsection that if $\structure$ is objectively-complete, then the
output, $\psi$, of $\reduce(\structure, \varphi)$ can be
\emph{rewritten} (using straightforward rewrite rules) to a logically
equivalent formula that is either $\top$ or $\bot$ or contains only
subjective atoms, conjunctions and disjunctions. This has practical
importance because, as compared to a formula with quantifiers, a
formula containing only subjective atoms, conjunctions and
disjunctions is more amenable to human inspection and audit.

There are two kinds of rewriting we need to perform on the output
$\psi$ to reduce it to our desired form. First, we need to eliminate
unnecessary occurrences of $\top$ and $\bot$ that arise either from
occurrences of $\top$ and $\bot$ in the input formula, or as
replacements of atoms that evaluate to $\symt$ and $\symf$
respectively. Such occurrences can be eliminated by repeatedly
applying the following eight rewriting rules anywhere in the output:
\[\begin{array}{l@{\hspace{20mm}}l}
\psi \conj \top \rewrite \psi & \top \conj \psi \rewrite \psi \\

\psi \conj \bot \rewrite \bot & \bot \conj \psi \rewrite \bot \\

\psi \disj \top \rewrite \top & \top \disj \psi \rewrite \top \\

\psi \disj \bot \rewrite \psi & \bot \disj \psi \rewrite \psi \\
\end{array}\]

For example, if $\varphi = P_O \conj P_S$ for an objective atom $P_O$
and a subjective atom $P_S$ and $\rho_\structure(P_O) = \symt$, then
$\reduce(\structure,\varphi) = \top \conj P_S$. This can be simplified
to $P_S$ using the second rule above. Note that each rule above
preserves logical equivalence of formulas.

Second, we need to eliminate those quantified subformulas in the
output that are called $\psi'$ in the definition of $\reduce$
(Figure~\ref{fig:reduce}). These have the forms $\forall \vec{x}. ((c
\conj x \not\in S) \imp \varphi)$ and $\exists \vec{x}. ((c \conj x
\not\in S) \conj \varphi)$. Because $S$ contains all instances of
$\vec{x}$ that satisfy $c$, $(c \conj x \not \in S)$ has no satisfying
instances in $\structure$, i.e., $\lift{\sat}(\structure, (c \conj x
\not \in S)) = \{\}$. Further, because $\structure$ is
objectively-complete, any extension $\structure'$ of $\structure$ must
agree with $\structure$ on valuation of objective atoms, so, by
Theorem~\ref{thm:sat:correct}, $\lift{\sat}(\structure', (c \conj x
\not \in S)) = \{\}$. Consequently, $\forall \vec{x}. ((c \conj x
\not\in S) \imp \varphi)$ is logically equivalent to $\top$ in all
extensions of $\structure$ and $\exists \vec{x}. ((c \disj x \not\in
S) \imp \varphi)$ is logically equivalent to $\bot$ in all extensions
of $\structure$. This immediately yields the following two rules for
elimination of quantifiers from the output of $\reduce$.
\[\begin{array}{l@{\hspace{20mm}}l}
\forall \vec{x}. (c \imp \varphi) \rewrite \top &

\exists \vec{x}. (c \conj \varphi) \rewrite \bot
\end{array}\]
We point out that, unlike the eight rewriting rules presented earlier,
the two rewriting rules above do not preserve logical equivalence in
general, but they preserve logical equivalence when applied to the
output $\psi = \reduce(\structure, \varphi)$ for objectively-complete
$\structure$.

Let $\rewrite^*$ denote the reflexive-transitive closure of
$\rewrite$. Since $\rewrite$ makes formulas strictly smaller, it
cannot be applied indefinitely to any formula.  Further, even though a
formula may be rewritten in many ways using a single application of
$\rewrite$, the formula obtained by applying $\rewrite$ exhaustively
starting from a fixed initial formula is unique because $\rewrite$ is
confluent.

\begin{thm}\label{thm:oc}
Suppose $\structure$ is objectively-complete, $\vdash \varphi$ and
$\psi = \reduce(\structure, \varphi)$. Then $\psi \rewrite^* \psi'$,
where (1)~$\psi'$ is either $\top$, or $\bot$, or contains only
subjective atoms and the connectives $\conj$, $\disj$, and (2)~For all
$\structure' \geq \structure$, $\structure' \models \psi$ iff
$\structure' \models \psi'$ and $\structure' \models \dual{\psi}$ iff
$\structure' \models \dual{\psi'}$.
\end{thm}
\begin{proof}
See Appendix~\ref{app:instances}, Theorem~\ref{thm:oc:app}.
\end{proof}

An interesting special case arises on inputs $\varphi$ without any
subjective predicates. In this case, it can be proved by induction on
$\varphi$ that if $\structure$ is objectively-complete, then either
$\structure \models \varphi$ or $\structure \models \dual{\varphi}$
(either $\varphi$ is true in $\structure$ or it is
false). Interestingly, for such inputs, Theorem~\ref{thm:oc} yields a
\emph{decision procedure} for determining the truth or falsity of
$\varphi$ in $\structure$. The proof of this fact is
straightforward. By minimality of $\reduce$
(Theorem~\ref{thm:minimality}), the output $\psi$ of
$\reduce(\structure, \varphi)$ cannot contain any subjective atoms if
$\varphi$ does not contain them, so neither can the formula $\psi'$
obtained by rewriting in Theorem~\ref{thm:oc}. Hence, $\psi'$ must be
either $\top$ or $\bot$. If $\psi' = \top$, then by
Theorem~\ref{thm:correctness}, $\structure \models \varphi$, and if
$\psi' = \bot$, then by the same theorem, $\structure \models
\dual{\varphi}$. This is a decision procedure because both $\reduce$
and $\rewrite^*$ terminate.

\subsection{Execution on Past-Complete Structures}
\label{sec:instance:pc}

Next, we analyze our enforcement algorithm on structures that have
complete information up to a specific point of time, say
$\ttime_0$. We call such structures $\ttime_0$-past-complete or,
briefly, $\ttime_0$-complete. Past-completeness corresponds to future
incompleteness from Section~\ref{sec:structures} and is practically
relevant because in many cases, audit logs record all relevant events
as they happen and the \emph{entire} history is available to an
enforcement algorithm. In fact, this is a standard assumption in all
existing literature on either runtime or post-hoc enforcement of
temporal properties. The classic result in this context is that, under
this assumption, a runtime monitor can detect both violation of
so-called safety properties (a given bad event never happens) and
satisfaction of so-called co-safety properties (a given good event
happens at some time either in the past or in the future) at the
earliest possible time. In the rest of this subsection, we show that
on past-complete structures similar results hold for our enforcement
method.

We start by formally defining past-complete structures, then adapt a
standard characterization of safety and co-safety properties in
temporal logic to our setting, and finally prove that the function
$\reduce$, together with rewriting $\rewrite$, yields a method to
enforce both safety and co-safety properties. It is important to
mention here that violation or satisfaction of a property cannot be
defined formally if the property has subjective
predicates. Consequently, we assume in this subsection, like existing
literature on the subject, that policies do not contain subjective
predicates.

\begin{defn}
Given a ground time $\ttime_0$, a structure $\structure$ is called
$\ttime_0$-past-complete or $\ttime_0$-complete if the following two
conditions hold:
\begin{enumerate}
\item For all predicates $p$, all ground $t_1,\ldots,t_n$ and all
  $\ttime \leq \ttime_0$, $\rho_\structure(p(t_1,\ldots,t_n,\ttime))
  \in \{\symt,\symf\}$.
\item For all ground $\ttime_1,\ttime_2,\ttime_3$ such that $\ttime_1
  \leq \ttime_0$,
  $\rho_\structure(\pred{in}(\ttime_1,\ttime_2,\ttime_3)) \in
  \{\symt,\symf\}$.
\end{enumerate}
\end{defn}
The first condition means that the truth or falsity of every atom in
the temporal logic can be determined at time $\ttime$ if $\ttime \leq
\ttime_0$. The second condition states that $\structure$ records all
relevant states up to time $\ttime_0$.

\paragraph{Safety and Co-safety}
Informally, a safety property states that a specified bad condition is
never satisfied. Dually, a co-safety property states that a specified
good condition is satisfied at some time (either in the past or in the
future). Although the two kinds of properties are often characterized
in terms of traces (semantically)~\cite{alpern87:safety,bauer10:rm},
characterizations of the two kinds of properties as classes of
formulas in logic are more relevant for us.  It is
known~\cite{manna95:temporal} that safety properties correspond to
formulas of the form $\globally \alpha_p$, where $\globally$ is the
``in every state'' operator introduced in
Example~\ref{eg:policies:translated} and $\alpha_p$ is an arbitrary
formula of the temporal logic not containing any future operators
($\boxp$ and $\until$). In words, $\globally \alpha_p$ means that in
every state (the bad condition) $\neg\alpha_p$ does not hold. As an
illustration, the policy $\globally \alpha_{pol1}$ in
Example~\ref{eg:policies:translated} is a safety property, but
$\globally \alpha_{pol2}$ is not because it contains a future
operator. Dually, co-safety properties can be characterized as
formulas of the form $\eventually \alpha_p = \exists
\ttime. (\pred{in}(\ttime,0,\infty) \conj \trans{\ttime}{\alpha_p})$,
informally meaning that in some state $\ttime$, (the good condition)
$\alpha_p$ holds.\footnote{We have not seen this characterization of
  co-safety properties in literature, but it is easily derived as the
  dual of the known characterization of safety properties.}

We say that a safety property $\globally \alpha_p$ is \emph{violated}
at time $\ttime$ in a structure $\structure$ if $\structure \models
\dual{\trans{\ttime}{\alpha_p}}$. In other words, $\globally \alpha_p$
is violated at time $\ttime$ if at that time, the negation of
$\alpha_p$ holds in $\structure$. Similarly, we say that a co-safety
property $\eventually \alpha_p$ is \emph{satisfied} at time $\ttime$
in a structure $\structure$ if $\structure \models
\trans{\ttime}{\alpha_p}$.

Our first result (Theorem~\ref{thm:safety}) is that if a safety
property $\globally \alpha_p$ is violated at time $\ttime$ in a
structure $\structure$ that is $\ttime_0$-complete ($\ttime \leq
\ttime_0$), then $\reduce(\structure, \globally \alpha_p) \rewrite^*
\bot$ (and conversely). This result is important because it implies
that violations of safety properties can be detected in the \emph{next
  iteration of enforcement after they occur} if audit logs contain all
past information. An analogous result -- Theorem~\ref{thm:cosafety} --
holds for co-safety properties, wherein satisfaction can be detected
at the earliest. The justification for both theorems is similar to
that for Theorem~\ref{thm:oc}, but more involved. Because both
$\reduce$ and $\rewrite^*$ terminate, the theorems also provide
decision procedures for enforcing safety and co-safety properties on
past-complete structures.

\begin{thm}[Enforcement of safety properties]\label{thm:safety}
Suppose ${\globally \alpha_p}$ is a safety property, $\vdash
{\globally \alpha_p}$, $\structure$ is $\ttime_0$-complete, and for
all $\ttime$, $(\rho_\structure(\pred{in}(\ttime,0,\infty)) = \symt)
\Rightarrow \ttime \leq \ttime_0$. Then, $\reduce(\structure,
\globally \alpha_p) \rewrite^* \bot$ iff there is a $\ttime$ such that
$\structure \models \pred{in}(\ttime,0,\ttime_0)$ and $\structure
\models \dual{\trans{\ttime}{\alpha_p}}$.
\end{thm}
\begin{proof}
See Appendix~\ref{app:instances}, Theorem~\ref{thm:safety:app}.
\end{proof}

\begin{thm}[Enforcement of co-safety properties]\label{thm:cosafety}
Suppose ${\eventually \alpha_p}$ is a co-safety property, $\vdash
{\eventually \alpha_p}$, $\structure$ is $\ttime_0$-complete, and for
all $\ttime$, $(\rho_\structure(\pred{in}(\ttime,0,\infty)) = \symt)
\Rightarrow \ttime \leq \ttime_0$. Then, $\reduce(\structure,
\eventually \alpha_p) \rewrite^* \top$ if and only if there is a
$\ttime$ such that $\structure \models \pred{in}(\ttime,0,\ttime_0)$
and $\structure \models \trans{\ttime}{\alpha_p}$.
\end{thm}
\begin{proof}
See Appendix~\ref{app:instances}, Theorem~\ref{thm:cosafety:app}.
\end{proof}

\begin{eg}\label{eg:safety}
We check Theorem~\ref{thm:safety} on the safety property $\globally
\alpha_{pol1}$ from Example~\ref{eg:policies:translated}. The policy
states that if a message $m$ is sent by $p_1$ to $p_2$ for purpose $u$
and the message is tagged as containing $q$'s data about attribute $t$
(which is a form of $phi$), then either the recipient $p_2$ is $q$'s
doctor and the purpose $u$ is treatment, or $q$ has previously
consented to this message transmission.

We consider a simple structure $\structure$ in which this policy is
violated. $\structure$ has only one time point $7$, at which principal
$\mbox{A}$ sends principal $\mbox{B}$ a message $M$. The message $M$
is labeled with purpose $test$ ($\pred{purp\_in}(test,treatment)$
holds) and tagged as containing principal $\mbox{C}$'s information
about attribute $meds$ (medications), which is a form of
$phi$. Further, $\mbox{B}$, the recipient, is not $\mbox{C}$'s
doctor. Suppose that we audit at a later point of time ($10$) and that
$\structure$ described above is $10$-complete. Since there is no other
information in $\structure$ besides what has been mentioned,
$\mbox{C}$ has not consented explicitly to this message transmission,
so the policy has been violated at time $7$. We seek to verify that
$\reduce(\structure, \globally \alpha_{pol1}) \rewrite^* \bot$.

We start by computing $\reduce(\structure, \globally
\alpha_{pol1})$. The reader is advised to revisit the definition of
$\globally \alpha_{pol1}$ in Example~\ref{eg:policies:translated}. At
the top-level, $\globally \alpha_{pol1}$ contains a universal
quantifier with restriction $c = (\pred{in}(\ttime,0,\infty) \conj
\pred{send}(p_1, p_2, m,\ttime) \conj \pred{purp}(m,u,\ttime)
\conj\pred{tagged}(m,q,t,\ttime) \conj \pred{attr\_in}(t,
\attr{phi},\ttime))$. Computing $\lift{\sat}(\structure, c)$ yields
$\{(\ttime, p_1, p_2, m, u, q, t) \mapsto (7, \mbox{A}, \mbox{B}, M,
test, \mbox{C}, meds)\}$. Hence, $\reduce(\structure, \globally
\alpha_{pol1})\linebreak[6] = \reduce(\structure, \varphi_1) \conj
\varphi_0'$, where $\varphi_1$ is shown below and $\varphi_0'$ is
almost a copy of the original policy, with a larger restriction. The
only aspect of $\varphi_0'$ relevant for this example is that it
contains a top-level universal quantifier.
\begin{tabbing}
$\varphi_1$ = \=$($\=$\pred{inrole}(\mbox{B},
  \pred{doc}(\mbox{C}),7) \conj$\\
\>\>$  \pred{purp\_in}(test,\purpose{treatment},7)) \disj$\\
\>$(\exists \ttime'.~ ($\=$ \pred{in}(\ttime',0,7) \conj$\\
\>\>$\pred{consents}(\mbox{C}, \pred{sendaction}(\mbox{A}, \mbox{B}, (\mbox{C},meds)),\ttime')))$
\end{tabbing} 
Next, we calculate $\reduce(\structure, \varphi_1)$. Since
$\rho_\structure(\pred{inrole}(\mbox{B}, \pred{doc}(\mbox{C}),7)) =
\symf$ and\linebreak[6]
$\rho_\structure(\pred{purp\_in}(test,\purpose{treatment},7)) =
\symt$, $\reduce(\structure, \varphi_1) = (\bot \conj \top) \disj
\reduce(\structure, \varphi_2)$, where $\varphi_2$ is the second
disjunct of $\varphi_1$.  Finally, we compute $\reduce(\structure,
\varphi_2)$. The top-level connective of $\varphi_2$ is an existential
quantifier restricted by $\pred{in}(\ttime',0,7)$. Since
$\lift{\sat}(\structure, \pred{in}(\ttime',0,7)) = \{\ttime' \mapsto
7\}$, $\reduce(\structure, \varphi_2) = \reduce(\structure, \varphi_3)
\disj \varphi_2'$, where $\varphi_3 = \pred{consents}(\mbox{C},
\pred{sendaction}(\mbox{A}, \mbox{B}, (\mbox{C},meds)),7)$ and
$\varphi_2'$ begins with an existential quantifier. Clearly,
$\reduce(\structure, \varphi_3) = \bot$. Putting the pieces back
together, we get $\reduce(\structure, \globally \alpha_{pol1}) =
((\bot \conj \top) \disj (\bot \disj \varphi_2')) \conj \varphi_0'$.

Since $\varphi_0'$ and $\varphi_2'$ begin with a universal and an
existential quantifier, they can be rewritten to $\top$ and $\bot$
respectively. So, $\reduce(\structure, \globally \alpha_{pol1})
\rewrite ((\bot \conj \top) \disj (\bot \disj \bot)) \conj \top$,
which can easily be rewritten to $\bot$, thus indicating a
violation. If we change the example to avoid a violation, say by
setting $\rho_\structure(\pred{inrole}(\mbox{B},
\pred{doc}(\mbox{C}),7))$ to $\symt$ instead of $\symf$, then the
result of rewriting changes from $\bot$ to $\top$, indicating a lack
of violation thus far. Finally, if we do not assume that $\structure$
is past-complete, then the rewriting of $\varphi_2'$ to $\bot$ is
unsound because there may be an extension of $\structure$ in which
$\varphi_2'$ is true and, hence, the original property may not have
been violated, but our procedure would conclude that it is. So,
past-completeness is a necessary assumption in
Theorem~\ref{thm:safety} (and also Theorem~\ref{thm:cosafety}).
\end{eg}

\section{Application to HIPAA}
\label{sec:hipaa-case}

We comment on application of our algorithm to transmission-relevant
clauses of the HIPAA Privacy Rule. These clauses can be viewed as a
template for actual privacy policies, which may be obtained by
instantiating abstract roles like ``covered entity'' in HIPAA with
actual roles like ``doctor'', ``nurse'', etc. In prior work on
PrivacyLFP~\cite{DeYoung+10:privacylfp:wpes}, we have shown that all
84 transmission-related clauses in HIPAA can be represented in the
logic. Since we have restricted the syntax of quantifiers in this
paper to facilitate enforcement, an immediate question is whether we
can still represent all the clauses of HIPAA in our logic. A careful
re-analysis of the prior work reveals that 81 of the 84 clauses fall
in the fragment considered in this paper. The three remaining clauses,
namely Sections 164.506(c)(4), 164.512(k)(1)(i), and 164.512(k)(1)(iv)
of HIPAA, contain quantifiers with subjective restrictions. However,
in each such case, the formula under the quantifier contains only
subjective predicates and, therefore, the entire formula may be
considered a single subjective predicate. With this minor change, the
algorithm of Section~\ref{sec:enforcement} can be applied to all 84
clauses of HIPAA.

The next question is the usefulness of the algorithm, given that HIPAA
contains many subjective predicates (in fact, 578 out of a total of
881 atoms in our formalization of HIPAA are subjective). The answer to
this question is two-fold. First, irrespective of the percentage of
subjective atoms, one practical advantage of using our algorithm is
that it instantiates quantifiers automatically using log data, which
could otherwise be a daunting task for a human auditor.

Second, our algorithm automatically discharges objective atoms from
fully instantiated formulas, leaving only subjective atoms for a human
auditor. As discussed in the prior work, with a slight amount of
design effort, e.g., standardizing message formats, 402 of the
subjective atoms can be mechanized, leaving a total of 176 subjective
atoms, and improving the effectiveness of the algorithm significantly.
A reasonable method to quantify the effectiveness of the algorithm on
instantiated formulas is to calculate the ratio of the number of
objective atoms to the total number of atoms for all 84 clauses. (A
more accurate assessment can be made if we also know how frequently
each clause of HIPAA gets instantiated, but this is impossible without
real data.)  In Appendix~\ref{app:hipaa}, we list for each clause the
numbers of subjective and objective atoms in it (\#S and \#O
respectively), as well as the number of subjective atoms that can be
mechanized by simple design effort such as standardizing message
formats (\#O'). The ratio (\#O' + \#O) / (\#S + \#O) shown in the last
column is an estimate of the percentage of the clause our algorithm
will reduce automatically, assuming that the required design effort
has been made. Based on these figures, we count that in 17 clauses,
all atoms can be reduced automatically; in 24 other clauses, at least
80\% of the atoms can be reduced automatically; and in 29 other
clauses, at least 50\% of the atoms can be reduced automatically. On
the other hand, in 6 clauses our algorithm cannot reduce any atoms
automatically but 5 out of these 6 clauses contain exactly one
subjective atom each.

In summary, even though completely automatic enforcement of policies
derived from HIPAA is impossible due its use of subjective predicates,
our algorithm can help reduce the burden of human auditors
significantly, both by instantiating quantifiers automatically and by
discharging objective atoms in fully instantiated formulas.

\section{Related Work}
\label{sec:related}

\paragraph{Policy Enforcement with Temporal Logic} A
lot of prior work addresses the problem of \emph{runtime monitoring}
of policies expressed in Linear Temporal Logic
(LTL)~\cite{Thati+2005:MAM, Baader+2009:alcltl,Rosu+2005:rewrite,
  Sokolsky+2006:RCD,Basin+2010:cav, Barringer+2004:RBRV} and its
extensions~\cite{Roger+2001:LAT,
  Barringer+2004:RBRV,Sokolsky+2006:RCD}. Although similar in the
spirit of enforcing policies, the intended deployment of our work is
different: we expect our algorithm to be used for after-the-fact audit
for violations, rather than for online monitoring. Consequently, the
issue of retaining only necessary portions of logs, which is central
to runtime monitoring, is largely irrelevant for our work (and hence
not considered in this paper).

Comparing only the expressiveness of the logic, our work is more
advanced than all existing work on policy enforcement. First, we
enforce a large fragment of first-order temporal logic, whereas prior
work is either limited to propositional logic~\cite{Thati+2005:MAM,
  Baader+2009:alcltl,Rosu+2005:rewrite}, or, when quantifiers are
considered, they are severely restricted~\cite{Roger+2001:LAT,
  Barringer+2004:RBRV,Sokolsky+2006:RCD}. A recent exception to such
syntactic restrictions is the work of Basin et
al.~\cite{Basin+2010:cav}, to which we compare in detail
below. Second, no prior work considers either subjective predicates,
or the possibility of gaps in past information, both of which our
partial structures and enforcement algorithm account for.

Recent work by Basin et al.~\cite{Basin+2010:cav} considers runtime
monitoring over an expressive fragment of Metric First-order Temporal
Logic. Similar to our work, Basin et al.\ allow quantification over
infinite domains, and use a form of mode analysis (called a safe-range
analysis) to ensure finiteness during enforcement. However, Basin et
al's mode analysis is weaker than ours; in particular, it cannot
relate the same variable in the input and output positions of two
different conjuncts of a restriction and requires that each free
variable appear in at least one predicate with a finite model.  As a
consequence, some policies such as $\alpha_{pol1}$
(Example~\ref{eg:past}), whose top-level restriction
$(\pred{send}(p_1, p_2, m) \conj \pred{purp}(m,u) \conj \ldots)$
contains a variable $u$ not occurring in any predicate with a finite
model, cannot be enforced in their framework, but can be enforced in
ours. Due to their goal of runtime enforcement, Basin et al.\ use
auxiliary data structures to cache relevant portions of the log in
memory, which may form the basis of useful optimizations in an
implementation of our work.




Cederquist et al.~\cite{cederquist07audit} present a proof-based
system for a-posteriori audit, where policy obligations are discharged
by constructing formal proofs. The leaves of proofs are established
from logs, but the audit process only checks that an obligation has
been satisfied somewhere in the past, thus allowing only for
obligations of the form $\diam \varphi$. Further, there is no
systematic mechanism to instantiate quantifiers in proofs. However,
using connectives of linear logic, the mechanism admits policies that
rely on consumable permissions.

The idea of iteratively rewriting the policy over evolving audit logs
has been considered
previously~\cite{Rosu+2005:rewrite,Thati+2005:MAM}, but only for
propositional logic. Bauer et al.~\cite{Baader+2009:alcltl} use a
different approach for iterative enforcement: they convert an LTL
formula with limited first-order quantification to a B\"{u}chi
automaton and check whether the automaton accepts the input
log. Further, they also use a three-valued semantic model similar to
ours, but assume past-completeness.  Three-valued structures have also
been considered in work on generalized model
checking~\cite{Bruns:2000:GMC,Godefroid:2005:MCV}. However, the
problems addressed in that line of work are different; the objective
there is to check whether there exist extensions of a given structure
in which a formula is satisfied (or falsified).



\paragraph{Policy Specification}
Several variants of LTL have been used to \emph{specify} the
properties of programs, business processes and security and privacy
policies~\cite{Barth+06:paci:fap,DeYoung+10:privacylfp:wpes,Basin+10:mfotl:sacmat,
  Giblin+05:REALM,Liu:2007:SCF}. Our representation of policies and
our logic, PrivacyLFP, draw inspiration from
LPU~\cite{Barth+06:paci:fap}.

Further, several access-control models have extensions for specifying
usage control and future obligations~\cite{hilty05:obligations,
  Bettini:2003:POP,
  Park:2002:UCON,Irwin:2006:MAO,Ni:2008:OMB,Dougherty+2007:OTIP,XACML}.
Some of these models assume a pre-defined notion of
obligations~\cite{Irwin:2006:MAO,Ni:2008:OMB}. For instance, Irwin et
al~\cite{Irwin:2006:MAO} model obligations as tuples containing the
subject of the obligation, the actions to be performed, the objects
that are targets of the actions and the time frames of the
obligations. Other models leave specifications for obligations
abstract~\cite{hilty05:obligations, Bettini:2003:POP, Park:2002:UCON}.
Such specific models and the ensuing policies can be encoded in our
logic using quantifiers and temporal operators.

There also has been much work on analyzing the properties of policies
represented in formal models. For instance, Ni et al.\ study the
interaction between obligation and authorization~\cite{Ni:2008:OMB},
Irwin et al.\ have analyzed accountability problems with
obligations~\cite{Irwin:2006:MAO}, and Dougherty et al.\ have modeled
the interaction between obligations and
programs~\cite{Dougherty+2007:OTIP}.  These methods are orthogonal to
our objective of policy enforcement. It may be possible to adapt ideas
from these papers to analyze similar properties of policies expressed
in PrivacyLFP also.

Finally, privacy languages such as EPAL~\cite{Backes+03:atfmepp} and
privacyAPI~\cite{May+06:pa:acttaavlpp} do not include obligations or
temporal modalities as primitives, and are less expressive than our
framework.

\section{Conclusion}
\label{sec:conclusion}



We have presented an expressive and provably correct iterative method
for enforcing privacy policies that works by reducing policies, even
in the face of incomplete system logs. Our method is expressive enough
to enforce real privacy legislation like HIPAA, yet tractable due to a
carefully designed static analysis. Under standard assumptions about
system logs, we obtain methods to mechanically enforce safety and
co-safety properties.

Our planned next step is to implement the proposed enforcement
mechanism and to test its performance on real privacy legislation. A
specific goal is to develop generic optimization and caching
techniques that encompass all forms of log incompleteness, to the
extent possible. Prior work on runtime monitoring may provide valuable
insights in this regard, but a significant challenge is to generalize
it beyond past-completeness.

\bibliographystyle{plain}

\bibliography{audit}

\appendix

\section{Details from Section~\ref{sec:logic}}
\label{app:logic}

The full definition of the $\dual{\varphi}$ is shown below:
\[\begin{array}{ccc}
\dual{p_O(t_1,\ldots,t_n)} & = & \dual{p_O}(t_1,\ldots,t_n) \\
\dual{p_S(t_1,\ldots,t_n)} & = & \dual{p_S}(t_1,\ldots,t_n) \\
\dual{\top} & = & \bot\\
\dual{\bot} & = & \top\\
\dual{\varphi \conj \psi} & = & \dual{\varphi} \disj \dual{\psi} \\
\dual{\varphi \disj \psi} & = & \dual{\varphi} \conj \dual{\psi} \\
\dual{\forall \vec{x} \not \in S. (c \imp \varphi)} & = & \exists \vec{x} \not \in S. (c \conj \dual{\varphi})\\
\dual{\exists \vec{x} \not \in S. (c \conj \varphi)} & = & \forall \vec{x} \not \in S. (c \imp \dual{\varphi})
\end{array}\]

The full translation $\trans{\ttime}{\bullet}$ from the temporal logic
to the sublogic is shown below:

\[\begin{array}{ccl}

\trans{\ttime}{p_O(t_1,\ldots,t_n)} & = &
p_O(t_1,\ldots,t_n,\ttime)\\
\trans{\ttime}{\top} & = & \top\\
\trans{\ttime}{\bot} & = & \bot\\
\trans{\ttime}{c_1 \conj c_2} & = & \trans{\ttime}{c_1}
\conj \trans{\ttime}{c_2} \\
\trans{\ttime}{c_1 \disj c_2} & = & \trans{\ttime}{c_1}
\disj \trans{\ttime}{c_2} \\
\trans{\ttime}{\exists x. c} & = & \exists
x. \trans{\ttime}{c}\\

\\

\trans{\ttime}{p_O(t_1,\ldots,t_n)} & = &
p_O(t_1,\ldots,t_n,\ttime) \\

\trans{\ttime}{p_S(t_1,\ldots,t_n)} & = &
p_S(t_1,\ldots,t_n,\ttime) \\

\trans{\ttime}{\top} & = & \top\\

\trans{\ttime}{\bot} & = & \bot\\

\trans{\ttime}{\alpha \conj \beta} & = & \trans{\ttime}{\alpha} \conj
\trans{\ttime}{\beta}\\

\trans{\ttime}{\alpha \disj \beta} & = & \trans{\ttime}{\alpha} \disj
\trans{\ttime}{\beta}\\

\trans{\ttime}{\neg \alpha} & = & \dual{\trans{\ttime}{\alpha}}\\

\trans{\ttime}{\forall \vec{x}. (c \imp \alpha)} & =
& \forall \vec{x}. (\trans{\ttime}{c} \imp
\trans{\ttime}{\alpha})\\

\trans{\ttime}{\exists \vec{x}. (c \conj \alpha)} & =
& \exists \vec{x}. (\trans{\ttime}{c} \conj
\trans{\ttime}{\alpha})\\

\trans{\ttime}{\here x. \alpha} & = &
\trans{\ttime}{\alpha[\ttime/x]}\\

\trans{\ttime}{\alpha \since \beta} & = & \exists
\ttime'. (\pred{in}(\ttime', 0, \ttime) \conj
\trans{\ttime'}{\beta} 
\\ & & ~ \conj (\forall \ttime''. (({\tt
  in}(\ttime'',\ttime',\ttime) \conj \ttime' \not= \ttime'') 
\\& & \qquad\quad \imp
\trans{\ttime''}{\alpha})))\\

\trans{\ttime}{\alpha \until \beta} & = & \exists
\ttime'. (\pred{in}(\ttime', \ttime, \infty) \conj
\trans{\ttime'}{\beta} 
\\ & & ~ \conj (\forall \ttime''. (({\tt
  in}(\ttime'',\ttime,\ttime') \conj \ttime'' \not= \ttime') 
\\& & \qquad\quad\imp
  \trans{\ttime''}{\alpha})))\\

\trans{\ttime}{\boxm{\alpha}} & = & \forall
\ttime'. (\pred{in}(\ttime',0,\ttime) \imp
\trans{\ttime'}{\alpha})\\

\trans{\ttime}{\boxp{\alpha}} & = & \forall
\ttime'. (\pred{in}(\ttime', \ttime,\infty) \imp
\trans{\ttime'}{\alpha})
\end{array}\]

\section{Proofs from Section~\ref{sec:enforcement}}
\label{app:enforcement}

This appendix contains proofs of theorems presented in
Section~\ref{sec:enforcement}. The proofs are presented in an order
different from the order of theorems in the main body of the paper
because of dependencies in the proofs.


\begin{lem}[Monotonicity]\label{lem:monotonicity}
$\structure' \geq \structure$ and $\structure \models \varphi$ imply
  $\structure' \models \varphi$.
\end{lem}
\begin{proof} By induction on $\varphi$.
\end{proof}


\begin{lem}[Consistency]\label{lem:consistency}
For all $\structure$ and $\varphi$, either $\structure \not \models
\varphi$ or $\structure \not \models \dual{\varphi}$.
\end{lem}
\begin{proof}
By induction on $\varphi$.
\end{proof}


\begin{thm}[Correctness of $\lift{\sat}$; Theorem~\ref{thm:sat:correct}]\label{thm:sat:correct:app}
If $\lift{\sat}(\structure, c)$ is defined then for any substitution
$\sigma'$ with $\dom(\sigma') \supseteq \fv(c)$, $\structure \models
c\sigma'$ iff there is a substitution $\sigma
\in\lift{\sat}(\structure, c)$ such that $\sigma' \geq \sigma$.
\end{thm}
\begin{proof}
By induction on $c$ and case analysis of its top-level constructor.\\

\case ~$c = p_o(t_1,\ldots,t_n)$. Then, $\lift{\sat}(\structure, c) =
\sat(\structure, c)$. The result follows from the condition that
$\sat$ is required to satisfy (Section~\ref{sec:mode}).\\

\case ~$c = \top$. Then, $\lift{\sat}(\structure, c) =
\{\bullet\}$. If $\structure \models c \sigma'$, $\sigma'$ trivially
extends $\bullet$ by definition. Conversely, any substitution
$\sigma'$ trivially satisfies $\structure \models \top \sigma'$.\\

\case ~$c = \bot$. Then, $\lift{\sat}(\structure, c) = \{\}$. The
result is vacuously true in both directions because $\structure \not
\models \bot \sigma'$, and $\sigma' \not \in \{\}$.\\

\case ~$c = c_1 \conj c_2$. Then, $\lift{\sat}(\structure, c) =
\bigcup_{\sigma_1 \in \lift{\sat}(\structure, c_1) } \sigma_1 +
\lift{\sat}(\structure, c_2 \sigma_1)$. Clearly, if this exists, then
$\lift{\sat}(\structure, c_1)$ must be defined also, and for each
$\sigma_1 \in \lift{\sat}(\structure, c_1)$, $\lift{\sat}(\structure,
c_2 \sigma_1)$ must also be defined. 

Suppose $\structure \models (c_1 \conj c_2) \sigma'$. By definition of
$\models$, we get $\structure \models c_1 \sigma'$ and $\structure
\models c_2 \sigma'$. By the i.h., the former implies that there is a
$\sigma_1 \in \lift{\sat}(\structure, c_1)$ such that $\sigma' \geq
\sigma_1$. This also implies that $c_2 \sigma' = (c_2 \sigma_1)
\sigma'$. So, $\structure \models c_2 \sigma'$ implies $\structure
\models (c_2 \sigma_1) \sigma'$. Consequently, by the i.h.\ on $c_2
\sigma_1$, there must be a $\sigma_2 \in \lift{\sat}(\structure, c_2
\sigma_1)$ such that $\sigma' \geq \sigma_2$. It follows that $\sigma'
\geq \sigma_1 + \sigma_2$. Clearly, $(\sigma_1 + \sigma_2) \in
(\sigma_1 + \lift{\sat}(\structure, c_2 \sigma_1)) \subseteq
(\bigcup_{\sigma_1 \in \lift{\sat}(\structure, c_1) } \sigma_1 +
\lift{\sat}(\structure, c_2 \sigma_1)) = \lift{\sat}(\structure, c)$,
as required.

Conversely, suppose that there is a $\sigma \in \bigcup_{\sigma_1 \in
  \lift{\sat}(\structure, c_1) } \sigma_1 + \lift{\sat}(\structure,
c_2 \sigma_1)$ and $\sigma' \geq \sigma$ with $\dom(\sigma') \supseteq
\fv(\sigma)$. We need to show that $\structure \models (c_1 \conj c_2)
\sigma'$ or, equivalently, $\structure \models c_1 \sigma'$ and
$\structure \models c_2 \sigma'$. By set-theory, there must be a
$\sigma_1 \in \lift{\sat}(\structure, c_1)$ and a $\sigma_2 \in
\lift{\sat}(\structure, c_2\sigma_1)$ such that $\sigma = \sigma_1 +
\sigma_2$. Clearly, $\sigma' \geq \sigma_1$. So, by the i.h., we
immediately have $\structure \models c_1 \sigma'$. Similarly,
$\sigma' \geq \sigma_2$. So, by i.h.\ on $c_2 \sigma_1$, $\structure
\models c_2\sigma_1 \sigma'$. But, $c_2 \sigma_1 \sigma' = c_2
\sigma'$. Therefore, $\structure \models c_2\sigma'$.\\

\case ~$c = c_1 \disj c_2$. Then, $\lift{\sat}(\structure, c) =
\lift{\sat}(\structure, c_1) \cup \lift{\sat}(\structure, c_2)$. If
this is defined, then, clearly, both $\lift{\sat}(\structure, c_1)$
and $\lift{\sat}(\structure, c_2)$ must be defined.

Suppose $\structure \models (c_1 \disj c_2) \sigma'$. By definition of
$\models$, we get that either $\structure \models c_1 \sigma'$ or
$\structure \models c_2 \sigma'$. We consider here the former case
(the latter is similar). So $\structure \models c_1 \sigma'$. By the
i.h., there is a $\sigma \in \lift{\sat}(\structure, c_1)$ such that
$\sigma' \geq \sigma_1$. The proof is complete by noting that
$\sigma_1 \in \lift{\sat}(\structure, c_1) \in \lift{\sat}(\structure,
c)$.

Conversely, suppose that there is a $\sigma \in
\lift{\sat}(\structure, c_1) \cup \lift{\sat}(\structure, c_2)$ and
$\sigma' \geq \sigma$ with $\dom(\sigma') \supseteq \fv(\sigma)$. We
need to show that $\structure \models (c_1 \disj c_2) \sigma'$ or,
equivalently, either $\structure \models c_1\sigma'$ or $\structure
\models c_2\sigma'$. From $\sigma \in \lift{\sat}(\structure, c_1)
\cup \lift{\sat}(\structure, c_2)$, we get that either $\sigma \in
\lift{\sat}(\structure, c_1)$ or $\sigma \in \lift{\sat}(\structure,
c_2)$. Consider the former case (the latter is similar): $\sigma \in
\lift{\sat}(\structure, c_1)$. By i.h.\ on $c_1$, we immediately get
$\structure \models c_1 \sigma'$, as required.\\

\case ~$c = \exists x. c'$. Then, $\lift{\sat}(\structure, c) =
\lift{\sat}(\structure, c') \backslash \{x\}$. If this is defined,
then, clearly, $\lift{\sat}(\structure, c')$ must also be defined.

Suppose $\structure \models (\exists x. c') \sigma'$. By definition of
$\models$, there must be a $t$ such that $\structure \models
c'[t/x]\sigma'$. By i.h.\ on $c'$, there must be a $\sigma \in
\lift{\sat}(\structure, c')$ such that $(\sigma' + [x \mapsto t]) \geq
\sigma$. Clearly, $\sigma' \geq \sigma \backslash \{x\}$ and
$\sigma\backslash\{x\} \in \lift{\sat}(\structure, c)$, as required.

Conversely, suppose that there is a $\sigma \in
\lift{\sat}(\structure, c') \backslash \{x\}$ and $\sigma' \geq
\sigma$ with $\dom(\sigma') \supseteq \fv(c)$. We need to show that
$\structure \models c \sigma'$. Because $\sigma \in
\lift{\sat}(\structure, c') \backslash \{x\}$, there is a $\sigma''
\in \lift{\sat}(\structure, c')$ and a $t$ such that $\sigma'' =
\sigma + [x \mapsto t]$. Clearly, $\sigma' + [x \mapsto t] \geq \sigma
+ [x \mapsto t] = \sigma''$. By i.h.\ on $c'$, $\structure \models c'
[t/x] \sigma'$, which implies (by definition of $\models$) that
$\structure \models (\exists x.c') \sigma'$, i.e., $\structure \models
c\sigma'$.
\end{proof}


\begin{lem}[Duality of $\reduce$]\label{lem:duality:reduce}
$\reduce(\structure, \dual{\varphi})$ = $\dual{\reduce(\structure,
    \varphi)}$.
\end{lem}
\begin{proof}
By a straightforward induction on $\varphi$. We show some
representative cases below.\\

\case ~$\varphi = P$. Then, 
\[\reduce(\structure, P) ~ = ~ \left\lbrace
\begin{array}{ll}
\top & \mbox{ if $\rho_{\structure}(P) = \symt$}\\
\bot & \mbox{ if $\rho_{\structure}(P) = \symf$}\\
P & \mbox{ if $\rho_{\structure}(P) = \symu$}
\end{array}\right.\]
We consider all three possible subcases on $\rho_\structure(P)$. If
$\rho_\structure(P) = \symt$, then, by definition,
$\rho_\structure(\dual{P}) = \symf$, so $\reduce(\structure, \dual{P})
= \bot = \dual{\top} = \dual{\reduce(\structure, P)}$. The case of
$\rho_\structure(P) = \symf$ is similar. For $\rho_\structure(P) =
\symu$, we have $\rho_\structure(\dual{P}) = \symu$, so
$\reduce(\structure, \dual{P}) = \dual{P} = \dual{\reduce(\structure,
  P)}$.\\

\case ~$\varphi = \forall \vec{x}. (c \imp \varphi')$.  Then,
$\reduce(\structure, \varphi)$ is calculated as follows:
\[\reduce(\structure, \forall \vec{x}.(c \imp \varphi')) ~ = ~ 
\begin{array}[t]{@{}l}
\mbox{let}\\
~\{\sigma_1,\ldots,\sigma_n\} \leftarrow \lift{\sat}(\structure,c)\\
~\{ \vec{t_i} \leftarrow \sigma_i(\vec{x}) \}_{i=1}^n\\
~S \leftarrow \{\vec{t_1},\ldots,\vec{t_n}\}\\
~\{ \psi_i \leftarrow \reduce(\structure,
\varphi'[\vec{t_i}/\vec{x}]) \}_{i=1}^n\\
~\psi' \leftarrow \forall \vec{x}.((c \conj \vec{x} \not \in S) \imp \varphi')\\
\mbox{return } \\
~~~ \psi_1 \conj \ldots \conj \psi_n \conj \psi'
\end{array} \]
Note that $\dual{\varphi} = \exists \vec{x}.(c \conj
\dual{\varphi'})$. Consequently, $\reduce(\structure, \dual{\varphi})$
is calculated as follows, where we have renamed some bound variables to
distinguish them from those in the above display.
\[\reduce(\structure, \exists \vec{x}.(c \conj \dual{\varphi'})) ~ = ~ 
\begin{array}[t]{@{}l}
\mbox{let}\\
~\{\sigma_1',\ldots,\sigma_n'\} \leftarrow \lift{\sat}(\structure,c)\\
~\{ \vec{t_i'} \leftarrow \sigma_i'(\vec{x}) \}_{i=1}^n\\
~S' \leftarrow \{\vec{t_1'},\ldots,\vec{t_n'}\}\\
~\{ \psi_i' \leftarrow \reduce(\structure,
\dual{\varphi'}[\vec{t_i'}/\vec{x}]) \}_{i=1}^n\\
~\psi'' \leftarrow \exists \vec{x}.((c \conj \vec{x} \not \in S') \conj \dual{\varphi'})\\
\mbox{return } \\
~~~ \psi_1' \disj \ldots \disj \psi_n' \disj \psi''
\end{array} \]

We must have $\sigma_i = \sigma_i'$ (because both are calculated using
$\lift{\sat}(\structure, c)$) and, consequently, $\vec{t_i} =
\vec{t_i'}$ and $S = S'$. Thus, by the i.h., we get that
$\reduce(\structure, \dual{\varphi'}[\vec{t_i'}/\vec{x}]) =
\dual{\reduce(\structure, \varphi'[\vec{t_i}/\vec{x}])}$, i.e.,
$\psi_i' = \dual{\psi_i}$. Also observe that directly from definition
of duality, $\psi'' = \dual{\psi'}$. Thus, $\reduce(\structure,
\dual{\varphi}) = \psi_1' \disj \ldots \disj \psi_n' \disj \psi'' =
\dual{\psi_1} \disj \ldots \disj \dual{\psi_n} \disj \dual{\psi'} =
\dual{\psi_1 \conj \ldots \conj \psi_n \conj \psi'} =
\dual{\reduce(\structure, \varphi)}$.
\end{proof}


\begin{thm}[Correctness of $\reduce$; Theorem~\ref{thm:correctness}]\label{thm:correctness:app}
If $\reduce(\structure, \varphi) = \psi$ and $\structure' \geq
\structure$, then (1)~$\structure' \models \varphi$ iff $\structure'
\models \psi$ and (2)~$\structure' \models \dual{\varphi}$ iff
$\structure' \models \dual{\psi}$.
\end{thm}
\begin{proof}
First observe that (1) implies (2). Why? Suppose (1) holds for all
$\varphi$. We need to show that (2) holds. So suppose
$\reduce(\structure, \varphi) = \psi$ and $\structure' \geq
\structure$. By Lemma~\ref{lem:duality:reduce}, $\reduce(\structure,
\dual{\varphi}) = \dual{\psi}$. Applying the assumed (1) to
$\dual{\varphi}$ instead of $\varphi$, we immediately deduce that
$\structure' \models \dual{\varphi}$ iff $\structure' \models
\dual{\psi}$, as required. 

Hence, we only need to prove (1). We do that by induction on
$\varphi$, and a case analysis of its top-level constructor.\\

\case ~$\varphi = P$. Then, 
\[\reduce(\structure, \varphi) ~ = ~ \left\lbrace
\begin{array}{ll}
\top & \mbox{ if $\rho_{\structure}(P) = \symt$}\\
\bot & \mbox{ if $\rho_{\structure}(P) = \symf$}\\
P & \mbox{ if $\rho_{\structure}(P) = \symu$}
\end{array} \right. 
\]
We consider three subcases on the value of $\rho_\structure(P)$. \\

\subcase ~$\rho_\structure(P) = \symt$. Here, $\psi = \top$. First,
assume that $\structure' \models \varphi$. Then, we need to prove that
$\structure' \models \psi$, i.e, $\structure' \models \top$. This
follows directly from the definition of $\models$. Conversely, assume
that $\structure' \models \psi$. We need to prove that $\structure'
\models P$. By definition, this is equivalent to proving
$\rho_{\structure'}(P) = \symt$, which follows immediately from the
subcase assumption $\rho_\structure(P) = \symt$ and the assumption
$\structure' \geq \structure$. \\

\subcase ~$\rho_\structure(P) = \symf$. Here $\psi = \bot$. First,
assume that $\structure' \models \varphi$. We need to show that
$\structure' \models \psi$. From the subcase assumption, we have
$\rho_\structure(P) = \symf$, so the definition of $\structure' \geq
\structure$ implies that $\rho_{\structure'}(P) = \symf$. However,
$\structure' \models \varphi$ implies $\structure' \models P$, i.e.,
$\rho_{\structure'}(P) = \symt$ -- a contradiction.  Thus,
$\structure' \models \psi$ holds vacuously.

Conversely, suppose that $\structure' \models \psi$, i.e.,
$\structure' \models \bot$. By definition of $\models$, this is a
contradiction, so $\structure' \models \varphi$ holds vacuously, as
required.\\

\subcase ~$\rho_\structure(P) = \symu$. Here, $\varphi = \psi = P$, so
the case is trivial.\\

\case ~$\varphi = \top$. Then, $\psi = \reduce(\structure, \varphi) =
\reduce(\structure, \top) = \top$. Since $\varphi = \psi$, the case is
trivial.\\

\case ~$\varphi = \bot$. Then, $\psi = \reduce(\structure, \varphi) =
\reduce(\structure, \bot) = \bot$. Since $\varphi = \psi$, the case is
trivial.\\

\case ~$\varphi = \varphi_1 \conj \varphi_2$. Then, $\psi =
\reduce(\structure, \varphi_1) \conj \reduce(\structure, \varphi_2)$,
so both the conjuncts exist. First, suppose that $\structure' \models
\varphi$, i.e., $\structure' \models \varphi_1$ and $\structure'
\models \varphi_2$. By the i.h., $\structure' \models
\reduce(\structure, \varphi_1)$ and $\structure' \models
\reduce(\structure, \varphi_2)$ or, equivalently, $\structure' \models
\psi$.

Conversely, suppose that $\structure' \models \psi$. Then,
$\structure' \models \reduce(\structure, \varphi_1)$ and $\structure'
\models \reduce(\structure, \varphi_2)$. By the i.h., $\structure'
\models \varphi_1$ and $\structure' \models \varphi_2$, i.e.,
$\structure'\models \varphi$.\\

\case ~$\varphi = \varphi_1 \disj \varphi_2$. Then, $\psi =
\reduce(\structure, \varphi_1) \disj \reduce(\structure, \varphi_2)$,
so both the disjuncts exist. First, suppose that $\structure' \models
\varphi$, i.e., either $\structure' \models \varphi_1$ or $\structure'
\models \varphi_2$. By the i.h., either $\structure' \models
\reduce(\structure, \varphi_1)$ or $\structure' \models
\reduce(\structure, \varphi_2)$. Equivalently, $\structure' \models
\psi$.

Conversely, suppose that $\structure' \models \psi$. Then, either
$\structure' \models \reduce(\structure, \varphi_1)$ or $\structure'
\models \reduce(\structure, \varphi_2)$. By the i.h., either
$\structure' \models \varphi_1$ or $\structure' \models
\varphi_2$. Equivalently, $\structure'\models \varphi$.\\

\case ~$\varphi = \forall \vec{x}.(c \imp \varphi')$. Then, $\psi =
\reduce(\structure, \varphi)$ is calculated as follows.
\[\reduce(\structure, \forall \vec{x}.(c \imp \varphi')) ~ = ~ 
\begin{array}[t]{@{}l}
\mbox{let}\\
~\{\sigma_1,\ldots,\sigma_n\} \leftarrow \lift{\sat}(\structure,c)\\
~\{ \vec{t_i} \leftarrow \sigma_i(\vec{x}) \}_{i=1}^n\\
~S \leftarrow \{\vec{t_1},\ldots,\vec{t_n}\}\\
~\{ \psi_i \leftarrow \reduce(\structure,
\varphi'[\vec{t_i}/\vec{x}]) \}_{i=1}^n\\
~\psi' \leftarrow \forall \vec{x}.((c \conj \vec{x} \not \in S) \imp \varphi')\\
\mbox{return } \\
~~~ \psi_1 \conj \ldots \conj \psi_n \conj \psi'
\end{array} \]

So $\psi = \psi_1 \conj \ldots \conj \psi_n \conj \forall \vec{x}.((c
\conj \vec{x} \not \in S) \imp \varphi')$. First, suppose that
$\structure' \models \varphi$, i.e., $\structure' \models \forall
\vec{x}.(c \imp \varphi')$. We need to prove that $\structure' \models
\psi$, i.e., $\structure' \models \psi_i$ and $\structure' \models
\forall \vec{x}.((c \conj \vec{x} \not \in S) \imp \varphi')$. We
first prove that $\structure' \models \psi_i$. Because
$\reduce(\structure, \varphi'[\vec{t_i}/\vec{x}]) = \psi_i$, by the
i.h., it suffices to show that $\structure' \models
\varphi'[\vec{t_i}/\vec{x}]$. From the definition of $\structure'
\models \forall \vec{x}.(c\imp \varphi')$, either $\structure' \models
\dual{c}[\vec{t_i}/\vec{x}]$ or $\structure' \models
\varphi'[\vec{t_i}/\vec{x}]$. Hence, it suffices to prove that
$\structure' \not\models \dual{c}[\vec{t_i}/\vec{x}]$. Suppose, for
the sake of contradiction, that $\structure' \models
\dual{c}[\vec{t_i}/\vec{x}]$.  Since $\sigma_i \in
\lift{\sat}(\structure, c)$, Theorem~\ref{thm:sat:correct:app} yields
$\structure \models c\sigma_i$, i.e., $\structure \models
c[\vec{t_i}/\vec{x}]$ (note that because $\forall \vec{x}.(c \imp
\varphi')$ is closed, $\fv(c) \subseteq \vec{x}$; so
$c[\vec{t_i}/\vec{x}] = c \sigma_i$). Hence, by
Lemma~\ref{lem:monotonicity}, $\structure' \models
c[\vec{t_i}/\vec{x}]$, which, by Lemma~\ref{lem:consistency},
contradicts the earlier fact $\structure' \models
\dual{c}[\vec{t_i}/\vec{x}]$.

Next, we show that $\structure' \models \forall \vec{x}.((c \conj
\vec{x} \not \in S) \imp \varphi')$. Following the definition of
$\models$, pick any $\vec{t}$. We show that either $\structure'
\models \dual{(c \conj \vec{x} \not \in S)}[\vec{t}/\vec{x}]$ or
$\structure' \models \varphi'[\vec{t}/\vec{x}]$. Since we assumed that
$\structure' \models \forall \vec{x}. (c\imp \varphi')$, either
$\structure' \models \dual{c}[\vec{t}/\vec{x}]$ or $\structure'
\models \varphi'[\vec{t}/\vec{x}]$. The proof is complete by observing
that $\structure' \models \dual{c}[\vec{t}/\vec{x}]$ implies
$\structure' \models \dual{(c \conj \vec{x} \not \in
  S)}[\vec{t}/\vec{x}]$.

Conversely, assume that $\structure' \models \psi$, i.e., $\structure'
\models \psi_i$ and $\structure' \models \forall \vec{x}.((c \conj
\vec{x} \not \in S) \imp \varphi')$. We need to prove that
$\structure' \models \varphi$, i.e., $\structure' \models \forall
\vec{x}.(c \imp \varphi')$. Following the definition of $\models$,
pick any $\vec{t}$. We need to prove that either $\structure' \models
\dual{c}[\vec{t}/\vec{x}]$ or $\structure' \models
\varphi'[\vec{t}/\vec{x}]$. We consider two subcases. Either $\vec{t}
\in S$ or $\vec{t} \not \in S$. \\

\subcase ~$\vec{t} \in S$. Then, $\vec{t} = \vec{t_i}$ for some
$i$. Since $\reduce(\structure, \varphi'[\vec{t_i}/\vec{x}]) = \psi_i$
and $\structure' \models \psi_i'$, by the i.h.\ we get $\structure'
\models \varphi'[\vec{t_i}/\vec{x}]$, as required. \\

\subcase ~$\vec{t} \not \in S$. We already know that $\structure'
\models \forall \vec{x}.((c \conj \vec{x} \not \in S) \imp
\varphi')$. So, either $\structure' \models \dual{(c \conj \vec{x}
  \not \in S)}[\vec{t}/\vec{x}]$ or $\structure' \models
\varphi'[\vec{t}/\vec{x}]$. If the latter, we are done, so assume the
former. Thus, $\structure' \models \dual{(c \conj \vec{x} \not \in
  S)}[\vec{t}/\vec{x}]$, i.e., $\structure' \models
\dual{c}[\vec{t}/\vec{x}] \disj \vec{t} \in S$. This immediately
implies that either $\structure' \models \dual{c}[\vec{t}/\vec{x}]$ or
$\vec{t} \in S$. The former case is sufficient for our purpose, and
the latter case contradicts the subcase assumption.\\

\case ~$\varphi = \exists \vec{x}.(c \conj \varphi')$. Then,
$\reduce(\structure, \varphi)$ is calculated as follows.
\[\reduce(\structure, \exists \vec{x}.(c \conj \varphi')) ~ = ~ 
\begin{array}[t]{@{}l}
\mbox{let}\\
~\{\sigma_1,\ldots,\sigma_n\} \leftarrow \lift{\sat}(\structure,c)\\
~\{ \vec{t_i} \leftarrow \sigma_i(\vec{x}) \}_{i=1}^n\\
~S \leftarrow \{\vec{t_1},\ldots,\vec{t_n}\}\\
~\{ \psi_i \leftarrow \reduce(\structure,
\varphi'[\vec{t_i}/\vec{x}]) \}_{i=1}^n\\
~\psi' \leftarrow \exists \vec{x}.((c \conj \vec{x} \not \in S) \conj \varphi')\\
\mbox{return } \\
~~~ \psi_1 \disj \ldots \disj \psi_n \disj \psi'
\end{array}\]

So, $\psi = \psi_1 \disj \ldots \disj \psi_n \disj \exists \vec{x}.((c
\conj \vec{x} \not \in S) \conj \varphi')$. First suppose that
$\structure' \models \varphi$. We show that $\structure' \models
\psi$. Following the definition of $\models$ on $\structure' \models
\varphi$, we obtain a $\vec{t}$ such that $\structure' \models
c[\vec{t}/\vec{x}]$ and $\structure' \models
\varphi'[\vec{t}/\vec{x}]$. We consider two subcases: either $\vec{t}
\in S$ or $\vec{t} \not \in S$.\\

\subcase ~$\vec{t} \in S$. So, $\vec{t} = \vec{t_i}$ for some $i$ and
from $\structure' \models \varphi'[\vec{t}/\vec{x}]$ we obtain
$\structure' \models \varphi'[\vec{t_i}/\vec{x}]$. Since
$\reduce(\structure, \varphi'[\vec{t_i}/\vec{x}]) = \psi_i$, by the
i.h., we get $\structure' \models \psi_i$, which immediately implies
$\structure' \models \psi$.\\

\subcase ~$\vec{t} \not \in S$. Combining this and $\structure'
\models c[\vec{t}/\vec{x}]$, we get $\structure' \models (c \conj
\vec{x}\not\in S)[\vec{t}/\vec{x}]$. Since $\structure' \models
\varphi'[\vec{t}/\vec{x}]$, we derive from the definition of $\models$
that $\structure' \models \exists \vec{x}.((c \conj \vec{x} \not\in S)
\conj \varphi')$. This immediately yields $\structure' \models \psi$.\\

Conversely, suppose that $\structure' \models \psi$. We show that
$\structure' \models \varphi$. $\structure' \models \psi$ implies that
either $\structure' \models \psi_i$ for some $i$ or $\structure'
\models \exists \vec{x}.((c \conj \vec{x} \not \in S) \conj
\varphi')$. We consider both subcases below.\\

\subcase ~$\structure \models \psi_i$. Since $\reduce(\structure,
\varphi'[\vec{t_i}/\vec{x}]) = \psi_i$, by the i.h., $\structure'
\models \varphi'[\vec{t_i}/\vec{x}]$. Further, observe that because
$\sigma_i \in \lift{\sat}(\structure, c)$ and $\dom(\sigma_i)
\supseteq \vec{x} \supseteq \fv(c)$ (the latter because $\exists
\vec{x}. (c \imp \varphi')$ must be closed),
Theorem~\ref{thm:sat:correct:app} yields $\structure \models
c\sigma_i$ and, hence, $\structure \models c[\vec{t_i}/\vec{x}]$. By
Lemma~\ref{lem:monotonicity}, $\structure' \models
c[\vec{t_i}/\vec{x}]$. Since we have already derived $\structure'
\models \varphi'[\vec{t_i}/\vec{x}]$, the definition of $\models$
yields that $\structure' \models \exists \vec{x}. (c \conj \varphi')$,
as required.\\

\subcase ~$\structure' \models \exists \vec{x}.((c \conj \vec{x} \not
\in S) \conj \varphi')$. Thus there must be a $\vec{t}$ such that
$\structure' \models c[\vec{t}/\vec{x}]$, $\vec{t} \not \in S$, and
$\structure' \models \varphi'[\vec{t}/\vec{x}]$. The first and third
facts in the last sentence imply that $\structure' \models \exists
\vec{x}. (c \conj \varphi')$, as required.
\end{proof}


\begin{thm}[Totality of $\lift{\sat}$;
    Theorem~\ref{thm:sat:total}]
\label{thm:sat:total:app}
If $\chi_I \vdash c: \chi_O$, then for all structures $\structure$ and
all substitutions $\sigma$ with $\dom(\sigma) \supseteq \chi_I$,
$\lift{\sat}(\structure, c\sigma)$ is defined and, further, for each
substitution $\sigma' \in \lift{\sat}(\structure, c\sigma)$, $\chi_I
\cup \dom(\sigma') \supseteq \chi_O$.
\end{thm}
\begin{proof}
By induction on the given derivation of $\chi_I \vdash c: \chi_O$ and
case analysis of its last rule.\\

\case ~$\inferrule{\forall k \in I(p_O). ~\fv(t_k) \subseteq \chi_I
  \\ \chi_O = \chi_I \cup (\bigcup_{j \in O(p_O)} \fv(t_j))}{ \chi_I
  \vdash p_O(t_1,\ldots,t_n) : \chi_O}$

We are given $\sigma$ such that $\dom(\sigma) \supseteq \chi_I$. From
this and the first premise it follows that $\forall k \in
I(p_O). ~{\tt ground}(t_k\sigma)$. Thus, by definition,
$\sat(\structure, p_O(t_1,\ldots,t_n)\sigma)$ is
defined. Consequently, $\lift{\sat}(\structure,
p_O(t_1,\ldots,t_n)\sigma)$, which equals $\sat(\structure,
p_O(t_1,\ldots,t_n)\sigma)$ is also defined. Pick any $\sigma'
\in\linebreak[6] \sat(\structure, p_O(t_1,\ldots,t_n)\sigma)$. By
definition of $\sat$, $\dom(\sigma') \supseteq \bigcup_{j \in O(p_O)}
\fv(t_j)$. Consequently, $\chi_I \cup \dom(\sigma') \supseteq \chi_I
\cup (\bigcup_{j \in O(p_O)} \fv(t_j)) \supseteq \chi_O$, where the
last relation follows from the second premise.\\

\case ~$\inferrule{ }{ \chi_I \vdash \top: \chi_I}$

Suppose $\dom(\sigma) \supseteq \chi_I$. Note that
$\lift{\sat}(\structure, \top\sigma) = \lift{\sat}(\structure, \top) =
\{\bullet\}$ is always defined. If $\sigma' \in \{\bullet\}$, then
$\sigma' = \bullet$. Clearly, $\chi_I \cup \dom(\sigma') = \chi_I \cup
\dom(\bullet) = \chi_I = \chi_O$.\\

\case ~$\inferrule{ }{ \chi_I \vdash \bot: \chi_I}$

Suppose $\dom(\sigma) \supseteq \chi_I$. Note that
$\lift{\sat}(\structure, \bot\sigma) = \lift{\sat}(\structure, \bot) =
\{\}$ is always defined. Because there cannot be a $\sigma' \in \{\}$,
the rest of the proof holds vacuously in this case.\\

\case ~$\inferrule{ \chi_I \vdash c_1 : \chi \\ \chi \vdash c_2:
  \chi_O}{ \chi_I \vdash c_1 \conj c_2: \chi_O}$

Suppose $\dom(\sigma) \supseteq \chi_I$. By i.h.\ on the first
premise, $\lift{\sat}(\structure, c_1\sigma)$ is defined. Let
$\lift{\sat}(\structure, c_1\sigma) = \{\sigma_1, \ldots,
\sigma_n\}$. Also by the i.h., $\chi_I \cup \dom(\sigma_i)
\supseteq \chi$. Call this fact (A). Since $\dom(\sigma)
\supseteq \chi_I$, fact (A) implies $\dom(\sigma + \sigma_i)
\supseteq \chi$. Using the latter, by i.h.\ on the second premise and
each of $\{\sigma + \sigma_1 , \ldots, \sigma + \sigma_n\}$, we obtain
that each of $\lift{\sat}(\structure, c_2\sigma\sigma_i)$ are also
defined for each $i$ and $\forall \sigma_i' \in
\lift{\sat}(\structure, c_2\sigma\sigma_i)$, $\chi \cup \dom(\sigma_i') \supseteq \chi_O$. Call the last fact~(B). We
immediately have that $\lift{\sat}(\structure, (c_1 \conj c_2)\sigma)
= \bigcup_{\sigma_1 \in \lift{\sat}(\structure, c_1\sigma)}
\lift{\sat}(\structure, c_2\sigma\sigma_1)$ is also defined.  

Pick any $\sigma' \in \lift{\sat}(\structure, (c_1 \conj
c_2)\sigma)$. Then for some $i$ and some $\sigma_i' \in
\lift{\sat}(\structure, c_2\sigma\sigma_i)$, we have $\sigma' =
\sigma_i + \sigma_i'$. We want to show that $\chi_I \cup \dom(\sigma')
\supseteq \chi_O$. Or, equivalently, $\chi_I \cup \dom(\sigma_i +
\sigma_i') \supseteq \chi_O$. However, $\chi_I \cup \dom(\sigma_i
+ \sigma_i') = \chi_I \cup \dom(\sigma_i) \cup \dom
(\sigma_i') \supseteq \chi \cup \dom (\sigma_i') \supseteq
\chi_O$, where the last two relations follow from facts (A) and (B),
respectively.\\

\case ~$\inferrule{ \chi_I \vdash c_1: \chi_1 \\  \chi_I
  \vdash c_2: \chi_2}{ \chi_I \vdash c_1 \disj
  c_2: \chi_1 \cap \chi_2}$

Suppose $\dom(\sigma) \supseteq \chi_I$. By i.h.\ on the first
premise, $\lift{\sat}(\structure, c_1\sigma)$ is defined and $\forall
\sigma' \in \lift{\sat}(\structure, c_1\sigma)$, $\chi_I \cup
\dom(\sigma') \supseteq \chi_1$. Call this fact~(A). Similarly, by
i.h.\ on the second premise, $\lift{\sat}(\structure, c_2\sigma)$ is
defined and $\forall \sigma' \in \lift{\sat}(\structure, c_2\sigma)$,
$\chi_I \cup \dom(\sigma') \supseteq \chi_2$. Call this
fact~(B). By definition of $\lift{\sat}$, $\lift{\sat}(\structure,
(c_1 \disj c_2) \sigma) = \lift{\sat}(\structure, c_1\sigma) \cup
\lift{\sat}(\structure, c_2\sigma)$ is defined.

Pick any $\sigma' \in \lift{\sat}(\structure, c_1\sigma) \cup
\lift{\sat}(\structure, c_2\sigma)$. We want to show $\chi_I \cup
\dom(\sigma') \supseteq \chi_O$. Either $\sigma' \in
\lift{\sat}(\structure, c_1\sigma)$ or $\sigma' \in
\lift{\sat}(\structure, c_2\sigma)$. Consider the former case (the
other case is similar). We have $\chi_I \cup \dom(\sigma')
\supseteq \chi_1 \supseteq \chi_1 \cap \chi_2 = \chi_O$, where the
first relation follows from fact~(A).\\

\case ~$\inferrule{\chi_I \vdash c: \chi_O'}{ \chi_I \vdash \exists x.c
  : \chi_O' \backslash \{x\}}$ 

Suppose $\dom(\sigma) \supseteq \chi_I$. By i.h.\ on the premise,
$\lift{\sat}(\structure, c\sigma)$ is defined and $\forall \sigma''
\in \lift{\sat}(\structure, c\sigma)$, $\chi_I \cup \dom(\sigma'')
\supseteq \chi_O'$. Call the latter fact (A). By definition of
$\lift{\sat}$, $\lift{\sat}(\structure, \exists x.c) =
\lift{\sat}(\structure, c) \backslash \{x\}$ is defined.

Pick any $\sigma' \in \lift{\sat}(\structure, c) \backslash \{x\}$. We
want to prove that $\chi_I \cup \dom(\sigma') \supseteq \chi_O'
\backslash \{x\}$. However, $\sigma' \in \lift{\sat}(\structure, c)
\backslash \{x\}$ implies (by definition) that there is a $\sigma''
\in \lift{\sat}(\structure, c)$ such that $\sigma' = \sigma''
\backslash \{x\}$. Thus, $\chi_I \cup \dom(\sigma') = \chi_I \cup
(\dom(\sigma'') \backslash \{x\}) \supseteq \chi_O' \backslash
\{x\}$. The last inclusion follows from fact (A).
\end{proof}


\begin{lem}\label{lem:mode:condition:inclusion}
If $\chi_I \vdash c: \chi_O$, then $\chi_O \subseteq \chi_I \cup
\fv(c)$.
\end{lem}
\begin{proof}
By a straightforward induction on the given derivation of $\chi_I
\vdash c: \chi_O$. 
\end{proof}


\begin{lem}[Mode substitution]\label{lem:subst:mode} The following hold:
\begin{enumerate}
\item If $\chi_I \vdash c: \chi_O$, then $\chi_I \backslash
  \dom(\sigma) \vdash c\sigma: \chi_O \backslash \dom(\sigma)$.
\item If $\chi \vdash \varphi$, then $\chi \backslash \dom(\sigma)
  \vdash \varphi \sigma$.
\end{enumerate}
\end{lem}
\begin{proof}
By induction on the given derivations of $\chi_I \vdash c: \chi_O$ and
$\chi \vdash \varphi$.
\end{proof}


\begin{lem}[Mode weakening]\label{lem:weaken:mode} The following hold:
\begin{enumerate}
\item If $\chi_I \vdash c: \chi_O$ and $\chi_I' \supseteq \chi_I$,
  then there is a $\chi_O' \supseteq \chi_O$ such that $\chi_I' \vdash
  c: \chi_O'$.
\item If $\chi \vdash \varphi$ and $\chi' \supseteq \chi$, then $\chi'
  \vdash \varphi$.
\end{enumerate}
\end{lem}
\begin{proof}
By induction on the given derivations of $\chi_I \vdash c: \chi_O$ and
$\chi \vdash \varphi$.
\end{proof}


\begin{thm}[Totality of $\reduce$; Theorem~\ref{thm:reduce:total}]
  \label{thm:reduce:total:app}
If $\vdash \varphi$ then there is a $\psi$ such that
$\reduce(\structure,\varphi) = \psi$ and $ \vdash \psi$.
\end{thm}
\begin{proof}
We prove a more general result: If $\chi \vdash \varphi$ and
$\dom(\sigma) \supseteq \chi$, then there is a $\psi$ such that
$\reduce(\structure,\varphi\sigma) = \psi\sigma$ and $\chi \vdash
\psi$. The statement of the theorem follows by choosing $\chi = \{\}$
and $\sigma = \bullet$ in this result. We proceed by induction on the
assumed derivation of $\chi \vdash \varphi$, and case analysis of its
last rule.\\

\case ~$\inferrule{\forall k.~\fv(t_k) \subseteq \chi}{ \chi \vdash
  p(t_1,\ldots,t_k)}$

Here, $\varphi = p(t_1,\ldots,t_k)$. Suppose $\dom(\sigma) \supseteq
\chi$. Due to the premise, $p(t_1,\ldots,t_n) \sigma$ is
ground. Hence, $\lift{\sat}(\structure, p(t_1,\ldots,t_n)\sigma)$ is
defined. Depending on whether it is $\symt$, $\symf$, or
$\symu$,\linebreak[6] $\reduce(\structure, p(t_1,\ldots,t_n) \sigma)$
is $\top$, $\bot$ or $p(t_1,\ldots,t_n)\sigma$
respectively. Accordingly, we choose $\psi = \top$, $\psi = \bot$ or
$\psi = p(t_1,\ldots, t_n)$. In each case, $\chi \vdash \psi$.\\

\case ~$\inferrule{ }{ \chi \vdash \top}$

Here, $\varphi = \top$. Suppose $\dom(\sigma) \supseteq
\chi$. Clearly, we can choose $\psi = \top$ because
$\reduce(\structure, \top\sigma) = \top = \psi\sigma$ and $\chi \vdash
\top$, i.e., $\chi \vdash \psi$.\\

\case ~$\inferrule{ }{ \chi \vdash \bot}$

Here, $\varphi = \bot$.  Suppose $\dom(\sigma) \supseteq
\chi$. Clearly, we can choose $\psi = \bot$ because
$\reduce(\structure, \bot\sigma) = \bot = \psi\sigma$ and $\chi \vdash
\bot$, i.e., $\chi \vdash \psi$.\\

\case ~$\inferrule{ \chi \vdash \varphi_1 \\ \chi \vdash \varphi_2}{
  \chi \vdash \varphi_1 \conj \varphi_2}$

Here, $\varphi = \varphi_1 \conj \varphi_2$. Suppose $\dom(\sigma)
\supseteq \chi$. By i.h.\ on the first premise, there is a $\psi_1$
such that $\reduce(\structure, \varphi_1 \sigma) = \psi_1\sigma$ and
$\chi \vdash \psi_1$. Similarly, by i.h.\ on the second premise, there
is a $\psi_2$ such that $\reduce(\structure, \varphi_2 \sigma) =
\psi_2\sigma$ and $\chi \vdash \psi_2$. By definition of $\reduce$,
$\reduce (\structure, \varphi \sigma) = \reduce(\structure, (\varphi_1
\conj \varphi_2) \sigma) = \reduce(\structure, \varphi_1\sigma) \conj
\reduce(\structure, \varphi_2\sigma) = \psi_1\sigma \conj
\psi_2\sigma$. Further, $\chi \vdash \psi_1 \conj \psi_2$ follows from
$\chi \vdash \psi_1$ and $\chi \vdash \psi_2$. So we can choose $\psi
= \psi_1 \conj \psi_2$.\\

\case ~$\inferrule{ \chi \vdash \varphi_1 \\ \chi \vdash \varphi_2}{
  \chi \vdash \varphi_1 \disj \varphi_2}$

Here, $\varphi = \varphi_1 \disj \varphi_2$. Suppose $\dom(\sigma)
\supseteq \chi$. By i.h.\ on the first premise, there is a $\psi_1$
such that $\reduce(\structure, \varphi_1 \sigma) = \psi_1\sigma$ and
$\chi \vdash \psi_1$. Similarly, by i.h.\ on the second premise, there
is a $\psi_2$ such that $\reduce(\structure, \varphi_2 \sigma) =
\psi_2\sigma$ and $\chi \vdash \psi_2$. By definition of $\reduce$,
$\reduce (\structure, \varphi \sigma) = \reduce(\structure, (\varphi_1
\disj \varphi_2) \sigma) = \reduce(\structure, \varphi_1\sigma) \disj
\reduce(\structure, \varphi_2\sigma) = \psi_1\sigma \disj
\psi_2\sigma$. Further, $\chi \vdash \psi_1 \disj \psi_2$ follows from
$\chi \vdash \psi_1$ and $\chi \vdash \psi_2$. So we can choose $\psi
= \psi_1 \disj \psi_2$.\\

\case ~$\inferrule{ \chi \vdash c: \chi_O \\ \vec{x} \subseteq \chi_O
  \\ \fv(c) \subseteq \chi \cup \vec{x} \\ \chi_O \vdash \varphi'}{
  \chi \vdash \forall \vec{x}. (c \imp \varphi')}$

Here, $\varphi = \forall \vec{x}. (c \imp \varphi')$. Suppose
$\dom(\sigma) \supseteq \chi$. By
Theorem~\ref{thm:sat:total:app} on the first premise, there is
a set $\{\sigma_1,\ldots, \sigma_n\} = \lift{\sat}(\structure,
c\sigma)$ such that for each $\sigma_i$, $\chi \cup \dom(\sigma_i)
\supseteq \chi_O$. Call the latter fact~(A). From the second premise
and fact~(A) we also derive that $\chi \cup \dom(\sigma_i) \supseteq
\vec{x}$. Since $\vec{x}$ must be chosen fresh in the premise, this
also implies that $\dom(\sigma_i) \supseteq \vec{x}$. Consequently,
$\sigma_i(\vec{x})$ is defined. Let $\sigma_i(\vec{x}) =\vec{t_i}$ and
let $S = \{\vec{t_1}, \ldots, \vec{t_n}\}$. Further, note that by
Lemma~\ref{lem:mode:condition:inclusion} on the first premise, $\chi_O
\subseteq \chi\cup \fv(c)$. Hence, from the third premise we obtain
$\chi_O \subseteq \chi \cup \chi \cup \vec{x} = \chi \cup
\vec{x}$. So, $\dom(\sigma) \cup \vec{x} \supseteq \chi \cup \vec{x}
\supseteq \chi_O$. Call this fact~(B). From the i.h.\ applied to the
last premise and fact~(B) we get the existence of $\psi_i$ such that
$\reduce(\structure, \varphi'\sigma[\vec{t_i}/\vec{x}]) =
\psi_i\sigma[\vec{t_i}/\vec{x}]$ and $\chi_O \vdash \psi_i$. Call this
fact~(C).

By definition of $\reduce$, we obtain $\reduce(\structure, \varphi
\sigma) = \psi_1 \sigma [\vec{t_1}/\vec{x}] \conj \ldots \conj \psi_n
\sigma [\vec{t_n}/\vec{x}] \conj \psi'\sigma$, where $\psi' = \forall
\vec{x}.((c \conj \vec{x} \not \in S) \imp \varphi')$. Choose $\psi =
\psi_1[\vec{t_1}/\vec{x}] \conj \ldots \conj \psi_n[\vec{t_n}/\vec{x}]
\conj \psi'$. It only remains to show that $\chi \vdash \psi$. This is
equivalent to showing that $\chi \vdash \psi_i[\vec{t_i}/\vec{x}]$ and
$\chi \vdash \psi'$. The latter, which is equal to $\chi \vdash
\forall \vec{x}.((c \conj \vec{x} \not \in S) \imp \varphi')$, follows
from the four premises of the rule above. It remains to show that
$\chi \vdash \psi_i[\vec{t_i}/\vec{x}]$. Applying
Lemma~\ref{lem:subst:mode} to fact(C), we derive that $\chi_O
\backslash \vec{x} \vdash \psi_i[\vec{t_i}/\vec{x}]$. Since we already
derived that $\chi_O \subseteq \chi \cup \vec{x}$, we also have
$\chi_O \backslash \vec{x} \subseteq \chi$. Hence, by
Lemma~\ref{lem:weaken:mode}, we get $\chi \vdash
\psi_i[\vec{t_i}/\vec{x}]$, as required.\\

\case ~$\inferrule{ \chi \vdash c: \chi_O \\ \vec{x} \subseteq \chi_O
  \\ \fv(c) \subseteq \chi \cup \vec{x} \\ \chi_O \vdash \varphi'}{
  \chi \vdash \exists \vec{x}. (c \conj \varphi')}$

Here, $\varphi = \exists \vec{x}. (c \conj \varphi')$. Suppose
$\dom(\sigma) \supseteq \chi$. By
Theorem~\ref{thm:sat:total:app} on the first premise, there is
a set $\{\sigma_1,\ldots, \sigma_n\} = \lift{\sat}(\structure,
c\sigma)$ such that for each $\sigma_i$, $\chi \cup \dom(\sigma_i)
\supseteq \chi_O$. Call the latter fact~(A). From the second premise
and fact~(A) we also derive that $\chi \cup \dom(\sigma_i) \supseteq
\vec{x}$. Since $\vec{x}$ must be chosen fresh in the premise, this
also implies that $\dom(\sigma_i) \supseteq \vec{x}$. Consequently,
$\sigma_i(\vec{x})$ is defined. Let $\sigma_i(\vec{x}) =\vec{t_i}$ and
let $S = \{\vec{t_1}, \ldots, \vec{t_n}\}$. Further, note that by
Lemma~\ref{lem:mode:condition:inclusion} on the first premise, $\chi_O
\subseteq \chi\cup \fv(c)$. Hence, from the third premise we obtain
$\chi_O \subseteq \chi \cup \chi \cup \vec{x} = \chi \cup
\vec{x}$. So, $\dom(\sigma) \cup \vec{x} \supseteq \chi \cup \vec{x}
\supseteq \chi_O$. Call this fact~(B). From the i.h.\ applied to the
last premise and fact~(B) we get the existence of $\psi_i$ such that
$\reduce(\structure, \varphi'\sigma[\vec{t_i}/\vec{x}]) =
\psi_i\sigma[\vec{t_i}/\vec{x}]$ and $\chi_O \vdash \psi_i$. Call this
fact~(C).

By definition of $\reduce$, we obtain $\reduce(\structure, \varphi
\sigma) = \psi_1 \sigma [\vec{t_1}/\vec{x}] \disj \ldots \disj \psi_n
\sigma [\vec{t_n}/\vec{x}] \disj \psi'\sigma$, where $\psi' = \exists
\vec{x}.((c \conj \vec{x} \not \in S) \conj \varphi')$. Choose $\psi =
\psi_1[\vec{t_1}/\vec{x}] \disj \ldots \disj \psi_n[\vec{t_n}/\vec{x}]
\disj \psi'$. It only remains to show that $\chi \vdash \psi$. This is
equivalent to showing that $\chi \vdash \psi_i[\vec{t_i}/\vec{x}]$ and
$\chi \vdash \psi'$. The latter, which is equal to $\chi \vdash
\exists \vec{x}.((c \conj \vec{x} \not \in S) \conj \varphi')$,
follows from the four premises of the rule above. It remains to show
that $\chi \vdash \psi_i[\vec{t_i}/\vec{x}]$. Applying
Lemma~\ref{lem:subst:mode} to fact(C), we derive that $\chi_O
\backslash \vec{x} \vdash \psi_i[\vec{t_i}/\vec{x}]$. Since we already
derived that $\chi_O \subseteq \chi \cup \vec{x}$, we also have
$\chi_O \backslash \vec{x} \subseteq \chi$. Hence, by
Lemma~\ref{lem:weaken:mode}, we get $\chi \vdash
\psi_i[\vec{t_i}/\vec{x}]$, as required.
\end{proof}


\begin{lem}[Totality of $\atoms$]\label{lem:atoms:total}
Suppose $\chi \vdash \varphi$ and $\dom(\sigma) \supseteq \chi$. Then,
$\atoms(\structure, \varphi\sigma)$ is defined and ground.
\end{lem}
\begin{proof}
By induction on the given derivation of $\chi \vdash \varphi$ and case
analysis of its last rule.\\

\case ~$\inferrule{\forall k.~\fv(t_k) \subseteq \chi}{ \chi \vdash
  p(t_1,\ldots,t_k)}$

Here $\varphi = p(t_1,\ldots,t_k)$. From the premise and given
condition $\dom(\sigma) \supseteq \chi$, we know that
$p(t_1,\ldots,t_k) \sigma$ is ground. Clearly, then
$\atoms(\structure, p(t_1,\ldots,t_k)\sigma) =
\{p(t_1,\ldots,t_k)\sigma\}$ is defined and ground.\\

\case ~$\inferrule{ }{ \chi \vdash \top}$

Here $\varphi = \top$. So $\atoms(\structure, \varphi\sigma) =
\atoms(\structure, \top) = \{\}$ is defined and ground.\\

\case ~$\inferrule{ }{ \chi \vdash \bot}$

Here $\varphi = \bot$. So $\atoms(\structure, \varphi\sigma) =
\atoms(\structure, \bot) = \{\}$ is defined and ground.\\

\case ~$\inferrule{ \chi \vdash \varphi_1 \\ \chi \vdash \varphi_2}{
  \chi \vdash \varphi_1 \conj \varphi_2}$

Here $\varphi = \varphi_1 \conj \varphi_2$. By the i.h.\ applied to
the premises, $\atoms(\structure, \varphi_i \sigma)$ for $i = 1,2$ is
defined and ground. It follows that $\atoms(\structure, \varphi
\sigma) = \atoms(\structure, \varphi_1 \sigma \conj \varphi_2\sigma) =
\atoms(\structure, \varphi_1 \sigma) \cup \atoms(\structure, \varphi_2
\sigma)$ is also defined and ground.\\

\case ~$\inferrule{ \chi \vdash \varphi_1 \\ \chi \vdash \varphi_2}{
  \chi \vdash \varphi_1 \disj \varphi_2}$

Here $\varphi = \varphi_1 \disj \varphi_2$. By the i.h.\ applied to
the premises, $\atoms(\structure, \varphi_i \sigma)$ for $i = 1,2$ is
defined and ground. It follows that $\atoms(\structure, \varphi
\sigma) = \atoms(\structure, \varphi_1 \sigma \disj \varphi_2\sigma) =
\atoms(\structure, \varphi_1 \sigma) \cup \atoms(\structure, \varphi_2
\sigma)$ is also defined and ground.\\

\case ~$\inferrule{ \chi \vdash c: \chi_O \\ \vec{x} \subseteq \chi_O
  \\ \fv(c) \subseteq \chi \cup \vec{x} \\ \chi_O \vdash \varphi'}{
  \chi \vdash \forall \vec{x}. (c \imp \varphi')}$

Here $\varphi = \forall \vec{x}. (c \imp \varphi')$. By
Theorem~\ref{thm:sat:total:app} on the first premise and the
given condition $\dom(\sigma) \supseteq \chi$,
$\lift{\sat}(\structure, c\sigma)$ is defined and for all $\sigma' \in
\lift{\sat}(\structure, c\sigma)$, $\chi \cup \dom(\sigma') \supseteq
\chi_O$. The latter implies that for all $\sigma' \in
\lift{\sat}(\structure, c\sigma)$, $\dom(\sigma\sigma') \supseteq
\chi_O$. By i.h.\ on the last premise, for each $\sigma' \in
\lift{\sat}(\structure, c\sigma)$, $\atoms(\structure,
\varphi'\sigma\sigma')$ is defined and ground. Hence, by definition,
$\atoms(\structure, \varphi\sigma) = \bigcup_{\sigma' \in
  \lift{\sat}(\structure, c\sigma)} \atoms(\structure,
\varphi'\sigma\sigma')$ is defined and ground.\\

\case ~$\inferrule{ \chi \vdash c: \chi_O \\ \vec{x} \subseteq \chi_O
  \\ \fv(c) \subseteq \chi \cup \vec{x} \\ \chi_O \vdash \varphi'}{
  \chi \vdash \exists \vec{x}. (c \conj \varphi')}$

Here $\varphi = \exists \vec{x}. (c \conj \varphi')$. By
Theorem~\ref{thm:sat:total:app} on the first premise and the
given condition $\dom(\sigma) \supseteq \chi$,
$\lift{\sat}(\structure, c\sigma)$ is defined and for all $\sigma' \in
\lift{\sat}(\structure, c\sigma)$, $\chi \cup \dom(\sigma') \supseteq
\chi_O$. The latter implies that for all $\sigma' \in
\lift{\sat}(\structure, c\sigma)$, $\dom(\sigma\sigma') \supseteq
\chi_O$. By i.h.\ on the last premise, for each $\sigma' \in
\lift{\sat}(\structure, c\sigma)$, $\atoms(\structure,
\varphi'\sigma\sigma')$ is defined and ground. Hence, by definition,
$\atoms(\structure, \varphi\sigma) = \bigcup_{\sigma' \in
  \lift{\sat}(\structure, c\sigma)} \atoms(\structure,
\varphi'\sigma\sigma')$ is defined and ground.
\end{proof}


\begin{thm}[Minimality; Theorem~\ref{thm:minimality}]
  \label{thm:minimality:app}
Suppose $\vdash \varphi$ and $\reduce(\structure, \varphi) =
\psi$. Then $\atoms(\structure, \psi) \subseteq \atoms(\structure,
\varphi) \cap \{P ~|~ \rho_\structure(P) = \symu\}$.
\end{thm}
\begin{proof}
By Lemma~\ref{lem:atoms:total}, $\atoms(\structure, \varphi)$ is
defined. Further, by Theorem~\ref{thm:reduce:total:app}, $\vdash
\psi$, so $\atoms(\structure, \psi)$ is also defined. Hence, the
statement of the theorem makes sense. We prove the relation
$\atoms(\structure, \psi) \subseteq \atoms(\structure, \varphi) \cap
\{P ~|~ \rho_\structure(P) = \symu\}$ by induction on $\varphi$ and
case analysis of its form. Let $U = \{P ~|~ \rho_\structure(P) =
\symu\}$. We want to show that $\atoms(\structure, \psi) \subseteq
\atoms(\structure, \varphi) \cap U$.\\

\case ~$\varphi = P$ where $P$ is either a subjective or an objective
atom. We perform a sub-case analysis on $\rho_\structure(P)$. \\

\subcase ~$\rho_\structure(P) = \symt$. Then, $\psi = \top$. So,
trivially, $\atoms(\psi) = \{\} \subseteq \atoms(\structure, \varphi)
\cap U$.\\

\subcase ~$\rho_\structure(P) = \symf$. Then, $\psi = \bot$. So,
trivially, $\atoms(\psi) = \{\} \subseteq \atoms(\structure, \varphi)
\cap U$.\\

\subcase ~$\rho_\structure(P) = \symu$. Then, $\psi = P$. Further, in
this case, $\atoms(\structure,\psi) = \{P\} = \atoms(\structure,
\varphi)$ and $P \in U$ (the latter by definition of $U$). Clearly,
$\atoms(\structure, \psi) \subseteq \atoms(\structure, \varphi) \cap
U$.\\

\case ~$\varphi = \top$. Here, $\psi = \top$. So, trivially,
$\atoms(\psi) = \{\} \subseteq \atoms(\structure, \varphi) \cap U$.\\

\case ~$\varphi = \bot$. Here, $\psi = \bot$. So, trivially,
$\atoms(\psi) = \{\} \subseteq \atoms(\structure, \varphi) \cap U$.\\

\case ~$\varphi = \varphi_1 \conj \varphi_2$. Then, $\psi =
\reduce(\structure, \varphi_1) \conj \reduce(\structure,
\varphi_2)$. By inversion on the derivation of $\vdash \varphi$, we
know that $\vdash \varphi_1$ and $\vdash \varphi_2$. Hence, by the
i.h., for $i=1,2$, $\atoms(\structure, \reduce(\structure, \varphi_i))
\subseteq \atoms(\structure, \varphi_i) \cap U$. Thus, we have,
$\atoms(\structure, \psi) = \atoms(\structure, \reduce(\structure,
\varphi_1)) \cup \atoms(\structure, \reduce(\structure, \varphi_2))
\subseteq (\atoms(\structure, \varphi_1) \cap U) \cup
(\atoms(\structure, \varphi_2) \cap U) = (\atoms(\structure,
\varphi_1) \cup \atoms(\structure, \varphi_2)) \cap U =
\atoms(\structure, \varphi) \cap U$.\\

\case ~$\varphi = \varphi_1 \disj \varphi_2$. Then, $\psi =
\reduce(\structure, \varphi_1) \disj \reduce(\structure,
\varphi_2)$. By inversion on the derivation of $\vdash \varphi$, we
know that $\vdash \varphi_1$ and $\vdash \varphi_2$. Hence, by the
i.h., for $i=1,2$, $\atoms(\structure, \reduce(\structure, \varphi_i))
\subseteq \atoms(\structure, \varphi_i) \cap U$. Thus, we have,
$\atoms(\structure, \psi) = \atoms(\structure, \reduce(\structure,
\varphi_1)) \cup \atoms(\structure, \reduce(\structure, \varphi_2))
\subseteq (\atoms(\structure, \varphi_1) \cap U) \cup
(\atoms(\structure, \varphi_2) \cap U) = (\atoms(\structure,
\varphi_1) \cup \atoms(\structure, \varphi_2)) \cap U =
\atoms(\structure, \varphi) \cap U$.\\

\case ~$\varphi = \forall \vec{x}.(c\imp \varphi')$. Then,
\[
\psi = \reduce(\structure, \varphi) ~=~
\begin{array}[t]{@{}l}
\mbox{let}\\
~\{\sigma_1,\ldots,\sigma_n\} \leftarrow \lift{\sat}(\structure,c)\\
~\{ \vec{t_i} \leftarrow \sigma_i(\vec{x}) \}_{i=1}^n\\
~S \leftarrow \{\vec{t_1},\ldots,\vec{t_n}\}\\
~\{ \psi_i \leftarrow \reduce(\structure,
\varphi'[\vec{t_i}/\vec{x}]) \}_{i=1}^n\\
~\psi' \leftarrow \forall \vec{x}.((c \conj \vec{x} \not \in S) \imp \varphi')\\
\mbox{return } \\
~~~ \psi_1 \conj \ldots \conj \psi_n \conj \psi'
\end{array} \]

By inversion on the given derivation of $\vdash \varphi$, we know that
there is a $\chi_O$ such that (1)~$\{\} \vdash c: \chi_O$,
(2)~$\vec{x} \subseteq \chi_O$, (3)~$\fv(c) \subseteq \vec{x}$, and
(4)~$\chi_O \vdash \varphi'$.  By
Lemma~\ref{lem:mode:condition:inclusion} on (1), $\chi_O \subseteq
\fv(c)$. From this, (2), and (3), it follows that $\vec{x} = \fv(c) =
\chi_O$. Call this fact~(A). Using Lemma~\ref{lem:subst:mode} on~(4),
we get $\chi_O \backslash \vec{x} \vdash
\varphi'[\vec{t_i}/\vec{x}]$. This and fact~(A) imply that $\vdash
\varphi'[\vec{t_i}/\vec{x}]$. Call this fact~(B). By the i.h.\ on
fact~(B) and $\psi_i \leftarrow \reduce(\structure,
\varphi'[\vec{t_i}/\vec{x}])$, we get that $\atoms(\structure, \psi_i)
\subseteq \atoms(\structure, \varphi'[\vec{t_i}/\vec{x}]) \cap
U$. Call this fact~(C).

Next, $\lift{\sat}(\structure, (c \conj \vec{x} \not \in S)) =
\bigcup_{\sigma' \in \lift{\sat}(\structure, c)} (\sigma' +
\lift{\sat}(\structure, \sigma'(\vec{x}) \not \in S)) =
\bigcup_{i=1}^n (\sigma_i + \lift{\sat}(\structure, \vec{t_i} \not \in
S)) = \bigcup_{i=1}^n (\sigma_i + \{\}) = \{\}$. Hence, by definition,
$\atoms(\structure, \psi') = \atoms(\structure, \forall \vec{x}.((c
\conj \vec{x} \not \in S) \imp \varphi')) = \{\}$. Call this fact~(D).

Also, $\atoms(\structure, \varphi) = \atoms(\structure, \forall
\vec{x}.(c\imp \varphi')) = \bigcup_{\sigma \in
  \lift{\sat}(\structure, c)} \atoms(\structure, \varphi'\sigma) =
\bigcup_{i=1}^n \atoms(\structure, \varphi'\sigma_i) = \bigcup_{i=1}^n
\atoms(\structure, \varphi'[\vec{t_i}/\vec{x}])$ (the last equality
follows from $\fv(\varphi') \subseteq \vec{x}$, which in turn follows
from fact~(B)). Call this fact~(E).

Finally, we have,
\[\begin{array}{llll}
\atoms(\structure, \psi) & = & \atoms(\structure, \psi_1
\conj \ldots \conj \psi_n \conj \psi')\\

& = & \atoms(\structure, \psi') \cup (\bigcup_{i=1}^n
\atoms(\structure, \psi_i)) & \mbox{(Defn. of $\atoms$)}\\

& = & \{\} \cup (\bigcup_{i=1}^n
\atoms(\structure, \psi_i)) & \mbox{(Fact~(D))}\\

& = & \bigcup_{i=1}^n \atoms(\structure, \psi_i) \\

& \subseteq & \bigcup_{i=1}^n (\atoms(\structure,
\varphi'[\vec{t_i}/\vec{x}]) \cap U) & \mbox{(Fact~(C))}\\

& = & (\bigcup_{i=1}^n \atoms(\structure,
\varphi'[\vec{t_i}/\vec{x}])) \cap U \\

& = & \atoms(\structure, \varphi) \cap U & \mbox{(Fact~(E))}
\end{array}\]

\case ~$\varphi = \exists \vec{x}.(c\conj \varphi')$. Then,
\[
\psi = \reduce(\structure, \varphi) ~=~
\begin{array}[t]{@{}l}
\mbox{let}\\
~\{\sigma_1,\ldots,\sigma_n\} \leftarrow \lift{\sat}(\structure,c)\\
~\{ \vec{t_i} \leftarrow \sigma_i(\vec{x}) \}_{i=1}^n\\
~S \leftarrow \{\vec{t_1},\ldots,\vec{t_n}\}\\
~\{ \psi_i \leftarrow \reduce(\structure,
\varphi'[\vec{t_i}/\vec{x}]) \}_{i=1}^n\\
~\psi' \leftarrow \exists \vec{x}.((c \conj \vec{x} \not \in S) \conj \varphi')\\
\mbox{return } \\
~~~ \psi_1 \disj \ldots \disj \psi_n \disj \psi'
\end{array} \]

By inversion on the given derivation of $\vdash \varphi$, we know that
there is a $\chi_O$ such that (1)~$\{\} \vdash c: \chi_O$,
(2)~$\vec{x} \subseteq \chi_O$, (3)~$\fv(c) \subseteq \vec{x}$, and
(4)~$\chi_O \vdash \varphi'$.  By
Lemma~\ref{lem:mode:condition:inclusion} on (1), $\chi_O \subseteq
\fv(c)$. From this, (2), and (3), it follows that $\vec{x} = \fv(c) =
\chi_O$. Call this fact~(A). Using Lemma~\ref{lem:subst:mode} on~(4),
we get $\chi_O \backslash \vec{x} \vdash
\varphi'[\vec{t_i}/\vec{x}]$. This and fact~(A) imply that $\vdash
\varphi'[\vec{t_i}/\vec{x}]$. Call this fact~(B). By the i.h.\ on
fact~(B) and $\psi_i \leftarrow \reduce(\structure,
\varphi'[\vec{t_i}/\vec{x}])$, we get that $\atoms(\structure, \psi_i)
\subseteq \atoms(\structure, \varphi'[\vec{t_i}/\vec{x}]) \cap
U$. Call this fact~(C).

Next, $\lift{\sat}(\structure, (c \conj \vec{x} \not \in S)) =
\bigcup_{\sigma' \in \lift{\sat}(\structure, c)} (\sigma' +
\lift{\sat}(\structure, \sigma'(\vec{x}) \not \in S)) =
\bigcup_{i=1}^n (\sigma_i + \lift{\sat}(\structure, \vec{t_i} \not \in
S)) = \bigcup_{i=1}^n (\sigma_i + \{\}) = \{\}$. Hence, by definition,
$\atoms(\structure, \psi') = \atoms(\structure, \exists \vec{x}.((c
\conj \vec{x} \not \in S) \conj \varphi')) = \{\}$. Call this fact~(D).

Also, $\atoms(\structure, \varphi) = \atoms(\structure, \exists
\vec{x}.(c \conj \varphi')) = \bigcup_{\sigma \in
  \lift{\sat}(\structure, c)} \atoms(\structure, \varphi'\sigma) =
\bigcup_{i=1}^n \atoms(\structure, \varphi'\sigma_i) = \bigcup_{i=1}^n
\atoms(\structure, \varphi'[\vec{t_i}/\vec{x}])$ (the last equality
follows from $\fv(\varphi') \subseteq \vec{x}$, which in turn follows
from fact~(B)). Call this fact~(E).

Finally, we have,
\[\begin{array}{llll}
\atoms(\structure, \psi) & = & \atoms(\structure, \psi_1
\disj \ldots \disj \psi_n \disj \psi')\\

& = & \atoms(\structure, \psi') \cup (\bigcup_{i=1}^n
\atoms(\structure, \psi_i)) & \mbox{(Defn. of $\atoms$)}\\

& = & \{\} \cup (\bigcup_{i=1}^n
\atoms(\structure, \psi_i)) & \mbox{(Fact~(D))}\\

& = & \bigcup_{i=1}^n \atoms(\structure, \psi_i) \\

& \subseteq & \bigcup_{i=1}^n (\atoms(\structure,
\varphi'[\vec{t_i}/\vec{x}]) \cap U) & \mbox{(Fact~(C))}\\

& = & (\bigcup_{i=1}^n \atoms(\structure,
\varphi'[\vec{t_i}/\vec{x}])) \cap U \\

& = & \atoms(\structure, \varphi) \cap U & \mbox{(Fact~(E))}
\end{array}\]
\end{proof}

\section{Proofs from Section~\ref{sec:instances}}
\label{app:instances}

This appendix contains proofs of theorems presented in
Section~\ref{sec:instances}.

\begin{lem}\label{lem:oc:noquant}
Suppose $\psi$ does not contain any quantifiers or objective
atoms. Then, $\psi \rewrite^* \psi'$ such that (1)~$\psi'$ is either
$\top$, or $\bot$, or contains only subjective atoms and the
connectives $\conj$, $\disj$, and (2)~For all structures $\structure$,
$\structure \models \psi$ iff $\structure \models \psi'$ and
$\structure \models \dual{\psi}$ iff $\structure \models
\dual{\psi'}$.
\end{lem}
\begin{proof}
By induction on $\psi$. If $\psi$ is either $\top$, $\bot$, or $P_S$,
we can choose $\psi' = \psi$. 

If $\psi = \psi_1 \conj \psi_2$, then we inductively rewrite both
$\psi_1$ and $\psi_2$ to $\psi_1'$ and $\psi_2'$, respectively. Thus,
$\psi_1 \conj \psi_2 \rewrite^* \psi_1' \conj \psi_2'$. If either
$\psi_1'$ or $\psi_2'$ equals $\bot$, then $\psi_1' \conj \psi_2'
\rewrite \bot$ and we choose $\psi' = \bot$. If $\psi_1' = \top$, then
$\psi_1' \conj \psi_2' \rewrite \psi_2'$, so we can choose $\psi' =
\psi_2'$. Similarly, if $\psi_2' = \top$, then $\psi_1' \conj \psi_2'
\rewrite \psi_1'$, so we can choose $\psi' = \psi_1'$. Finally, if
both $\psi_1'$ and $\psi_2'$ contain only subjective atoms and
connectives $\conj$, $\disj$, then we choose $\psi' = \psi_1' \conj
\psi_2'$.

The case of $\psi = \psi_1 \disj \psi_2$ is similarly handles. No
other cases apply. 
\end{proof}


\begin{lem}\label{lem:oc:restriction:em}
If $\structure$ is objectively-complete, then for all restrictions
$c$, either $\structure \models c$ or $\structure \models \dual{c}$.
\end{lem}
\begin{proof}
By induction on $c$. 
\end{proof}


\begin{lem}\label{lem:oc:restriction:extend}
If $\structure$ is objectively-complete and $\structure' \geq
\structure$, then for all restrictions $c$, $\structure' \models c$
iff $\structure \models c$.
\end{lem}
\begin{proof}
Suppose $\structure' \geq \structure$. Observe that because
$\structure$ is objectively-complete, $\structure'$ and $\structure$
agree on valuation of objective atoms, which are the only atoms in
$c$. The result now follows by a straightforward induction on $c$.
\end{proof}


\begin{thm}[Theorem~\ref{thm:oc}]\label{thm:oc:app}
Suppose $\structure$ is objectively-complete, $\vdash \varphi$ and
$\psi = \reduce(\structure, \varphi)$. Then $\psi \rewrite^* \psi'$,
where (1)~$\psi'$ is either $\top$, or $\bot$, or contains only subjective
atoms and the connectives $\conj$, $\disj$, and (2)~For all
$\structure' \geq \structure$, $\structure' \models \psi$ iff
$\structure' \models \psi'$ and $\structure' \models \dual{\psi}$ iff
$\structure' \models \dual{\psi'}$.
\end{thm}
\begin{proof}
By induction on $\varphi$ and case analysis of its form. Define
$\pred{simp}(\psi')$ to mean statement~(1) of the theorem, i.e., that
$\psi'$ is either $\top$, or $\bot$, or contains only subjective atoms
and the connectives $\conj$, $\disj$. Define $\pred{equiv}(\structure,
\psi, \psi')$ to mean statement~(2) of the theorem, i.e., for all
$\structure' \geq \structure$, $\structure' \models \psi$ iff
$\structure' \models \psi'$ and $\structure' \models \dual{\psi}$ iff
$\structure' \models \dual{\psi'}$. \\

\case ~$\varphi = P_O$. In this case, $\rho_\structure(P_O) \in
\{\symt, \symf\}$ and, accordingly, $\psi = \top$ or $\psi = \bot$. So
we can choose $\psi' = \psi$ to trivially satisfy both~(1) and~(2).\\

\case ~$\varphi = P_S$. In this case $\psi = \top$ or $\psi = \bot$ or
$\psi = P_S$. So we can choose $\psi' = \psi$ to trivially
satisfy both~(1) and~(2). \\

\case ~$\varphi = \top$. Then, $\psi = \top$. We choose $\psi' = \psi$
to trivially satisfy both~(1) and~(2).\\

\case ~$\varphi = \bot$. Then, $\psi = \bot$. We choose $\psi' = \psi$
to trivially satisfy both~(1) and~(2).\\

\case ~$\varphi = \varphi_1 \conj \varphi_2$. Then, $\psi = \psi_1
\conj \psi_2$, where $\psi_i = \reduce(\structure, \varphi_i)$ for $i
= 1,2$. By inversion on the given derivation of $\vdash \varphi$, we
deduce $\vdash \varphi_1$ and $\vdash \varphi_2$. Hence, from the
i.h., $\psi_i \rewrite^* \psi_i'$ where $\pred{simp}(\psi_i')$ and
$\pred{equiv}(\structure, \psi_i, \psi_i')$. The last fact implies
that $\pred{equiv}(\structure, \psi, \psi_1' \conj \psi_2')$. Further,
$\psi = \psi_1 \conj \psi_2 \rewrite^* \psi_1' \conj \psi_2'$. Using
Lemma~\ref{lem:oc:noquant}, we obtain a $\psi'$ such that $\psi_1'
\conj \psi_2' \rewrite^* \psi'$, $\pred{simp}(\psi')$ and
$\pred{equiv}(\structure, \psi_1' \conj \psi_2', \psi')$. The last
fact and $\pred{equiv}(\structure, \psi, \psi_i' \conj \psi_2')$ imply
$\pred{equiv}(\structure, \psi, \psi')$. So $\psi'$ satisfies all our
requirements.\\

\case ~$\varphi = \varphi_1 \disj \varphi_2$. Then, $\psi = \psi_1
\disj \psi_2$, where $\psi_i = \reduce(\structure, \varphi_i)$ for $i
= 1,2$. By inversion on the given derivation of $\vdash \varphi$, we
deduce $\vdash \varphi_1$ and $\vdash \varphi_2$. Hence, from the
i.h., $\psi_i \rewrite^* \psi_i'$ where $\pred{simp}(\psi_i')$ and
$\pred{equiv}(\structure, \psi_i, \psi_i')$. The last fact implies
that $\pred{equiv}(\structure, \psi, \psi_1' \disj \psi_2')$. Further,
$\psi = \psi_1 \disj \psi_2 \rewrite^* \psi_1' \disj \psi_2'$. Using
Lemma~\ref{lem:oc:noquant}, we obtain a $\psi'$ such that $\psi_1'
\disj \psi_2' \rewrite^* \psi'$, $\pred{simp}(\psi')$ and
$\pred{equiv}(\structure, \psi_1' \disj \psi_2', \psi')$. The last
fact and $\pred{equiv}(\structure, \psi, \psi_i' \disj \psi_2')$ imply
$\pred{equiv}(\structure, \psi, \psi')$. So $\psi'$ satisfies all our
requirements.\\

\case ~$\varphi = \forall \vec{x}. (c \imp \varphi')$. Then, $\psi$ is
calculated as follows:
\[
\psi = \reduce(\structure, \varphi) ~=~
\begin{array}[t]{@{}l}
\mbox{let}\\
~\{\sigma_1,\ldots,\sigma_n\} \leftarrow \lift{\sat}(\structure,c)\\
~\{ \vec{t_i} \leftarrow \sigma_i(\vec{x}) \}_{i=1}^n\\
~S \leftarrow \{\vec{t_1},\ldots,\vec{t_n}\}\\
~\{ \psi_i \leftarrow \reduce(\structure,
\varphi'[\vec{t_i}/\vec{x}]) \}_{i=1}^n\\
~\psi'' \leftarrow \forall \vec{x}.((c \conj \vec{x} \not \in S) \imp \varphi')\\
\mbox{return } \\
~~~ \psi_1 \conj \ldots \conj \psi_n \conj \psi''
\end{array} \]

By inversion on the given derivation of $\vdash \varphi$, we know that
there is a $\chi_O$ such that (1)~$\{\} \vdash c: \chi_O$,
(2)~$\vec{x} \subseteq \chi_O$, (3)~$\fv(c) \subseteq \vec{x}$, and
(4)~$\chi_O \vdash \varphi'$. By
Lemma~\ref{lem:mode:condition:inclusion} on (1), $\chi_O \subseteq
\fv(c)$. From this, (2), and (3), it follows that $\vec{x} = \fv(c) =
\chi_O$. Call this fact~(A). Note also that by
Theorem~\ref{thm:sat:total:app}, $\dom(\sigma_i) \supseteq \chi_O =
\vec{x}$. Call this fact~(B).

Next, we show that $\pred{equiv}(\structure, \psi'', \top)$. Since for
all $\structure'$, $\structure' \models \top$, it suffices to show
that for all $\structure' \geq \structure$, $\structure' \models
\forall \vec{x}.((c \conj \vec{x} \not \in S) \imp \varphi')$. By
definition of $\models$, it suffices to prove that for all $\vec{t}$
and $\structure' \geq \structure$, $\structure' \models
\dual{c[\vec{t}/\vec{x}] \conj \vec{t} \not \in S}$, i.e.,
$\structure' \models \dual{c[\vec{t}/\vec{x}]} \disj \vec{t} \in
S$. If $\vec{t} = \vec{t_i}$ for some $i$, then $\structure' \models
\vec{t} \in S$ by definition of $S$, so we are done. Hence, we need
only consider the case where $\vec{t} \not \in S$. In this case we
show that $\structure' \models \dual{c[\vec{t}/\vec{x}]}$. By
Lemma~\ref{lem:oc:restriction:em}, this is implied by $\structure'
\not \models c[\vec{t}/\vec{x}]$, so we show the latter. Suppose, for
the sake of contradiction, that $\structure' \models
c[\vec{t}/\vec{x}]$. By Lemma~\ref{lem:oc:restriction:extend},
$\structure \models c[\vec{t}/\vec{x}]$. Hence, by
Theorem~\ref{thm:sat:correct:app}, there is a $\sigma \in
\lift{\sat}(\structure, c)$ such that $[\vec{x} \mapsto \vec{t}] \geq
\sigma$.  $\sigma \in \lift{\sat}(\structure, c)$ forces $\sigma =
\sigma_i$ for some $i$ and, by fact~(B), $\vec{t} = \vec{t_i}$. Hence,
$\vec{t} = \vec{t_i} \in S$, a contradiction. Therefore,
$\pred{equiv}(\structure, \psi'', \top)$. Call this fact~(C).

By Lemma~\ref{lem:subst:mode} on~(4), we derive $\chi_O \backslash
\vec{x} \vdash \varphi'[\vec{t}/\vec{x}]$. Using fact~(A), we have
$\vdash \varphi'[\vec{t}/\vec{x}]$. Applying the i.h.\ to this and
$\psi_i \leftarrow \reduce(\structure, \varphi'[\vec{t_i}/\vec{x}])$,
we know that there is a $\psi_i'$ such that $\psi_i \rewrite^*
\psi_i'$, $\pred{simp}(\psi_i')$ and $\pred{equiv}(\structure, \psi_i,
\psi_i')$. Call this fact~(D).

Note that $\psi = \psi_1 \conj \ldots \conj \psi_n \conj \psi''
\rewrite^* \psi_1' \conj \ldots \conj \psi_n' \conj \top$ (the second
relation follows because $\psi'' \rewrite \top$). Further, because
$\pred{equiv}(\structure, \psi_i, \psi_i')$ (fact~(D)) and
$\pred{equiv}(\structure, \psi'', \top)$ (fact~(C)), it follows that
$\pred{equiv}(\structure, \psi, (\psi_1' \conj \ldots \conj \psi_n'
\conj \top))$. Also, from fact~(C), $\pred{simp}(\psi_1' \conj \ldots
\conj \psi_n' \conj \top)$. The proof is complete by choosing the
$\psi'$ obtained by applying Lemma~\ref{lem:oc:noquant} to $\psi_1'
\conj \ldots \conj \psi_n' \conj \top$.\\

\case ~$\varphi = \exists \vec{x}. (c \conj \varphi')$. Then, $\psi$ is
calculated as follows:
\[
\psi = \reduce(\structure, \varphi) ~=~
\begin{array}[t]{@{}l}
\mbox{let}\\
~\{\sigma_1,\ldots,\sigma_n\} \leftarrow \lift{\sat}(\structure,c)\\
~\{ \vec{t_i} \leftarrow \sigma_i(\vec{x}) \}_{i=1}^n\\
~S \leftarrow \{\vec{t_1},\ldots,\vec{t_n}\}\\
~\{ \psi_i \leftarrow \reduce(\structure,
\varphi'[\vec{t_i}/\vec{x}]) \}_{i=1}^n\\
~\psi'' \leftarrow \exists \vec{x}.((c \conj \vec{x} \not \in S) \conj \varphi')\\
\mbox{return } \\
~~~ \psi_1 \disj \ldots \disj \psi_n \disj \psi''
\end{array} \]

By inversion on the given derivation of $\vdash \varphi$, we know that
there is a $\chi_O$ such that (1)~$\{\} \vdash c: \chi_O$,
(2)~$\vec{x} \subseteq \chi_O$, (3)~$\fv(c) \subseteq \vec{x}$, and
(4)~$\chi_O \vdash \varphi'$. By
Lemma~\ref{lem:mode:condition:inclusion} on (1), $\chi_O \subseteq
\fv(c)$. From this, (2), and (3), it follows that $\vec{x} = \fv(c) =
\chi_O$. Call this fact~(A). Note also that by
Theorem~\ref{thm:sat:total:app}, $\dom(\sigma_i) \supseteq \chi_O =
\vec{x}$. Call this fact~(B).

Next, we show that $\pred{equiv}(\structure, \psi'', \bot)$. Since for
all $\structure'$, $\structure' \models \dual{\bot} = \top$, it
suffices to show that for all $\structure' \geq \structure$,
$\structure' \models \dual{\exists \vec{x}.((c \conj \vec{x} \not \in
  S) \conj \varphi')}$, i.e., $\structure' \models \forall \vec{x}.((c
\conj \vec{x} \not \in S) \imp \varphi')$. By definition of $\models$,
it suffices to prove that for all $\vec{t}$ and $\structure' \geq
\structure$, $\structure' \models \dual{c[\vec{t}/\vec{x}] \conj
  \vec{t} \not \in S}$, i.e., $\structure' \models
\dual{c[\vec{t}/\vec{x}]} \disj \vec{t} \in S$. If $\vec{t} =
\vec{t_i}$ for some $i$, then $\structure' \models \vec{t} \in S$ by
definition of $S$, so we are done. Hence, we need only consider the
case where $\vec{t} \not \in S$. In this case we show that
$\structure' \models \dual{c[\vec{t}/\vec{x}]}$. By
Lemma~\ref{lem:oc:restriction:em}, this is implied by $\structure'
\not \models c[\vec{t}/\vec{x}]$, so we show the latter. Suppose, for
the sake of contradiction, that $\structure' \models
c[\vec{t}/\vec{x}]$. By Lemma~\ref{lem:oc:restriction:extend},
$\structure \models c[\vec{t}/\vec{x}]$. Hence, by
Theorem~\ref{thm:sat:correct:app}, there is a $\sigma \in
\lift{\sat}(\structure, c)$ such that $[\vec{x} \mapsto \vec{t}] \geq
\sigma$.  $\sigma \in \lift{\sat}(\structure, c)$ forces $\sigma =
\sigma_i$ for some $i$ and, by fact~(B), $\vec{t} = \vec{t_i}$. Hence,
$\vec{t} = \vec{t_i} \in S$, a contradiction. Therefore,
$\pred{equiv}(\structure, \psi'', \bot)$. Call this fact~(C).

By Lemma~\ref{lem:subst:mode} on~(4), we derive $\chi_O \backslash
\vec{x} \vdash \varphi'[\vec{t}/\vec{x}]$. Using fact~(A), we have
$\vdash \varphi'[\vec{t}/\vec{x}]$. Applying the i.h.\ to this and
$\psi_i \leftarrow \reduce(\structure, \varphi'[\vec{t_i}/\vec{x}])$,
we know that there is a $\psi_i'$ such that $\psi_i \rewrite^*
\psi_i'$, $\pred{simp}(\psi_i')$ and $\pred{equiv}(\structure, \psi_i,
\psi_i')$. Call this fact~(D).

Note that $\psi = \psi_1 \disj \ldots \disj \psi_n \disj \psi''
\rewrite^* \psi_1' \disj \ldots \disj \psi_n' \disj \bot$ (the second
relation follows because $\psi'' \rewrite \bot$). Further, because
$\pred{equiv}(\structure, \psi_i, \psi_i')$ (fact~(D)) and
$\pred{equiv}(\structure, \psi'', \bot)$ (fact~(C)), it follows that
$\pred{equiv}(\structure, \psi, (\psi_1' \disj \ldots \disj \psi_n'
\disj \bot))$. Also, from fact~(C), $\pred{simp}(\psi_1' \disj \ldots
\disj \psi_n' \disj \bot)$. The proof is complete by choosing the
$\psi'$ obtained by applying Lemma~\ref{lem:oc:noquant} to $\psi_1'
\disj \ldots \disj \psi_n' \disj \bot$.
\end{proof}


Next, we turn to proofs of Theorems~\ref{thm:safety}
and~\ref{thm:cosafety}. Both theorems rely on a central lemma
(Lemma~\ref{lem:past:reduce}). In order to prove the lemma cleanly, we
need a few definitions and some other lemmas. Note that in the rest of
this Appendix we assume that there are no subjective predicates.

\begin{defn}[Protected restrictions]\label{def:protected:restriction}
Let $T$ be a set of time points (possibly non-ground). We define a
subclass ``$T$-protected'' of restrictions $c$ of the sublogic
inductively as follows:
\begin{enumerate}
\item $p_O(t_1,\ldots,t_n,\ttime_0)$ is $T$-protected if $\ttime_0 \in T$
\item $\vec{x} \not \in S$ is $T$-protected
\item $\ttime \not= \ttime'$ is $T$-protected
\item $\pred{in}(\ttime, \ttime', \ttime_0)$ is $T$-protected if
  $\ttime_0 \in T$
\item $\top$ is $T$-protected
\item $\bot$ is $T$-protected
\item $c_1 \conj c_2$ is $T$-protected if both $c_1$ and
  $c_2$ are $T$-protected.
\item $c_1 \disj c_2$ is $T$-protected if both $c_1$ and
  $c_2$ are $T$-protected.
\item $\exists x.c$ is $T$-protected if $c$ is
  $T$-protected.
\end{enumerate}
\end{defn}

\begin{defn}[Protected formulas]\label{def:protected:formula}
Let $T$ be a set of time points (possibly non-ground).  We define a
subclass ``$T$-protected'' of formulas $\varphi$ of the sublogic
inductively as follows:
\begin{enumerate}
\item $p_O(t_1,\ldots,t_n,\ttime_0)$ is $T$-protected if $\ttime_0 \in
  T$
\item $\top$ is $T$-protected
\item $\bot$ is $T$-protected
\item $\varphi_1 \conj \varphi_2$ is $T$-protected if both
  $\varphi_1$ and $\varphi_2$ are $T$-protected
\item $\varphi_1 \disj \varphi_2$ is $T$-protected if both
  $\varphi_1$ and $\varphi_2$ are $T$-protected
\item $\forall \vec{x}.(c \imp \varphi)$ is $T$-protected if
  $c$ is $T$-protected and $\varphi$ is $T$-protected
\item $\forall \ttime.((\pred{in}(\ttime, \ttime', \ttime_0) \conj c)
  \imp \varphi)$ is $T$-protected if $c$ is $T$-protected, $\ttime_0
  \in T$, and $\varphi$ is ($T \cup \{\ttime\}$)-protected
\item $\exists \vec{x}.(c \conj \varphi)$ is $T$-protected if
  $c$ is $T$-protected and $\varphi$ is $T$-protected
\item $\exists \ttime.((\pred{in}(\ttime, \ttime', \ttime_0) \conj c)
  \conj \varphi)$ is $T$-protected if $c$ is $T$-protected, $\ttime_0
  \in T$, and $\varphi$ is ($T \cup \{\ttime\}$)-protected
\end{enumerate}
\end{defn}


\begin{lem}[Excluded middle for protected formulas]\label{lem:protected:em}
Let $T$, $\ttime_0$ be ground. Suppose $\structure$ is
$\ttime_0$-complete and for all $\ttime \in T$, $\ttime \leq
\ttime_0$.  Then, the following hold.
\begin{enumerate}
\item If $c$ is ground and $T$-protected, then either $\structure
  \models c$ or $\structure \models \dual{c}$.
\item If $\varphi$ is ground and $T$-protected, then either
  $\structure \models \varphi$ or $\structure \models \dual{\varphi}$.
\end{enumerate}
\end{lem}
\begin{proof}
Both statements follow by an induction on the respective definitions
of $T$-protected. We show some representative cases below.\\

\noindent \underline{\bf Proof of (1).} \\

\case ~$c = p_O(t_1,\ldots, t_n,\ttime)$ and $\ttime \in T$. By
definition of $\ttime_0$-complete and the fact $\ttime \leq \ttime_0$,
we know that either $\rho_\structure(p_O(t_1,\ldots, t_n,\ttime)) =
\symt$ or $\rho_\structure(p_O(t_1,\ldots, t_n,\ttime)) = \symf$. In
the former case, $\structure \models p_O(t_1,\ldots, t_n,\ttime)$,
while in the latter case, $\structure \models \dual{p_O}(t_1,\ldots,
t_n,\ttime)$.\\

\case ~$c = c_1 \conj c_2$ and both $c_1$ and $c_2$ are
$T$-protected. By the i.h., for each $i$, either $\structure \models
c_i$ or $\structure \models \dual{c_i}$. If $\structure \models c_1$
and $\structure \models c_2$, then $\structure \models c_1 \conj c_2$,
as required. If, on the other hand, for some $i$, $\structure \models
\dual{c_i}$, then $\structure \models \dual{c_1} \disj \dual{c_2}$,
i.e., $\structure \models \dual{c}$.\\

\case ~$c = \exists x.c$ and $c$ is $T$-protected. By the i.h., for
every $t$, either $\structure \models c[t/x]$ or $\structure \models
\dual{c}[t/x]$. If there is a $t$ such that $\structure \models
c[t/x]$, then also $\structure \models \exists x.c$. If, on the other
hand, for every $t$, $\structure \models \dual{c}[t/x]$, then also,
$\structure \models \forall x.\dual{c}$, i.e., $\structure \models
\dual{\exists x.c}$.\\

\noindent \underline{\bf Proof of (2).}\\

\case ~$\varphi = \forall \vec{x}.(c \imp \varphi')$ where $c$ is
$T$-protected and $\varphi'$ is $T$-protected. If for any $\vec{t}$,
$\structure \models c[\vec{t}/\vec{x}]$ and $\structure \models
\dual{\varphi'}[\vec{t}/\vec{x}]$, then, by definition, $\structure
\models \exists \vec{x}. (c \conj \dual{\varphi'})$, i.e., $\structure
\models \dual{\varphi}$ and we are done. Hence, we need only consider
the case where for every $\vec{t}$, either $\structure \not \models
c[\vec{t}/\vec{x}]$ or $\structure \not \models \dual{\varphi'}
[\vec{t}/\vec{x}]$. However, by~(1) and the i.h., we also deduce in
this case that for every $\vec{t}$, either $\structure \models
\dual{c}[\vec{t}/\vec{x}]$ or $\structure \models \varphi'
     [\vec{t}/\vec{x}]$. By definition of $\models$, $\structure
     \models \varphi$ in this case. \\

\case ~$\forall \ttime.((\pred{in}(\ttime, \ttime', \ttime_1) \conj c)
\imp \varphi')$ where $c$ is $T$-protected, $\ttime_1 \in T$, and
$\varphi'$ is ($T \cup \{\ttime\}$)-protected. We consider two
exhaustive subcases:\\

\subcase ~There is a ground $\ttime''$ such that $\structure \models
\pred{in}(\ttime'', \ttime', \ttime_1)$, $\structure \models
c[\ttime''/\ttime]$ and $\structure \models
\dual{\varphi'}[\ttime''/\ttime]$. By definition of $\models$,
$\structure \models \exists \tau. ((\pred{in}(\ttime'', \ttime',
\ttime_1) \conj c) \conj \dual{\varphi'})$, i.e., $\structure \models
\dual{\varphi}$.\\

\subcase ~For every ground $\ttime''$, either $\structure \not \models
\pred{in}(\ttime'',\ttime',\ttime_1)$, or $\structure \not \models
c[\ttime''/\ttime]$, or $\structure \not \models
\dual{\varphi'}[\ttime''/\ttime]$. In this case we show that
$\structure \models \varphi$. Following the definition of $\models$,
pick any $\ttime''$. It suffices to prove that either $\structure
\models \dual{\pred{in}}(\ttime'',\ttime',\ttime_1)$ or $\structure
\models \dual{c}[\ttime''/\ttime]$ or $\structure \models
\varphi'[\ttime''/\ttime]$. From the subcase assumption, $\structure
\not \models \pred{in}(\ttime'',\ttime',\ttime_1)$, or $\structure
\not \models c[\ttime''/\ttime]$, or $\structure \not \models
\dual{\varphi'}[\ttime''/\ttime]$. If $\structure \not \models
\pred{in}(\ttime'',\ttime',\ttime_1)$, then because
$\pred{in}(\ttime'',\ttime',\ttime_1)$ is $T$-protected (note that
$\ttime_1 \in T$), (1)~implies that $\structure \models
\dual{\pred{in}}(\ttime'',\ttime',\ttime_1)$. The case $\structure
\not \models c[\ttime''/\ttime]$ is similar. That leaves only the last
case: $\structure \not \models \dual{\varphi'}[\ttime''/\ttime]$.
Since $\varphi'$ is ($T \cup \{\ttime\}$)-protected,
$\varphi'[\ttime''/\ttime]$ is ($T \cup
\{\ttime''\}$)-protected. Further, because we already considered the
case $\structure \not \models \pred{in}(\ttime'',\ttime',\ttime_1)$,
we may assume here that $\structure \models
\pred{in}(\ttime'',\ttime',\ttime_1)$, which implies $\ttime'' \leq
\ttime_1 \leq \ttime_0$. Thus, we can apply the i.h.\ to
$\varphi'[\ttime''/\ttime]$ to deduce that either $\structure \models
\varphi'[\ttime''/\ttime]$ or $\structure \models
\dual{\varphi'}[\ttime''/\ttime]$. The latter is assumed to be false,
so we must have $\structure \models \varphi'[\ttime''/\ttime]$, as
required.
\end{proof}


\begin{lem}[Reduction of protected formulas]\label{lem:protected:reduction}
Let $T$, $\ttime_0$ be ground. Suppose $\varphi$ is $T$-protected,
$\vdash \varphi$, $\structure$ is $\ttime_0$-complete, and for all
$\ttime \in T$, $\ttime \leq \ttime_0$. Then, $\reduce(\structure,
\varphi) \rewrite^* \psi$, where $\psi = \top$ or $\psi = \bot$ and
$\structure \models \varphi$ iff $\structure \models \psi$.
\end{lem}
\begin{proof}
By induction on the derivation of $\varphi$ being $T$-protected. The
proof is very similar to that of Theorem~\ref{thm:oc:app} and we show
here only some representative cases of the induction.\\

\case ~$\varphi = p_O(t_1,\ldots,t_n,\ttime)$ where $\ttime \in
T$. Because $\structure$ is $\ttime_0$-complete and $\ttime \leq
\ttime_0$, we know that $\rho_\structure(p_O(t_1,\ldots,t_n,\ttime))
\in {\symt, \symf}$. Accordingly, $\reduce(\structure, \varphi) \in
\{\top,\bot\}$, so we can choose $\psi = \reduce(\structure, \varphi)$
to satisfy the theorem's requirements.\\

\case ~$\varphi = \forall \vec{x}.(c \imp \varphi')$ where $c$ and
$\varphi'$ are both $T$-protected. Then, $\reduce(\structure,
\varphi)$ is calculated as follows.
\[
\reduce(\structure, \varphi) ~=~
\begin{array}[t]{@{}l}
\mbox{let}\\
~\{\sigma_1,\ldots,\sigma_n\} \leftarrow \lift{\sat}(\structure,c)\\
~\{ \vec{t_i} \leftarrow \sigma_i(\vec{x}) \}_{i=1}^n\\
~S \leftarrow \{\vec{t_1},\ldots,\vec{t_n}\}\\
~\{ \psi_i \leftarrow \reduce(\structure,
\varphi'[\vec{t_i}/\vec{x}]) \}_{i=1}^n\\
~\psi' \leftarrow \forall \vec{x}.((c \conj \vec{x} \not \in S) \imp \varphi')\\
\mbox{return } \\
~~~ \psi_1 \conj \ldots \conj \psi_n \conj \psi'
\end{array} \]

By inversion on the given derivation of $\vdash \varphi$, we know that
there is a $\chi_O$ such that (1)~$\{\} \vdash c: \chi_O$,
(2)~$\vec{x} \subseteq \chi_O$, (3)~$\fv(c) \subseteq \vec{x}$, and
(4)~$\chi_O \vdash \varphi'$. By
Lemma~\ref{lem:mode:condition:inclusion} on (1), $\chi_O \subseteq
\fv(c)$. From this, (2), and (3), it follows that $\vec{x} = \fv(c) =
\chi_O$. Call this fact~(A). Note also that by
Theorem~\ref{thm:sat:total:app}, $\dom(\sigma_i) \supseteq \chi_O =
\vec{x}$. Call this fact~(B).

Next, we show that $\structure \models \psi'$. Following the
definition of $\models$, it suffices to prove that for all $\vec{t}$,
$\structure \models \dual{c[\vec{t}/\vec{x}] \conj \vec{t} \not \in
  S}$, i.e., either $\structure \models \dual{c}[\vec{t}/\vec{x}]$ or
$\vec{t} \in S$. Suppose $\vec{t} \not \in S$. Then, we show that
$\structure \models \dual{c}[\vec{t}/\vec{x}]$. Because $c$ is
$T$-protected, Lemma~\ref{lem:protected:em}(1) applies, so the last
fact is implied by $\structure \not \models c[\vec{t}/\vec{x}]$. So we
prove this instead. Suppose, for the sake of contradiction, that
$\structure \models c[\vec{t}/\vec{x}]$. Then, by
Theorem~\ref{thm:sat:correct:app}, there is a $\sigma \in
\lift{\sat}(\structure, c)$ such that $[\vec{x} \mapsto \vec{t}] \geq
\sigma$.  $\sigma \in \lift{\sat}(\structure, c)$ forces $\sigma =
\sigma_i$ for some $i$ and, by fact~(B), $\vec{t} = \vec{t_i}$. Hence,
$\vec{t} = \vec{t_i} \in S$, a contradiction. Hence, we must have
$\structure \models \psi'$. Call this fact~(C).

By Lemma~\ref{lem:subst:mode} on~(4), we derive $\chi_O \backslash
\vec{x} \vdash \varphi'[\vec{t}/\vec{x}]$. Using fact~(A), we have
$\vdash \varphi'[\vec{t}/\vec{x}]$. We already know that $\varphi'$ is
$T$-protected and, hence, $\varphi'[\vec{t}/\vec{x}]$ is also
$T$-protected. Applying the i.h.\ to the last two facts, and $\psi_i
\leftarrow \reduce(\structure, \varphi'[\vec{t_i}/\vec{x}])$, we know
that there is a $\psi_i' \in \{\top,\bot\}$ such that $\psi_i
\rewrite^* \psi_i'$ and $\structure \models
\varphi'[\vec{t_i}/\vec{x}]$ iff $\structure \models \psi_i'$. Note
that by Theorem~\ref{thm:correctness:app}, this also implies
$\structure \models \psi_i$ iff $\structure \models \psi_i'$. Call
this fact~(D). We consider two subcases:\\

\subcase ~For every $i$, $\psi_i' = \top$. Clearly, we have
$\reduce(\structure, \varphi) = (\psi_1 \conj \ldots \conj \psi_n
\conj \psi') \rewrite^* \top$ (note: $\psi' \rewrite \top$). We must
show that $\structure \models \varphi$. We have by fact~(D) that
$\structure \models \psi_i$ for each $i$ and by fact~(C) that
$\structure \models \psi'$. Consequently, $\structure \models (\psi_1
\conj \ldots \conj \psi_n \conj \psi')$ and, hence, by
Theorem~\ref{thm:correctness:app}, $\structure \models \varphi$.\\

\subcase ~There is a $i$ such that $\psi_i' = \bot$. Clearly, we have
$\reduce(\structure, \varphi) = (\ldots \conj \psi_i \conj \ldots)
\rewrite^* \bot$. We must show that $\structure \not \models
\varphi$. Note that by fact~(D), $\structure \not\models
\psi_i$. Consequently, by definition of $\models$, $\structure \not
\models \reduce(\structure, \varphi)$ and, hence, by
Theorem~\ref{thm:correctness:app}, $\structure \not \models \varphi$,
as required.\\

\case ~$\varphi = \forall x.((\pred{in}(x, \ttime', \ttime) \conj c)
\imp \varphi')$ where $c$ is $T$-protected, $\ttime \in T$, and
$\varphi'$ is ($T \cup \{x\}$)-protected. Then, $\reduce(\structure,
\varphi)$ is calculated as follows.
\[
\reduce(\structure, \varphi) ~=~
\begin{array}[t]{@{}l}
\mbox{let}\\
~\{\sigma_1,\ldots,\sigma_n\} \leftarrow \lift{\sat}(\structure,(\pred{in}(x,\ttime',\ttime) \conj c))\\
~\{ \ttime_i \leftarrow \sigma_i(x) \}_{i=1}^n\\
~S \leftarrow \{\ttime_1,\ldots,\ttime_n\}\\
~\{ \psi_i \leftarrow \reduce(\structure,
\varphi'[\ttime_i/x]) \}_{i=1}^n\\
~\psi' \leftarrow \forall x.((\pred{in}(x, \ttime', \ttime) \conj c \conj x \not \in S)
\imp \varphi')\\
\mbox{return } \\
~~~ \psi_1 \conj \ldots \conj \psi_n \conj \psi'
\end{array} \]

By inversion on the given derivation of $\vdash \varphi$, we know that
there is a $\chi_O$ such that (1)~$\{\} \vdash
\pred{in}(x,\ttime',\ttime) \conj c: \chi_O$, (2)~$\{x\} \subseteq
\chi_O$, (3)~$\fv(\pred{in}(x,\ttime',\ttime) \conj c) \subseteq
\{x\}$, and (4)~$\chi_O \vdash \varphi'$. By
Lemma~\ref{lem:mode:condition:inclusion} on (1), $\chi_O \subseteq
\fv(\pred{in}(x,\ttime',\ttime) \conj c)$. From this, (2), and (3), it
follows that $\{x\} = \fv(\pred{in}(x,\ttime',\ttime) \conj c) =
\chi_O$. Call this fact~(A). Note also that by
Theorem~\ref{thm:sat:total:app}, $\dom(\sigma_i) \supseteq \chi_O =
\{x\}$. Call this fact~(B).

Next, we show that $\structure \models \psi'$. Following the
definition of $\models$, it suffices to prove that for all $t$,
$\structure \models \dual{\pred{in}(t,\ttime',\ttime) \conj c[t/x]
  \conj t \not \in S}$, i.e., either $\structure \models
\dual{\pred{in}(t,\ttime',\ttime) \conj c[t/x]}$ or $t \in S$. Suppose
$t \not \in S$. Then, we show that $\structure \models
\dual{\pred{in}(t,\ttime',\ttime) \conj c[t/x]}$. Because
$\pred{in}(t,\ttime',\ttime) \conj c[t/x]$ is $T$-protected,
Lemma~\ref{lem:protected:em}(1) applies, so the last fact is implied
by $\structure \not \models \pred{in}(t,\ttime',\ttime) \conj
c[t/x]$. So we prove this instead. Suppose, for the sake of
contradiction, that $\structure \models \pred{in}(t,\ttime',\ttime)
\conj c[t/x]$. Then, by Theorem~\ref{thm:sat:correct:app}, there is a
$\sigma \in \lift{\sat}(\structure, \pred{in}(x,\ttime',\ttime) \conj
c)$ such that $[x \mapsto t] \geq \sigma$.  $\sigma \in
\lift{\sat}(\structure, \pred{in}(x,\ttime',\ttime) \conj c)$ forces
$\sigma = \sigma_i$ for some $i$ and, by fact~(B), $t =
\ttime_i$. Hence, $t = \ttime_i \in S$, a contradiction. Hence, we
must have $\structure \models \psi'$. Call this fact~(C).

By Lemma~\ref{lem:subst:mode} on~(4), we derive $\chi_O \backslash
\{x\} \vdash \varphi'[\ttime_i/x]$. Using fact~(A), we have $\vdash
\varphi'[\ttime_i/x]$. We already know that $\varphi'$ is ($T \cup
\{x\}$)-protected and, hence, $\varphi'[\ttime_i/x]$ is ($T \cup
\{\ttime_i\}$)-protected. Note also that $\ttime_i \leq \ttime \leq
\ttime_0$. Applying the i.h.\ to the last three facts, and $\psi_i
\leftarrow \reduce(\structure, \varphi'[\ttime_i/x])$, we know that
there is a $\psi_i' \in \{\top,\bot\}$ such that $\psi_i \rewrite^*
\psi_i'$ and $\structure \models \varphi'[\ttime_i/x]$ iff $\structure
\models \psi_i'$. Note that by Theorem~\ref{thm:correctness:app}, this
also implies $\structure \models \psi_i$ iff $\structure \models
\psi_i'$. Call this fact~(D). We consider two subcases:\\

\subcase ~For every $i$, $\psi_i' = \top$. Clearly, we have
$\reduce(\structure, \varphi) = (\psi_1 \conj \ldots \conj \psi_n
\conj \psi') \rewrite^* \top$ (note: $\psi' \rewrite \top$). We must
show that $\structure \models \varphi$. We have by fact~(D) that
$\structure \models \psi_i$ for each $i$ and by fact~(C) that
$\structure \models \psi'$. Consequently, $\structure \models (\psi_1
\conj \ldots \conj \psi_n \conj \psi')$ and, hence, by
Theorem~\ref{thm:correctness:app}, $\structure \models \varphi$.\\

\subcase ~There is a $i$ such that $\psi_i' = \bot$. Clearly, we have
$\reduce(\structure, \varphi) = (\ldots \conj \psi_i \conj \ldots)
\rewrite^* \bot$. We must show that $\structure \not \models
\varphi$. Note that by fact~(D), $\structure \not\models
\psi_i$. Consequently, by definition of $\models$, $\structure \not
\models \reduce(\structure, \varphi)$ and, hence, by
Theorem~\ref{thm:correctness:app}, $\structure \not \models \varphi$,
as required.
\end{proof}


\begin{lem}[Duality of protection]\label{lem:protected:dual}
$\varphi$ is $T$-protected iff $\dual{\varphi}$ is $T$-protected.
\end{lem}
\begin{proof}
By a straightforward induction on $\varphi$.
\end{proof}


\begin{lem}[Past translation]\label{lem:past:translation}
The following hold:
\begin{enumerate}
\item If $c$ is a restriction in the temporal logic, then for any
  $\ttime \in T$, $\trans{\ttime}{c}$ is $T$-protected.
\item If $\alpha_p$ is a temporal logic formula without future
  operators, then for any $\ttime \in T$, $\trans{\ttime}{\alpha_p}$
  is $T$-protected.
\end{enumerate}
\end{lem}
\begin{proof}
(1)~follows by a straightforward induction on $c$. Then,~(2) follows
  by induction on $\alpha_p$. The case $\alpha_p =
  p_S(t_1,\ldots,t_n)$ does not arise because we assume that there are
  no subjective predicates. Similarly, the cases $\alpha_p = \boxp
  \beta$ and $\alpha_p = \beta_1 \until \beta_2$ do not arise because
  $\alpha_p$ does not contain future operators. We show some other
  representative cases below.\\

\case ~$\alpha_p = p_O(t_1,\ldots,t_n)$. Then,
$\trans{\ttime}{\alpha_p} = p_O(t_1,\ldots,t_n,\ttime)$, which is
$T$-protected because $\ttime \in T$ is given.\\

\case ~$\alpha_p = \neg \alpha_p'$. Then, $\trans{\ttime}{\alpha_p} =
\dual{\trans{\ttime}{\alpha_p'}}$. By the i.h.,
$\trans{\ttime}{\alpha_p'}$ is $T$-protected. Hence, by
Lemma~\ref{lem:protected:dual}, $\dual{\trans{\ttime}{\alpha_p'}}$ is
also $T$-protected.\\

\case ~$\alpha_p = \forall \vec{x}.(c \imp \beta_p)$. Then,
$\trans{\ttime}{\alpha_p} = \forall \vec{x}. (\trans{\ttime}{c} \imp
\trans{\ttime}{\beta_p})$. By statement~(1) of the theorem,
$\trans{\ttime}{c}$ is $T$-protected, and by the i.h.,
$\trans{\ttime}{\beta_p}$ is $T$-protected. Hence,
$\trans{\ttime}{\alpha_p}$ is $T$-protected by clause~(6) of
Defn~\ref{def:protected:formula}.\\

\case ~$\alpha_p = \here x. \beta_p$. Then, $\trans{\ttime}{\alpha_p}
= \trans{\ttime}{\beta_p[\ttime/x]}$. By the i.h.\ on the smaller
formula $\beta_p[\ttime/x]$, we get that
$\trans{\ttime}{\beta_p[\ttime/x]}$ is $T$-protected.\\

\case ~$\alpha_p = \beta_1 \since \beta_2$. Then,
$\trans{\ttime}{\alpha_p} = \exists \ttime'. (\pred{in}(\ttime', 0,
\ttime) \conj \trans{\ttime'}{\beta_2} \conj (\forall \ttime''. (({\tt
  in}(\ttime'',\ttime',\ttime) \conj \ttime' \not= \ttime'') \imp
\trans{\ttime''}{\beta_1})))$. First, by the i.h.,
$\trans{\ttime''}{\beta_1}$ is ($T \cup
\{\ttime''\}$)-protected. Consequently, by clause~(7) of
Defn~\ref{def:protected:formula}, $(\forall \ttime''. (({\tt
  in}(\ttime'',\ttime',\ttime) \conj \ttime' \not= \ttime'') \imp
\trans{\ttime''}{\beta_1}))$ is $T$-protected. Hence, it is also ($T
\cup \{\tau'\}$)-protected. Call this fact~(A).  Next, by the i.h.,
$\trans{\tau'}{\beta_2}$ is ($T \cup \{\tau'\}$)-protected. Combining
this and fact~(A), we have that $\trans{\ttime'}{\beta_2} \conj
(\forall \ttime''. (({\tt in}(\ttime'',\ttime',\ttime) \conj \ttime'
\not= \ttime'') \imp \trans{\ttime''}{\beta_1}))$ is ($T \cup
\{\tau'\}$)-protected. By clause~(9) of
Defn~\ref{def:protected:formula}, $\trans{\ttime}{\alpha_p}$ is
$T$-protected, as required.\\

\case ~$\alpha_p = \boxm \beta_p$. Then, $\trans{\ttime}{\alpha_p} =
\forall \ttime'. (\pred{in}(\ttime', \ttime,\infty) \imp
\trans{\ttime'}{\beta_p})$. By the i.h., $\trans{\ttime'}{\beta_p}$ is
($T \cup \{\ttime'\}$)-protected. Hence, by clause~(7) of
Defn~\ref{def:protected:formula}, $\trans{\ttime}{\alpha_p}$ is
$T$-protected.
\end{proof}


\begin{lem}[Reduction of past formulas]\label{lem:past:reduce}
Let $\alpha_p$ be a temporal logic formula without future operators,
and suppose that $\ttime$ is a ground time point such that $\vdash
\trans{\ttime}{\alpha_p}$. Let $\structure$ be $\ttime_0$-complete and
$\ttime_0 \geq \ttime$. Then, either (1)~$\reduce(\structure,
\trans{\ttime}{\alpha_p}) \rewrite^* \top$ and $\structure \models
\trans{\ttime}{\alpha_p}$, or (2)~$\reduce(\structure,
\trans{\ttime}{\alpha_p}) \rewrite^* \bot$ and $\structure \models
\dual{\trans{\ttime}{\alpha_p}}$.
\end{lem}
\begin{proof}
By Lemma~\ref{lem:past:translation}(2), $\trans{\ttime}{\alpha_p}$ is
$\{\ttime\}$-protected. Because $\ttime \leq \ttime_0$ and $\vdash
\trans{\ttime}{\alpha_p}$, by Lemma~\ref{lem:protected:reduction},
$\reduce(\structure, \trans{\ttime}{\alpha_p}) \rewrite^* \psi$, where
$\psi = \top$ or $\psi = \bot$ and $\structure \models
\trans{\ttime}{\alpha_p}$ iff $\structure \models \psi$. Call the
latter fact~(A). We consider two cases:\\

\case ~$\psi = \top$. In this case, fact~(A) means that $\structure
\models \trans{\ttime}{\alpha_p}$ iff $\structure \models \top$, which
implies that $\structure \models \trans{\ttime}{\alpha_p}$. So~(1)
holds.\\

\case ~$\psi = \bot$. In this case, fact~(A) yields that $\structure
\not \models \trans{\ttime}{\alpha_p}$. Since
$\trans{\ttime}{\alpha_p}$ is $\{\ttime\}$-protected (already proved) and
$\ttime \leq \ttime_0$, Lemma~\ref{lem:protected:em}(2) yields
$\structure \models \dual{\trans{\ttime}{\alpha_p}}$. So~(2) holds.
\end{proof}


\begin{thm}[Enforcement of safety properties; Theorem~\ref{thm:safety}]
  \label{thm:safety:app}
Suppose ${\globally \alpha_p}$ is a safety property, $\vdash
{\globally \alpha_p}$, $\structure$ is $\ttime_0$-complete, and for
all $\ttime$, $(\rho_\structure(\pred{in}(\ttime,0,\infty)) = \symt)
\Rightarrow \ttime \leq \ttime_0$. Then,\linebreak[6]
$\reduce(\structure, \globally \alpha_p) \rewrite^* \bot$ iff there is
a $\ttime$ such that $\structure \models \pred{in}(\ttime,0,\ttime_0)$
and $\structure \models \dual{\trans{\ttime}{\alpha_p}}$.
\end{thm}
\begin{proof}
We have $\globally \alpha_p = \forall
\ttime.(\pred{in}(\ttime,0,\infty) \imp \trans{\ttime}{\alpha_p})$. Let $\reduce(\structure, {\globally \alpha_p}) = \psi$. Then, 
\[
\psi = \reduce(\structure, \forall \ttime.(\pred{in}(\ttime,0,\infty)
\imp \trans{\ttime}{\alpha_p})) ~ = ~
\begin{array}[t]{@{}l}
\mbox{let}\\
~\{\sigma_1,\ldots,\sigma_n\} \leftarrow \lift{\sat}(\structure, \pred{in}(\ttime,0,\infty))\\
~\{ \ttime_i \leftarrow \sigma_i(\ttime) \}_{i=1}^n\\
~S \leftarrow \{\ttime_1,\ldots,\ttime_n\}\\
~\{ \psi_i \leftarrow \reduce(\structure,
\trans{\ttime_i}{\alpha_p}) \}_{i=1}^n\\
~\psi' \leftarrow \forall \ttime.((\pred{in}(\ttime,0,\infty) \conj \ttime \not \in S) \imp \trans{\ttime}{\alpha_p})\\
\mbox{return } \\
~~~ \psi_1 \conj \ldots \conj \psi_n \conj \psi'
\end{array} 
\]

By inversion on $\vdash {\globally \alpha_p}$, we obtain a $\chi_O$
such that $\vdash \pred{in}(\ttime,0,\infty): \chi_O$ and $\chi_O
\vdash \trans{\ttime}{\alpha_p}$. The first of these forces $\chi_O =
\{\ttime\}$, so from the second one we have that $\ttime \vdash
\trans{\ttime}{\alpha_p}$. Using Lemma~\ref{lem:subst:mode}(2), we get
$\vdash \trans{\ttime_i}{\alpha_p}$. Call this fact~(A). Next, observe
that by Theorem~\ref{thm:sat:correct:app}, for each $\ttime_i$,
$\structure \models \pred{in}(\ttime_i, 0, \infty)$, i.e.,
$\rho_{\structure}(\pred{in}(\ttime_i, 0, \infty)) = \symt$. This
forces $\ttime_i \leq \ttime_0$ from the assumptions of the theorem we
are trying to prove. Call this fact~(B). We now prove the two
directions of the conclusion of the theorem.\\

\noindent \textbf{Direction ``if''.} Suppose there is a $\ttime$ with
$\structure \models \pred{in}(\ttime,0,\ttime_0)$ and $\structure
\models \dual{\trans{\ttime}{\alpha_p}}$. We prove that $\psi
\rewrite^* \bot$. By Theorem~\ref{thm:sat:correct:app} applied to
$\structure \models \pred{in}(\ttime,0,\ttime_0)$, $\ttime = \ttime_i$
for some $i$. Hence by Lemma~\ref{lem:past:reduce}, using facts~(A)
and~(B) and $\structure \models \dual{\trans{\ttime}{\alpha_p}}$, we
have that $\reduce(\structure, \trans{\ttime_i}{\alpha_p}) \rewrite^*
\bot$, i.e., $\psi_i \rewrite^* \bot$. Clearly, $\psi = (\ldots \conj
\psi_i \conj \ldots) \rewrite^* \bot$, as required. \\

\noindent \textbf{Direction ``only if''.} Suppose that
$\reduce(\structure, {\globally \alpha_p}) \rewrite^* \bot$, i.e.,
$\psi \rewrite^* \bot$. We show that there is a $\ttime$ such that
$\pred{in}(\ttime,0,\ttime_0)$ and $\structure \models
\dual{\trans{\ttime}{\alpha_p}}$. By definition of $\rewrite$, we
obtain that either for some $i$, $\psi_i \rewrite^* \bot$ or $\psi'
\rewrite^* \bot$. The latter is impossible because $\psi'$ has a
top-level $\forall$, which can only be rewritten to $\top$. Hence,
there is an $i$ such that $\psi_i \rewrite^* \bot$, i.e.,
$\reduce(\structure, \trans{\ttime_i}{\alpha_p}) \rewrite^* \bot$.
Choose $\ttime = \ttime_i$. By Lemma~\ref{lem:past:reduce}, using
facts~(A) and~(B) and $\reduce(\structure, \trans{\ttime_i}{\alpha_p})
\rewrite^* \bot$, we obtain that $\structure \models
\dual{\trans{\ttime_i}{\alpha_p}}$. The remaining requirement,
$\structure \models \pred{in}(\ttime_i,0,\ttime_0)$ follows from
fact~(B).
\end{proof}


\begin{thm}[Enforcement of co-safety properties; Theorem~\ref{thm:cosafety}]
  \label{thm:cosafety:app}
Suppose ${\eventually \alpha_p}$ is a co-safety property, $\vdash
{\eventually \alpha_p}$, $\structure$ is $\ttime_0$-complete, and for
all $\ttime$, $(\rho_\structure(\pred{in}(\ttime,0,\infty)) = \symt)
\Rightarrow \ttime \leq \ttime_0$. Then, $\reduce(\structure,
\eventually \alpha_p) \rewrite^* \top$ if and only if there is a
$\ttime$ such that $\structure \models \pred{in}(\ttime,0,\ttime_0)$
and $\structure \models \trans{\ttime}{\alpha_p}$.
\end{thm}
\begin{proof}
Similar to that of Theorem~\ref{thm:safety:app}.
\end{proof}

\onecolumn
\section{HIPAA Case study}
\label{app:hipaa}

This appendix lists the number of subjective and objective atoms in
each transmission-related clause in the HIPAA Privacy Rule. \#S
denotes the number of subjective atoms; \#O' denotes the number of
such subjective atoms that can be mechanized by a small amount of
design effort; and \#O denotes the number of objective atoms. The
table is sorted by the last column (\#O' + \#O) / (\#S + \#O).

\begin{longtable}{|l|l|l|l|c|}
\hline
Clause No. & \#S & \#O' & \#O & (\#O' + \#O) / (\#S + \#O) \\
\hline 
164.502(e)(1)(ii)(B) & 0 & 0 & 5 & 1.00 \\
164.502(a)(1)(i) & 1 & 1 & 3 & 1.00 \\
164.502(a)(1)(iv) & 37 & 37 & 4 & 1.00 \\
164.502(d)(1) & 2 & 2 & 2 & 1.00 \\
164.502(e)(1)(i) & 1 & 1 & 2 & 1.00 \\
164.508(a)(2) & 37 & 37 & 4 & 1.00 \\
164.508(a)(3)(i) & 38 & 38 & 4 & 1.00 \\
164.508(a)(3)(i)(A) & 2 & 2 & 3 & 1.00 \\
164.510(a)(1)(ii) & 2 & 2 & 3 & 1.00 \\
164.510(a)(2) & 2 & 2 & 2 & 1.00 \\
164.512(c)(2) & 1 & 1 & 0 & 1.00 \\
164.512(e)(1)(i) & 3 & 3 & 4 & 1.00 \\
164.512(e)(1)(ii) & 9 & 9 & 4 & 1.00 \\
164.512(e)(1)(vi) & 4 & 4 & 2 & 1.00 \\
164.512(f)(2) & 10 & 10 & 3 & 1.00 \\
164.512(f)(3)(i) & 6 & 6 & 4 & 1.00 \\
164.514(e)(1) & 25 & 25 & 1 & 1.00 \\
164.512(j)(3) & 11 & 10 & 1 & 0.92 \\
164.524(b)(2)(i) & 54 & 43 & 41 & 0.88 \\
164.524(b)(2)(ii) & 53 & 42 & 42 & 0.88 \\
164.512(g)(1) & 4 & 3 & 4 & 0.88 \\
164.510(b)(1)(i) & 2 & 1 & 5 & 0.86 \\
164.502(e)(1)(ii)(C) & 3 & 2 & 3 & 0.83 \\
164.506(c)(5) & 8 & 6 & 4 & 0.83 \\
164.512(b)(1)(v) & 5 & 3 & 7 & 0.83 \\
164.512(k)(1)(iii) & 3 & 2 & 3 & 0.83 \\
164.514(f)(1) & 3 & 2 & 3 & 0.83 \\
164.502(g)(3)(ii)(A) & 2 & 1 & 4 & 0.83 \\
164.502(g)(3)(ii)(B) & 2 & 1 & 4 & 0.83 \\
164.502(j)(2) & 2 & 1 & 4 & 0.83 \\
164.512(b)(1)(ii) & 3 & 2 & 3 & 0.83 \\
164.512(f)(5) & 4 & 3 & 2 & 0.83 \\
164.512(k)(1)(i) & 2 & 1 & 4 & 0.83 \\
164.512(k)(1)(iv) & 2 & 1 & 4 & 0.83 \\
164.512(k)(6)(i) & 3 & 2 & 3 & 0.83 \\
164.512(k)(6)(ii) & 7 & 5 & 4 & 0.82 \\
164.512(i)(1) & 20 & 15 & 6 & 0.81 \\
164.506(c)(3) & 2 & 1 & 3 & 0.80 \\
164.512(b)(1)(iii) & 3 & 2 & 2 & 0.80 \\
164.512(h) & 2 & 1 & 3 & 0.80 \\
164.512(k)(1)(ii) & 4 & 3 & 1 & 0.80 \\
164.512(g)(2) & 4 & 2 & 5 & 0.78 \\
164.512(d)(1) & 6 & 4 & 3 & 0.78 \\
164.502(a)(2)(ii) & 2 & 1 & 2 & 0.75 \\
164.506(c)(2) & 2 & 1 & 2 & 0.75 \\
164.510(b)(1)(ii) & 4 & 2 & 4 & 0.75 \\ 
164.512(b)(1)(iv) & 3 & 2 & 1 & 0.75 \\
164.510(b)(2) & 5 & 3 & 3 & 0.75 \\
164.512(f)(1)(i) & 17 & 10 & 10 & 0.74 \\
164.506(c)(4) & 6 & 3 & 4 & 0.70 \\
164.502(e)(1)(ii)(A)  & 1 & 0 & 2 & 0.67 \\
164.506(b)(1) & 1 & 0 & 2 & 0.67 \\
164.506(c)(1) & 4 & 2 & 2 & 0.67 \\
164.512(f)(6)(i) & 4 & 2 & 2 & 0.67 \\
164.502(b)(1) & 2 & 1 & 1 & 0.67 \\
164.502(j)(1) & 5 & 1 & 7 & 0.67 \\
164.512(a)(1) & 2 & 1 & 1 & 0.67 \\ 
164.512(f)(1)(ii) & 7 & 4 & 2 & 0.67 \\
164.512(f)(4) & 3 & 1 & 3 & 0.67 \\
164.512(j)(1)(ii)(A) & 18 & 11 & 3 & 0.67 \\
164.512(l) & 2 & 1 & 1 & 0.67 \\
164.512(k)(4) & 4 & 1 & 3 & 0.57 \\
164.512(k)(3) & 5 & 2 & 2 & 0.57 \\
164.512(b)(1)(i) & 6 & 1 & 5 & 0.55 \\
164.502(b)(2)(i) & 1 & 0 & 1 & 0.50 \\
164.508(a)(2)(i)(B) & 1 & 0 & 1 & 0.50 \\
164.508(a)(2)(i)(C) & 1 & 0 & 1 & 0.50 \\
164.508(a)(3)(i)(B) & 1 & 0 & 1 & 0.50 \\ 
164.510(a)(3)(ii) & 3 & 1 & 1 & 0.50 \\
164.512(j)(1)(ii)(B) & 4 & 1 & 2 & 0.50 \\
164.512(k)(5)(i) & 8 & 2 & 3 & 0.45 \\
164.512(k)(2) & 4 & 1 & 1 & 0.40 \\
164.510(b)(4) & 12 & 3 & 2 & 0.36 \\
164.512(c)(1) & 10 & 1 & 4 & 0.36 \\
164.512(f)(3)(ii) & 9 & 1 & 3 & 0.33 \\
164.512(j)(1)(i) & 5 & 1 & 1 & 0.33 \\
164.514(g) & 9 & 1 & 2 & 0.27 \\
164.510(b)(3) & 4 & 1 & 0 & 0.25 \\
164.502(a)(1)(iii) & 1 & 0 & 0 & 0.00 \\
164.510(a)(3)(i) & 4 & 0 & 0 & 0.00 \\
164.512(c)(2)(i) & 1 & 0 & 0 & 0.00 \\
164.512(f)(6)(ii) & 1 & 0 & 0 & 0.00 \\
164.512(j)(2)(i) & 1 & 0 & 0 & 0.00 \\
164.512(j)(2)(ii) & 1 & 0 & 0 & 0.00 \\
\hline
Total & 578 & 402 & 303 & 0.80 \\
\hline
Clause No. & \#S & \#S' & \#O & (\#S' + \#O) / (\#S + \#O) \\
\hline
\end{longtable}

\end{document}